\documentclass[sigplan,screen,
10pt]{acmart}
\usepackage{mypackages,mymacros,stmaryrd}

\setcopyright{acmcopyright}
\copyrightyear{2018}
\acmYear{2018}
\acmDOI{10.1145/1122445.1122456}

\acmConference[Haifa'22]{37th Annual ACM/IEEE Symposium on Logic in Computer Science (LICS)}{2–5 August 2022}{Haifa, Israel}
\acmBooktitle{Haifa'22: ACM/IEEE Symposium on Logic in Computer Science,2–5 August 2022, Haifa, Israel}
\acmPrice{15.00}
\acmISBN{978-1-4503-XXXX-X/18/06}

\begin{document}

\title{Resource approximation for the $\lamu$-calculus}


\author{Davide Barbarossa}
\affiliation{
  \institution{Laboratoire d'Informatique Paris-Nord, Universit\'e Sorbonne Paris-Nord}
  \streetaddress{99, Avenue Jean-Baptiste Clément}
  \city{Villetaneuse}
  \state{France}
  \country{France}
  \postcode{93430}
}
\affiliation{
  \institution{Dipartimento di Informatica - Scienza ed Ingegneria, Universit\`a di Bologna}
  \streetaddress{Mura Anteo Zamboni, 7}
  \city{Bologna}
  \state{Italy}
  \country{Italy}
  \postcode{40127}
}
\email{davide.barbarossa@unibo.it}


\begin{abstract}
The $\lamu$-calculus plays a central role in the theory of programming languages as it extends the Curry-Howard correspondence to classical logic. A major drawback is that it does not satisfy B\"ohm's Theorem and it lacks the corresponding notion of approximation.
On the contrary, we show that Ehrhard and Regnier's Taylor expansion can be easily adapted, thus providing a resource conscious approximation theory.
This produces a sensible $\lamu$-theory with which we prove some advanced properties of the $\lamu$-calculus, such as Stability and Perpendicular Lines Property, from which the impossibility of parallel computations follows.
\end{abstract}




\maketitle

\section{Introduction}
\subsection{Motivation}

\subsubsection{Curry-Howard correspondence for classical logic}
The celebrated \emph{Curry-Howard correspondence} states that a class of programs, written in a suitable programming language, and intuitionistic logic proofs, written in an suitable formal system, are the same mathematical objects. 
The typical suitable programming language is $\lam$-calculus, and the typical suitable formal system is intuitionistic natural deduction NJ; under this correspondence, the simply typed $\lam$-calculus is identified with NJ. 
A natural question is what happens for \emph{classical} logic proofs and whether it is possible to find such a correspondence at all. In the 90's, several ways for generalising such a correspondence to this framework appeared, starting from~\cite{DBLP:conf/popl/Griffin90} where Griffin suggests to type control operators (such as Scheme's $\prog{callcc}$ or Felleisen's $\prog{C}$ operator) with Peirce's law. 
One of the most notable ones is the $\lamu$-calculus, introduced by Parigot in~\cite{DBLP:conf/lpar/Parigot92}, which has the advantage of allowing the correspondence to take the exact same form as in the intuitionistic case: just like $\lam$-calculus is the Turing-complete programming language in which intuitionistic logic expresses its computational content, $\lamu$-calculus is the one expressing the computational content of classical logic. From the programming viewpoint, the big difference between $\lam$-calculus and $\lamu$-calculus is that the former is a \emph{purely functional} language, while the latter is \emph{impure}, due to the possibility of encoding control operators in it --- like the already mentioned $\prog{callcc}$ or $\prog{C}$. From the point of view of, e.g., Classical realizability \cite{krivine:hal-00154500}, this corresponds to the backtracking mechanism related to classical reasoning.

\subsubsection{Taylor expanding programs}
Just before the 90's another major discovery in logic and computer science appeared: Girard's linear logic \cite{GIRARD19871}. This opened a whole new field of research, in which the common line is the deep role reserved to \emph{resources} in a computation/proof. Linear logic allowed Ehrhard and Regnier to discover an astonishing correspondence between linearity in analysis and linearity in computer science, that is formalized in the \emph{differential $\lam$-calculus} \cite{DBLP:journals/tcs/EhrhardR03} (and differential interaction nets \cite{DBLP:journals/tcs/EhrhardR06}). It is possible to Taylor expand programs/proofs by -- as in analysis -- an infinite series of approximants weighted via the usual factorial coefficients. This is usually called the ``full'' or ``quantitative'' Taylor expansion. However, it turns out that even if we do not consider the coefficients, we still obtain a meaningful theory of program approximation: it is usually called the ``qualitative'' Taylor expansion, and will play a central role in the present article. 
Under the assumption of having an idempotent sum, the Taylor expansion is no longer a series but becomes a set, and the approximants can be written in a simple ``target language'', called \emph{resource calculus} (very similar to Boudol's calculus with multiplicities \cite{DBLP:conf/concur/Boudol93}).
In \cite{DBLP:journals/pacmpl/BarbarossaM20}, it is shown that all the fundamental results in the so called \emph{approximation theory} of $\lam$-calculus (Monotonicity, Genericity, Continuity, Stability, Perpendicular Lines Property), usually achieved via labelled reductions and B\"ohm trees~\cite[Chapter~14]{DBLP:books/daglib/0067558}, can be actually proven --~in an arguably more satisfactory way~-- via (qualitative) Taylor expansion; in other words, in $\lam$-calculus, resource approximation (which is at the basis of Taylor expansion) ``subsumes'' B\"ohm trees approximation, and answers positively to a proposal expressed in \cite{DBLP:journals/tcs/EhrhardR03} where Ehrhard and Regnier mention that:
\textit{
``Understanding the relation between the term and its full Taylor expansion might be the starting point of a renewing of the theory of approximation.''}

\subsubsection{The content of this paper}
The aforementioned ``approximation results'' are in fact at the basis of a \emph{mathematical study} of $\lam$-calculus, deep from the conceptual, mathematical and computational viewpoint.
A natural question is then: what about for other programming languages? There are many works which extend the notion of (qualitative as well as quantitative) Taylor expansion to other programming languages (\cite{DBLP:journals/entcs/Chouquet19,DBLP:journals/lmcs/Vaux19,DBLP:conf/rta/LagoL19,DBLP:conf/csl/ChouquetT20,DBLP:journals/corr/abs-1809-02659}), usually concentrating on its relations with normalisation. Not always easy is, then, \emph{applying} it to actually study the properties of the source language. This is maybe due to the fact that, unlike in the long time studied $\lam$-calculus, for other languages one does not really know what a ``mathematical theory of it'' should look like.
In the present work, we tackle the case of $\lamu$-calculus for which, to the best of our knowledge, the problem of directly defining a Taylor expansion has never been priorly considered.
In this sense, our work can be seen as a continuation of the above mentioned series of papers, and may be related to~\cite{DBLP:conf/rta/KesnerV17} where the authors study non-idempotent intersection and union types for $\lamu$-calculus.

We propose here to reverse the above proposal of Ehrhard and Regnier and \emph{start} by defining a resource sensitive version and a Taylor expansion for $\lamu$-calculus, trying then to \emph{use} this approximation machinery to prove -- following \cite{DBLP:journals/pacmpl/BarbarossaM20} -- mathematical properties of the language.

A notable work has to be mentioned: in \cite{DBLP:journals/tcs/Vaux07}, Vaux-Auclair defines a full differential $\lamu$-calculus following the steps of \cite{DBLP:journals/tcs/EhrhardR03}. Its version takes coefficients into account, and as always this raises a series of non-trivial problems to handle. However, he does not define a Taylor expansion nor does he apply those tools to find properties of the language. The present work can thus be seen also as a continuation of Vaux-Auclair's one.

There are, from our point of view, several reasons for considering the $\lamu$-calculus: first of all, from the Curry-Howard point of view, it is the natural ``successor'' of $\lam$-calculus. Moreover, it is a standard reference for the study of control operators in functional languages. 
Yet, there are just few attempts to really study its mathematical theory, and the state of the art is not comparable with the well-established one for $\lam$-calculus. For example, Laurent in \cite{Laurent04anote} makes the following observation:
\textit{
``Models of the simply typed $\lam$-calculus, of the untyped $\lam$-calculus and of the simply typed $\lamu$-calculus are well understood, but what about models of the untyped $\lamu$-calculus? As far as we know, this question has been almost ignored.''}
With the same motivation, we look at the other major part which constitutes a mathematical theory of a programming language, namely the \emph{theory of approximation}. In this sense, the present work can be seen as a continuation of \cite{Laurent04anote}.

Other points in relation with Krivine's classical realizability, proof-nets, CPS-translations and Saurin's $\Lam\mu$-calculus will be mentioned in the conclusions.

The article is organised as follows:
in Section 2 we define the resource $\lamu$-calculus and prove that it is strongly normalising and confluent (Corollaries~\ref{lamu-cor:SN} and~\ref{lamu-cor:confluence}). In Section 3 we define the qualitative Taylor expansion and prove its main properties, which give rise to a non-trivial sensible ``$\lamu$-theory'' (Corollary~\ref{lamu-cor:sensible!}). In Section 4 we apply these approximation tools to prove two important results: Stability (Theorem \ref{lamu-th:TeStability}) and Perpendicular Lines Property (Theorem \ref{lamu-thm:PLL}). As a consequence, we obtain the non-representability of \emph{parallel-or} in $\lamu$-calculus (Corollaries~\ref{lamu-cor:Por} and \ref{lamu-cor:Por'}).

\subsection{Quick overview of $\lamu$-calculus}

We briefly recall the definition of the $\lamu$-calculus, and introduce some basic notions and notation.

\begin{Definition}
Fix a countable set whose elements are called \emph{variables} and a disjoint countable set whose elements are called \emph{names}. The set $\lamu$ of \emph{$\lamu$-terms} is generated by the following grammar:
\[
	M::= x \ \mid \ \lam x.M \ \mid \ MM \ \mid \ \mu\alpha.\name{\beta}{M}
\]
(for $x$ a variable and $\alpha,\beta$ names)
in which, as usual, $\lam$ binds $x$ in $M$ as well as $\mu$ binds $\alpha$ in $\name{\beta}{M}$. Terms are considered up to renaming of bound variables and names.
\end{Definition}

Despite not being actual subterms, words of shape $\name{\alpha}{M}$ are called \emph{named terms}\footnote{
Historically named terms are written as $[\alpha]M$, as in \cite{DBLP:conf/lpar/Parigot92}. But this notation has to be given up since the use of square brackets is already imperatively taken by the \emph{finite multisets}, which we will encounter constantly in the following. Another notation, used in \cite{SAURIN2012106}, is to write $M\alpha$. However in our framework we find this notation not clear. The notation $\name{\alpha}{M}$ should, instead, clearly show what is inside a ``naming'' and what is not.
}. $M$ is said to be \emph{named under} $\alpha$.

The \emph{$k$-contexts} (also called \emph{multihole-contexts} when $k$ is generic) $C=C\set{\xi_1,\dots,\xi_k}$ are defined as expected by:
\[
C::= x \ \mid \ \xi_i \ \mid \ \lam x.C \ \mid \ CC \ \mid \ \mu\alpha.\name{\beta}{C}
\]
where $\set{\xi_1,\dots,\xi_k}$ is a new set whose elements are called \emph{holes}. $1$-contexts are simply called \emph{contexts}.
A context with exactly one occurrence of the hole is called \emph{single-hole}, and as usual satisfy: $C::= \xi_i \ \mid \ \lam x.C \ \mid \ CM \ \mid \ MC \ \mid \ \mu\alpha.\name{\beta}{C}$.

\begin{Definition}\label{def:reductionLamu}
The reduction relation $\to$ of $\lamu$-calculus is the contextual closure\footnote{The \emph{contextual closure} of a binary relation $\mathcal{R}$ is the binary relation given by the set: $\set{(C\hole{M},C\hole{N}) \st M\mathcal{R} N\mathrm{ \ and \ }C\hole{.}\mathrm{ \ single \ hole \ context}}$.} of the union $\to_{\mathrm{base}}$ of:
\[
(\lam x.M)N\rightarrow_\lam M\set{N/x}
\]
\[
(\mu\alpha.\name{\beta}{M})N\rightarrow_\mu \mu\alpha.(\name{\beta}{M})_\alpha N
\]
\[
\mu\gamma.\name{\alpha}{\mu\beta.\name{\eta}{M}}\rightarrow_\rho \mu\gamma.(\name{\eta}{M}\set{\alpha/\beta})
\]
where $M\set{N/x}$ is the usual capture-free substitution of $N$ for all free occurrences of $x$ in $M$, $\name{\eta}{M}\set{\alpha/\beta}$ replaces $\alpha$ for all the free occurrences of $\beta$ in $\name{\eta}{M}$, and $\namedapp{M}{\alpha}{N}$ is given by:
\begin{itemize}
	\item[] $\namedapp{x}{\alpha}{N}:=x$
	\item[] $\namedapp{\lam x.M}{\alpha}{N}:=\lam x.\namedapp{M}{\alpha}{N}$
	\item[] $\namedapp{MP}{\alpha}{N}:=(\namedapp{M}{\alpha}{N})(\namedapp{P}{\alpha}{N})$
	\item[] $\namedapp{\mu\beta.\name{\gamma}{M}}{\alpha}{N}:=\mu\beta.\name{\gamma}{\namedapp{M}{\alpha}{N}}$ (if $\gamma\neq\alpha$)
	\item[] $\namedapp{\mu\beta.\name{\alpha}{M}}{\alpha}{N}:=\mu\beta.\name{\alpha}{(\namedapp{M}{\alpha}{N})N}$.
\end{itemize}
We denote by $=_{\lamu\rho}$ the equivalence induced by $\to$ on $\lamu$.
\end{Definition}

The operation $\namedapp{M}{\alpha}{N}$ coincides with the substitution $M\left\{\name{\alpha}{(\cdot)N}/{\name{\alpha}{\cdot}}\right\}$: every named subterm $\name{\alpha}{\cdot}$ of $M$ gets substituted with the named term $\name{\alpha}{(\cdot)N}$. Nevertheless, ``morally'', it is an application: each subterm of $M$ named under $\alpha$ receives a copy of $N$ to be applied to. This is why we chose the notation ``$\namedapp{M}{\alpha}{N}$", which is reminiscent of the application of $M$ to $N$, and the term $\namedapp{M}{\alpha}{N}$ is called the \emph{named application of $M$ to $N$ through $\alpha$}. Such notation is due to 
\cite{DBLP:journals/tcs/Vaux07}.

The reduction $\lam$ is the usual $\beta$-reduction (which we call ``$\lam$'' in order to avoid confusion with names). The reduction $\rho$ is just a renaming of names.
The novelty is the $\mu$-reduction which, in the following section, we are going to ``linearise''.
There are many reductions that one can consider on the $\lamu$-calculus; we chose to stick to those three because they are the ones considered in the original paper \cite{DBLP:conf/lpar/Parigot92}.

\begin{theorem}\label{thm:lamuConfl}
The $\lamu$-calculus $(\lamu, \to)$ is confluent.
\end{theorem}
\begin{proof}
See proof of Theorem 4.1 of \cite{WalterPyThesis98}.
\end{proof}

\begin{Lemma}\label{lamu-lm:writingM}
Every $\lamu$-term $M$ has the following shape:
\[M=\lam\vec{x}_1.\mu\alpha_1.\name{\beta_1}{\dots\lam\vec{x}_k.\mu\alpha_k.\name{\beta_k}{R\vec{Q}}}\]
where $R$ is either a variable, or a $\lam$-redex or a $\mu$-redex; furthermore, $R$, $\vec{Q}$, $k$, $\vec{x}_i$ and $\alpha_i$ are unique.
$R$ is called the \emph{head redex} of $M$ if it is a $\lamu$-redex, and it is called the \emph{head variable} of $M$ otherwise.
The sequence $\lam\vec{x}_1.\mu\alpha_1.\name{\beta_1}{\dots\lam\vec{x}_k.\mu\alpha_k.\name{\beta_k}{*}}$ is called the \emph{head} of $M$.
Therefore, every $\lamu\rho$-\emph{normal} $\lamu$-term $M$ has a head variable, has no $\rho$-redexes in its head and (with the previous notations) $\vec{Q}$ are $\lamu\rho$-normal $\lamu$-terms.
\end{Lemma}

Other than what we already said in the introduction, we will not add more explanations of the logical and programming meaning of this calculus.
Let us just add here the encoding of $\prog{\callcc}$ in $\lamu$-calculus: $\prog{\callcc}:=\lam y.\mu\alpha.\name{\alpha}{y(\lam x.\mu\delta.\name{\alpha}{x})}$.

\section{Resource $\lamu$-calculus}

Recall that a \emph{multiset} $A$ on a set $X$ is a map from $X$ to $\mathbb{N}$.
We use a multiplicative notation: the empty multiset is denoted with $1$ and the union of two multisets $A,B$ is denoted with $A*B$.
The set of multisets on a set $X$ is a monoid w.r.t. $*$, with neutral element $1$.
We denote with $\fmsets{X}$ the set of \emph{finite multisets} on $X$, that is, multisets $A$ with $X-A^{-1}(0)$ finite.
Such an $A$ will be as usual written as $A=[a_1,\dots,a_k]$, with $A(a_i)$ repetitions for each $a_i$.
We will sometimes write $m\pm A$ for $[m\pm a_1,\dots,m\pm a_k]$ if $m\pm a_i$ happens to be defined.

\begin{Definition}
The set $\lamu^\mathrm{r}$ of \emph{resource $\lamu$-terms} is given by:
\[
t::= x \ \mid \ \lam x.t \ \mid \ t_0[t_1,\dots,t_n] \ \mid \ \mu\alpha.\name{\beta}{t}
\]
where $[t_1,\dots,t_n]\in \, !\lamu^\mathrm{r}$ ($n\geq 0$), and it is called a \emph{bag}.
Resource terms are considered up to renaming of bound variables and names.
\emph{Resource-contexts} are defined as expected.
For $\nu$ a variable or a name, the \emph{degree} $\dg{\nu}{t}\in\N$ of $\nu$ in $t$, is defined as the number of free occurrences of $\nu$ in $t$.
\end{Definition}

The meaning of a resource sensitive application $(\lam x.t)[\vec{u}]$ is to non-deterministically choose a way to associate each resource in the bag with exactly one occurrence of the argument $x$ in $t$.
It is thus natural to consider (formal) sums.
If this association cannot be done without erasing or duplicating resources, then it annihilates to the empty sum $0$.
The operational semantics of a resource sensitive application $(\mu\alpha.\name{\beta}{t})[\vec{u}]$ will be discussed in Definition~\ref{def:resRed}.

\begin{Definition}\label{def:extensionToSums}
Call $2\langle\lamu^\mathrm{r}\rangle$ the free module generated by $\lamu^\mathrm{r}$ over the boolean semiring, which simply means the set of the \emph{formal sums} of finitely many $\lamu^\mathrm{r}$-terms, quotiented by commutativity, idempotency and associativity of $+$.
An element of $2\langle\lamu^\mathrm{r}\rangle$ will be called a \emph{sum} (in fact, it is just a finite subset of $\lamu^\mathrm{r}$).
We extend the constructors of $\lamu^\mathrm{r}$ to $2\langle\lamu^\mathrm{r}\rangle$ by linearity, setting:
\[\left(\sum_{i_0} t_{i_0}\right)\left[ \sum_{i_1}t_{i_1} ,\dots,\sum_{i_n} t_{i_n} \right]:=\sum_{i_0,\dots,i_n} t_{i_0}[t_{i_1},\dots,t_{i_n}]\]
and analogous for $\lam x.\sum_i t_i$ and $\mu\alpha.{}_{\beta}\Big| \sum_i t_i \Big|$.
We denote with $0$ the empty sum.
It is the neutral element for $+$ and the annihilating element for the above constructors (i.e. when it appears as any subterm, the whole term becomes $0$).
\end{Definition}

Let us define now a reduction in $\lamu^\mathrm{r}$ (or, better said, in $2\langle\lamu^\mathrm{r}\rangle$).
For this, we will need to divide a multiset into a certain number of ``blocks''. 
This notion already exists in the literature of combinatorics (see for example \cite{bender1974partitions}).

\begin{Definition}
A \emph{partition} (resp.\ \emph{weak partition}) of a multiset $[\vec{u}]$ is a multiset $[[\vec{v}_1],\overset{(k\geq 1)}{\dots},[\vec{v}_k]]$ of \emph{non empty} (resp.\ possibly empty) multisets such that $[\vec{u}]=[\vec{v}_1]*\dots*[\vec{v}_k]$.
A \emph{composition} (resp. \emph{weak composition} - \emph{w.c.} for short) of a multiset $[\vec{u}]$ is a tuple $([\vec{v}_1],\dots,[\vec{v}_k])$ of multisets s.t.\ $[[\vec{v}_1],\dots,[\vec{v}_k]]$ is a partition (resp.\ weak partition) of $[\vec{u}]$.
\end{Definition}

Observe that the empty bag $1$ admits no partitions but admits infinite weak partitions: they are the multisets of shape $[1,\dots,1]$ ($h\geq 1$ times $1$).
Here are some other examples:
the set of all the weak partitions of the bag $[x]$ is $\set{[[x]],[[x],1],[[x],1,1],\dots}$.
The set of all weak partitions of $[x,x]$ is
 $\!\set{[[x,x]],[[x],[x]],[[x,x],1],[[x],[x],1],$ $[[x,x],1,1],[[x],[x],1,1],...}$.

\begin{Definition}
Let $t\in\lamu^\mathrm{r}$ and $[\vec{u}]\in\,\fmsets{\lamu^\mathrm{r}}$.
The \emph{linear substitution} $t\langle[u_1,\dots,u_k]/x\rangle\in 2\langle\lamu^\mathrm{r}\rangle$ is defined, as usual, in Figure~\ref{fig:LinSub}.
In order to linearise the $\mu$-reduction we introduce the \emph{linear named application} $\langle t\rangle_\alpha [\vec{u}]\in 2\langle\lamu^\mathrm{r}\rangle$, defined in Figure~\ref{fig:LinNamedApp} \footnote{The induction takes into account also the case of named terms $\name{\eta}{t}$; this is done for technical reasons.}.
\end{Definition}

\begin{figure*}[t!]
 \centering
 \[
 \begin{array}{rcl}
  x\langle [v]/x \rangle=v
  &
  y\langle 1/x \rangle=y \,\,\, (y\neq x)
  &
  (\lam y.t)\langle [\vec{u}]/x \rangle = \lam y.t\langle [\vec{u}]/x \rangle
  \\
  x\langle 1/x \rangle = x\langle [v,w,\vec{u}]/x \rangle = 0
  &
  y\langle [v,\vec{u}]/x \rangle=0 \,\,\, (y\neq x)
  &
  (\mu\alpha.\name{\beta}{t})\langle [\vec{u}]/x \rangle = \mu\alpha.\name{\beta}{t\langle [\vec{u}]/x \rangle}
  \\
  \\
  (t[v_1,\dots,v_n])\langle [\vec{u}]/x \rangle =
  & 
  \sum\limits_{
  ([\vec{s}^{\,0}],\dots,[\vec{s}^{\,n}])
  \textit{ w.c. of }[\vec{u}]}
  &
  \!\!\!\!\!!\!\!t\langle [\vec{s}^{\,0}]/x \rangle\left[\,v_1\langle [\vec{s}^{\,1}]]/x \rangle,\dots,v_n\langle [\vec{s}^{\,n}]/x \rangle\,\right].
 \end{array}
 \]
 \caption{Definition of linear substitution}
 \label{fig:LinSub}
\end{figure*}

\begin{figure*}[t!]
 \centering
 \[
 \begin{array}{rcl}
  \lnamedapp{x}{\alpha}{[v,\vec{u}]}=0
  &
  \lnamedapp{x}{\alpha}{1}=x
  &
  \lnamedapp{\name{\eta}{t}}{\alpha}{[\vec{u}]}=
  \name{\eta}{\lnamedapp{t}{\alpha}{[\vec{u}]}}\quad(\textit{if }\eta\neq \alpha)
  \\
  \\
  \lnamedapp{\mu\gamma.\name{\eta}{t}}{\alpha}{[\vec{u}]} = \mu\gamma.\lnamedapp{\name{\eta}{t}}{\alpha}{[\vec{u}]}
  &
  \lnamedapp{\lam y.t}{\alpha}{[\vec{u}]} = \lam y.\lnamedapp{t}{\alpha}{[\vec{u}]}
  &
  \lnamedapp{\name{\alpha}{t}}{\alpha}{[\vec{u}]}=
  \sum\limits_{
  ([\vec{w}^{\,1}],[\vec{w}^{\,2}])
  \textit{ w.c. of }[\vec{u}]}
  \name{\alpha}{\left(\lnamedapp{t}{\alpha}{[\vec{w}^{\,1}]}\right)[\vec{w}^{\,2}]}
  \\
  \\
  \lnamedapp{t[v_1,\dots,v_n]}{\alpha}{[\vec{u}]} =
  & 
  \sum\limits_{
  ([\vec{w}^{\,0}],\dots,[\vec{w}^{\,n}])
  \textit{ w.c. of }[\vec{u}]}
  &
  \left(\lnamedapp{t}{\alpha}{[\vec{w}^{\,0}]}\right)\left[\,\lnamedapp{v_1}{\alpha}{[\vec{w}^{\,1}]},\dots,\lnamedapp{v_n}{\alpha}{[\vec{w}^{\,n}]}\,\right].
 \end{array}
 \]
 \caption{Definition of linear named application}
 \label{fig:LinNamedApp}
\end{figure*}

Remark that, thus, if $\dg{\alpha}{t}=0$ then $\langle t\rangle_\alpha 1 := t$ and $\langle t\rangle_\alpha [v,\vec{u}] := 0$;
if $\dg{\alpha}{t} =: d \neq 0$ then:
   $\langle t\rangle_\alpha [\vec{u}]$ is the sum
  $\sum
   t\left\{\name{\alpha}{(\cdot)[\vec{s}^{\,1}]}/_{\name{\alpha}{\cdot}^{(1)}},\dots,\name{\alpha}{(\cdot)[\vec{s}^{\,d}]}/_{\name{\alpha}{\cdot}^{(d)}}\right\}$,
 where the sum is taken over all $([\vec{s}^{\,1}],\dots,[\vec{s}^{\,d}])$ w.c.\ of $[\vec{u}]$ of length $d$ and $\name{\alpha}{\cdot}^{(1)},\dots,\name{\alpha}{\cdot}^{(d)}$ is any fixed enumeration of the occurrences of $\alpha$ in $t$.

\begin{Definition}\label{def:resRed}
Define a reduction $\rightarrow_\mathrm{r}\,\subseteq\lamu^\mathrm{r}\times 2\langle\lamu^\mathrm{r}\rangle$ as the resource-context closure of the union $\to_\mathrm{base^\mathrm{r}}$ of:
\[
(\lam x.t)[\vec{u}]\rightarrow_{\lam^\mathrm{r}} t\langle[\vec{u}]/x\rangle
\qquad
\mu\gamma.\name{\alpha}{\mu\beta.\name{\eta}{t}}\rightarrow_{\rho^\mathrm{r}} \mu\gamma.(\name{\eta}{t}\set{\alpha/\beta})
\]
\[
(\mu\alpha.\name{\beta}{t})[\vec{u}]\rightarrow_{\mu^\mathrm{r}} \mu\alpha.\langle \name{\beta}{t}\rangle_\alpha [\vec{u}].
\]
We extend it to all $2\langle\lamu^\mathrm{r}\rangle\times2\langle\lamu^\mathrm{r}\rangle$ setting:
\[\rightarrow_\mathrm{r}\,:=\set{(t+\Sum{S},\Sum{T}+\Sum{S}) \st t\rightarrow_\mathrm{r}\Sum{T}\mathrm{ \ and \ }t\notin\Sum{S}}.\]
\end{Definition}

Observe that the analogue of Lemma~\ref{lamu-lm:writingM} holds for $\lamu^\mathrm{r}$-terms (in particular we will use the notion of head variable/redex).

The work~\cite{DBLP:journals/pacmpl/BarbarossaM20} is an example of how a resource calculus can be useful, as it enjoys strong properties such as linearity, strong normalisation and confluence.
In the resource $\lam$-calculus the last two properties are easy; as we are going to see, in our case they are more involved.

\subsection{Strong normalisation}

With $\to_{\lam^\mathrm{r}}$ we erase exactly one $\lam$, with $\to_{\rho^\mathrm{r}}$ we erase exactly one $\mu$. With $\to_{\mu^\mathrm{r}}$ however, the situation is more subtle: we are not creating nor erasing $\lam$'s or $\mu$'s (which remain thus in constant number), but we are eventually making the reduct grow by creating an arbitrarily large number of new applications.
However, in order to pass from the $\mu$-redex $(\mu\alpha.\name{\beta}{t})[\vec{u}]$ to a reduct $t'\in\mu\alpha.\lnamedapp{\name{\beta}{t}}{\alpha}{[\vec{u}]}$, we: first, decompose $[\vec{u}]$ in several blocks; then, erase $[\vec{u}]$; finally, put each block \emph{inside} a certain \emph{named} subterm of $\name{\beta}{t}$.
We replaced thus a bag with many new bags which are at a ``deeper depth''.
As we will see in Remark \ref{rm:depthDegree}, it will be immediate to recognize that actually this depth is necessary bounded by the number of $\mu$-occurrences in the term, which is invariant under $\to_{\mu^\mathrm{r}}$, so the former subtracted to the latter should decrease.
Remark that in the case $[\vec{}u]=1,\dg{\mu}{\name{\beta}{t}}=0$ we do not create new applications but we simply erase one already existing one, so we have to make sure our measure decreases in this case as well.

\begin{Definition}
Let $t$ be a $\lamu$-term and let $b$ be an \emph{occurrence} of a bag or of a subterm of $t$.
The \emph{depth} $d_t(b)\in\N$ of $b$ in $t$ is the number of \emph{named} subterms of $t$ containing $b$.
\end{Definition}

\begin{Remark}\label{rm:depthDegree}
By definition of the grammar of the $\lamu$-calculus there are as many named subterms of $t$ as $\mu$-abstractions in $t$, i.e. $\dg{\mu}{t}$.
So we must have: $d_t(b)\leq\dg{\mu}{t}$.
\end{Remark}

\begin{Definition}\label{def:Measure}
Define the \emph{multiset measure} $\ms{t}\in\,\fmsets{\N}$ of a $\lamu^\mathrm{r}$-term $t$ as:
\[
\ms{t}:=\dg{\mu}{t}-[\,d_t(b)\mid b\textit{ occurrence of bag in }t\,].
\]
\end{Definition}

Remark \ref{rm:depthDegree} assures that $\ms{t}\in\,\fmsets{\N}$ (and not in $\fmsets{\mathbb{Z}}$).
This is crucial because it allows us to reason by induction w.r.t.\ the multiset order on it.
The measure $\ms{\cdot}$ is ``almost'' the good one for strong normalization:

\begin{Proposition}\label{lamu-prop:MuNormalizes!}
If $t\to_{\mu^\mathrm{r}} t'+\Sum T$ then $\ms{t}>\ms{t'}$.
\end{Proposition}
\begin{proof}
If $t\to_{\mu^\mathrm{r}} t'+\Sum T$ then $t=c\hole{(\mu\alpha.\name{\beta}{s})b_0}
$ and $t'=c\hole{h}$ with $h\in \mu\alpha.\lnamedapp{\name{\beta}{s}}{\alpha}{b_0}$ 
and $c$ a single-hole resource context.
Call $k:=\dg{\mu}{t}=\dg{\mu}{t'}$ and consider $\dg{\alpha}{\name{\beta}{s}}\in\N$.
The are two cases:

- Case $\dg{\alpha}{\name{\beta}{s}}=0$.
 By definition of $\to_{\mu^\mathrm{r}}$ this is possible only if $b_0
 =1$ (otherwise $t\to_{\mu^\mathrm{r}} 0$) and $h=\mu\alpha.\name{\beta}{s}$.
 So in $t$ there are the exact same occurrences of bags as in $t'$ and they are at the same depth, except for $b_0$ which is in $t$ but not in $t'$.
 This means that $\ms{t}=\ms{t'}*[\,k-d_{t}(b_0)\,]>\ms{t'}$.

- Case $\dg{\alpha}{\name{\beta}{s}}=:n\geq1$.
 Then:
 \[
 h=\mu\alpha.\name{\beta}{s}\Big\{
 \name{\alpha}{(\cdot)b_1}/_{{\name{\alpha}{\cdot}}^{(1)}}
 ,\dots,
 \name{\alpha}{(\cdot)b_{n}}/_{{\name{\alpha}{\cdot}}^{(n)}}
 \Big\}
 \]
 for a w.c.\ $(b_1,\dots,b_{n})$ of $b_0$.
 So $\ms{t'} = k - A'$ and $\ms{t} = k - A$, with $A'$ and $A$ respectively the multisets:
 \[
 \arraycolsep=1pt
 \begin{small}
 \begin{array}{ccccccc}
  B^c_{t'} & * &
  B^s_{t'} & * &
  [\,d_{t'}(b) \mid b \textit{ in a }v\in b_i\textit{ for an }i\,]  & * &
  [\,d_{t'}(b_1),\dots,d_{t'}(b_{n})\,]
  \\
  B^c_{t} & * &
  B^s_{t} & * &
  [\,d_{t}(b) \mid b \textit{ in a }v\in b_0]\,] & * &
  [\,d_{t}(b_0)\,],
 \end{array}
 \end{small}
 \]
 where we put $B^c_{t}:=[\,d_{t}(b) \mid b\textit{ in }c]$ (and analogously for $s,t'$).
 Now for $i=1,\dots,n$ we have:
 $d_{t'}(b_i)=d_{t'}(h)+d_h(b_i)> d_{t}(b_0)$ since as one sees from the expression of $h$, we have $d_{t'}(h) = d_{t}(b_0)$ and $d_{h}(b_i) > 0$.
 Also, it is easily understood that for all $b$ occurring in $c$, or occurring in $s$, we have:
 $d_{t'}(b)=d_t(b)$.
 Finally, observe that since $(b_1,\dots,b_{n})$ is a w.c.\ of $b_0$, then: $b$ occurs in some $v\in b_0$ iff $b$ occurs in some $v\in b_i$ for some $i$. And for all such $b$ we have:
 $
  d_{t'}(b)=d_{t'}(v)+d_v(b)>d_{t}(v)+d_{v}(b)=d_t(b)
 $
 since $d_{t'}(v)=d_{t'}(b_i)>d_{t}(b_0)=d_{t}(v)$.
 All these considerations precisely mean $\ms{t}>\ms{t'}$.\qedhere
\end{proof}

Analogously we find:

\begin{Proposition}\label{lamu-prop:LamNormalizes!}
If $t\to_{\lam^\mathrm{r}} t'+\Sum T$ then $\ms{t}>\ms{t'}$.
\end{Proposition}

However, only $\ms{t}$ is not enough to prove strong normalization.
In fact (reasoning similarly as before):

\begin{Proposition}\label{lamu-prop:RhoNormalizes!}
 If $t\to_{\rho^\mathrm{r}} t'+\Sum T$ then $\ms{t}\geq\ms{t'}$, and there are cases in which the equality holds, such as (for $\beta\neq\eta$): $\ms{\mu\gamma.\name{\alpha}{\mu\beta.\name{\eta}{x}}}=1=\ms{\mu\gamma.\name{\eta}{x}}$ with $\mu\gamma.\name{\alpha}{\mu\beta.\name{\eta}{x}} \to_{\rho^{\mathrm{r}}} \mu\gamma.\name{\eta}{x}$.
\end{Proposition}

That is why, in order to get a strongly normalising measure, we add another component:

\begin{Definition}\label{lamu-cor:SNm}
We define the measure\index{Measure!$\SNm{\cdot}$ multiset-}:
\[\SNm{t}:=(\ms{t},\dg{\mu}{t})\in \, ! \N\times \N\] 
ordered by the (well-founded) lexicographic order.
\end{Definition}

\begin{Corollary}[SN]\label{lamu-cor:SN}
 If $t\to_{\mathrm{r}} t'+\Sum T$ then $\SNm{t}> \SNm{t'}$.
Therefore, the resource reduction $\to_\mathrm{r}$ on sums is strongly normalising.
\end{Corollary}
\begin{proof}
 The only case in which $\ms{\cdot}$ may remain constant is along a $\rho^{\mathrm{r}}$-reduction, but in this case $\dg{\mu}{t}$ strictly decreases.
\end{proof}

Before turning to the confluence, let us see some properties of the measure $\ms{\cdot}$ that we will use in the following.

\begin{Lemma}\label{lamu-lm:mDecreasesOnSubterms}
Let $c=c\hole{\xi}$ be a single-hole context and $t$ a $\lamu$-term. Then: $\ms{c\hole{t}}\geq \ms{t}$.
\end{Lemma}
\begin{proof}
We have $\ms{c\hole{t}}=A*[\,\dg{\mu}{c\hole{t}}-d_{c\hole{t}}(b) \mid b \textit{ in }t\,]$
where $A:=[\,\dg{\mu}{c\hole{t}}-d_{c\hole{t}}(b) \mid b \textit{ in }c\,]$.
But $\dg{\mu}{c\hole{t}}=\dg{\mu}{c}+\dg{\mu}{t}$ and, for all occurrence $b$ in $t$, we have:
$
 d_{c\hole{t}}(b)= d_t(b)+d_c(\xi)\leq d_t(b)+\dg{\mu}{c}
$.
Thus, for all occurrence $b$ of bag in $t$, we have:
$
 \dg{\mu}{c\hole{t}}-d_{c\hole{t}}(b)\geq
 \dg{\mu}{t}-d_t(b)
$
and this last integer is exactly a generic element of $\ms{t}$ (if it is non-empty). Hence $\ms{c\hole{t}}\geq A*\ms{t}\geq\ms{t}$.
\end{proof}

However, there are cases in which $\ms{c\hole{t}}=\ms{t}$ even if $c\neq\xi$.
For example, taking $c=\lam x.\xi$ one has $\ms{c\hole{t}}=1=\ms{t}$ for all $t\in\lamur$ \emph{not} containing any bags.
This is exactly why, in the following, we will consider a slightly different size, called $\emph{\textbf{ms}}$ (defined in~\autoref{lamu-cor:MSMeasure}).

\begin{Lemma}\label{lamu-lm:mOnCT}
 Let $c=c\hole{\xi}$ be a single-hole resource context and $t,s\in\lamu$. Then:
\begin{enumerate}
 \item $\ms{c\hole{t}}
  =
  (\dg{\mu}{t}+\ms{c})
  *
  ((\dg{\mu}{c}-d_c(\xi)+\ms{t})$.
 \item If $\dg{\mu}{s}\leq\dg{\mu}{t}$ and $\ms{s}<\ms{t}$, then $\ms{c\hole{s}}<\ms{c\hole{t}}$.
\end{enumerate}
\end{Lemma}
\begin{proof}[Proof sketch]
Easily checked, thanks to the clear fact that if $b$ is the occurrence of a bag in $c$, then $d_{c\hole{t}}(b)=d_{c}(b)$.
\end{proof}

In the following, we will need a strong normalising measure which, in addition, satisfies the properties of the following Corollary~\ref{lamu-cor:MSMeasure}.
However, we have seen with some lines above that $\SNm{\cdot}$ is not adapted for that.
This is why we operate a last slight modification.
First, let us consider the size\index{Size!-of a $\lamu$-term} $\size{t}\in\N_{\geq 1}
$ of resource $\lamu$-terms: $\size x := 1$, $\size{\lam x.t} := 1 + \size t =: \size{\mu\alpha.\name{\beta}{s}}$, $\size{t_0[t_1,\dots,t_k]} := 1 + k + \sum\limits_{i=0}^k \size{t_i}$.
Of course $\size{t}=1$ iff $t$ is a variable, and for all $c$ single-hole context, $\size{c\hole{t}}\geq\size{t}$ where the equality holds iff $c=\xi$.

\begin{Corollary}\label{lamu-cor:MSMeasure}
Define a measure $\textbf{ms}(\cdot)$\index{Measure!$\textbf{ms}(\cdot)$ multiset-} of $\lamu$-terms as:
\[\textbf{ms}(t):=(\SNm{t}, \size{t})\in\,\fmsets{\N}\times\N\times\N\]
ordered lexicographically (and thus well-founded).
Then:
\begin{enumerate}
	\item $t$ is a variable iff $\textbf{ms}(t)$ takes its minimal value $(1,0,1)$.
	\item For all single-hole context $c=c\hole{\xi}$, we have $\textbf{ms}(c\hole{t})\geq\textbf{ms}(t)$, and the equality holds iff $c=\xi$.
	\item If $t\to_\mathrm{r} t'+\Sum T$ then $\textbf{ms}(t)>\textbf{ms}(t')$.
\end{enumerate}
\end{Corollary}

\subsection{Confluence}

Due to the presence of three different reductions, the confluence or our resource $\lamu$-calculus is not easy.
Another difficulty is raised from the fact that we placed ourselves in a qualitative setting, that is, with idempotent sums, so that we cannot always reduce a sum component-wise.
This is why we split the problem of the confluence in two steps: first, we show that the \emph{quantitative} resource $\lamu$-calculus (that is, where sum is \emph{not} idempotent, and thus coefficients matter) is confluent (\autoref{lamu-sec:+confl}); second, we show that its confluence implies the confluence of the calculus with no coefficients (\autoref{lamu-sec:from+tor}).
Before all that, let us precisely explain the notion of quantitative resource calculus:

\begin{Definition}\label{lamu-def:qualitlamur}
 The \emph{quantitative resource $\lamu$-calculus} $\N\langle \lamur \rangle$ is built as the qualitative one ($2\langle \lamur \rangle$, Definition~\ref{def:extensionToSums}) except for taking now ``$+$'' \emph{non-idempotent}.
 We define the three base-case reductions $\to^+_{\lam^\mathrm{r}},\to^+_{\mu^\mathrm{r}},\to^+_{\rho^\mathrm{r}}$ in $\lamur \times \N\langle \lamur \rangle$: the reduction $\to^+_{\rho^\mathrm{r}}$ is defined as usual, while $\to^+_{\lam^\mathrm{r}}$ and $\to^+_{\mu^\mathrm{r}}$ are defined as in~\autoref{def:resRed}, except for the fact that the linear substitution and linear named application are replaced with a modified version of them, denoted respectively $t\langle[\vec{u}]/x\rangle^+$ and $\langle t\rangle^+_\alpha [\vec{u}]$, and defined in the next~\autoref{lamu-def:linSubst^+}.
The contextual union of the base-reductions $\to^+_{\lam^\mathrm{r}},\to^+_{\mu^\mathrm{r}},\to^+_{\rho^\mathrm{r}}$ forms a reduction $\tor^+$ on $\lamur\times\N\langle \lamur \rangle$ which is extended to all $\N\langle \lamur \rangle \times \N\langle \lamur \rangle$ by taking $\set{(t+\Sum S , \Sum T+\Sum S) \mid t \tor^+ \Sum T}$ (remark that we dropped the annoying condition ``$t\notin \Sum S$'', since now coefficients matter; it is the main reason why we turn to this calculus).
\end{Definition}

\begin{Notation}
If $[u_1,\dots,u_k]$ is a bag -- with the written enumeration of (possibly multiple) elements -- and $W$ is a function $W:\set{1,\dots,k}\longto I=:\set{i_0< \dots < i_n}$, we will sometimes denote it by $W:(u_1,\dots,u_k)\longto I$, or by $W:(\vec{u})\longto I$. 
When we use such notation we mean that $W$ generates the w.c.\ $([u_j\mid j\in W^{-1}(i_0)],\dots,[u_j\mid j\in W^{-1}(i_n)])$ of $[u_1,\dots,u_k]$, and denoted by $([\vec{w}^{\,i_0}],\dots,[\vec{w}^{\,i_n}])$.
In the case $[\vec{u}]=1$, we write $W:()\longto I$ and we say that there is exactly one w.c.\ generated by $W$, namely $(1,\overset{(n+1\textit{ times})}{\dots},1)$.
\end{Notation}

\begin{Definition}\label{lamu-def:linSubst^+}
The quantitative version $t\langle[\vec{u}]/x\rangle^+$ of the linear substitution is defined exactly as in~\autoref{fig:LinSub} but by replacing the sum on all the $([\vec{w}^{\,0}],\dots,[\vec{w}^{\,n}])$ w.c.\ of $[\vec{u}]$ with the sum on all $W:(\vec{u})\longto\set{0,\dots,n}$, and by taking the above w.c.'s as the ones generated by $W$.
The quantitative version $\langle t\rangle^+_\alpha [\vec{u}]$ of the linear named application is defined exactly as in~\autoref{fig:LinNamedApp} but by replacing, in the case of an application, the sum on all the $([\vec{w}^{\,0}],\dots,[\vec{w}^{\,n}]) \textit{ w.c.\ of }[\vec{u}]$ with the sum on all $W:(\vec{u})\longto\set{0,\dots,n}$, and by taking the above w.c.'s as the ones generated by $W$.
Analogously for the case of a named term, where we use $W:(\vec{u})\longto\set{1,2}$.
\end{Definition}

For instance:
$(\mu\alpha.\name{\alpha}{\mu\eta.\name{\alpha}{x}})[y,y]
\tor^+$
$\mu\alpha.\name{\alpha}{(\mu\eta.\name{\alpha}{x1})[y,y]}$
 $+
 2 \, \mu\alpha.\name{\alpha}{(\mu\eta.\name{\alpha}{x[y]})[y]}
 +
\mu\alpha.\name{\alpha}{(\mu\eta.\name{\alpha}{x[y,y]})1}$.

In the following, $\supp{\Sum T} \in 2\langle\lamur\rangle$ is the \emph{support} of a $\Sum T \in \N\langle\lamur\rangle$, that is, the \emph{set} of its addends (with no coefficients: $\supp{\Sum T}$ is $\Sum T$ when considered with an idempotent ``$+$'').

\begin{Remark}
 It is clear by the definitions that if, for $t\in\lamur$, one has $t \tor \Sum T$ (in $2\langle\lamur\rangle$) and $t\tor^+ \Sum S$ (in $\N\langle\lamur\rangle$) by reducing the \emph{same} redex, then $\supp{\Sum S}= \Sum T$.
That is, the two reductions only differ for the coefficients.
Said differently, the qualitative substitutions $t\langle[\vec{u}]/x\rangle$ and $\langle t\rangle_\alpha [\vec{u}]$ are just the quantitative substitutions $t\langle[\vec{u}]/x\rangle^+$ and $\langle t\rangle^+_\alpha [\vec{u}]$ taken with boolean coefficients.
\end{Remark}

\begin{Remark}\label{lamu-rmk:SN}
Using the fact that the reduction $\tor$ is strongly normalising in $\lamur$ (\autoref{lamu-cor:SN}), we can prove that the reduction $\tor^+$ is strongly normaling in $\N\langle \lamur \rangle$.
It suffices to extend the strongly normalising measure $\SNm{\cdot}$ of $\lamur$ (\autoref{lamu-cor:SNm}) to $\N\langle \lamur \rangle$ by setting $\SNm{\Sum T} := [\SNm{t} \mid t \in \Sum T] \, \in ! (! \N\times \N)$, and use the multiset order.
\end{Remark}

\begin{Remark}[Embedding inside the differential $\lamu$-calculus]
In~\cite{DBLP:phd/hal/Vaux07}, Vaux defines a differential $\lamu$-calculus, let us call it $(\de\lamu,\to_\partial)$ in this remark, and proves its confluence.
Our resource $\lamu$-calcului $2\langle \lamur \rangle$ and $\N\langle \lamur \rangle$ are strictly related to it, as they translate into $\de\lamu$ via\footnote{Here we are considering that the reader knowns the syntax of $\de\lamu$.} $\de{(\cdot)} :\lamur \longto \de\lamu$ defined as:
$
 \de x := x, \de{(\lam x.t)} := \lam x.\de t, \de{(\mu\alpha.\name{\beta}{t})} := \mu\alpha.\name{\beta}{\de t},
 (t[u_1,\dots,u_k])^{\partial}$
$:= \left(\prog{D}^k \, t^{\partial} \bullet (u_1^{\partial},\dots,u_k^{\partial})\right)0$.
We can extend it to sums, both in $2\langle\lamur\rangle$ and in $\N\langle\lamur\rangle$, by linearity.
In the qualitative case (that is, if we consider $\de{(\cdot)} :2\langle\lamur\rangle \longto \de\lamu$), it is not a well-behaved embedding, because it does not preserve reductions.
On the contrary, it does in the quantitative case (that is, if we consider $\de{(\cdot)} :\N\langle\lamur\rangle \longto \de\lamu$), in the sense that: if $t\to_{\lamur}^+ \Sum T$ in $\N\langle \lamur \rangle$, then $t^{\partial} \msto[\partial] \Sum T^{\partial}$ in $\de{\lamu}$ (``$\msto$'' is the reflexive transitive closure of $\to$).

One may wonder if it is possible to use the confluence of $(\de\lamu,\to_\partial)$ to infer the confluence of our calculi.
In fact, it is possible to show that the local confluence of $\to_{\lamur}^+$ follows from the confluence of $\to_{\partial}$.
However, as the reader has probably noticed, we only talked about $\to^+_{\lamur}$, and not about the whole $\tor^+ \,\, = \,\, \to^+_{\lamu^{\mathrm{r}}} \cup \to^+_{\rho^{\mathrm{r}}}$.
This is simply because in~\cite{DBLP:phd/hal/Vaux07} the $\rho$-reduction is not considered.
Remark that, even if it is possible to prove the confluence of $\to_{\rho^{\mathrm{r}}}$ by itself, we cannot use it in order to entail the confluence of $\tor^+ \,\, = \,\, \to^+_{\lamu^{\mathrm{r}}} \cup \to^+_{\rho^{\mathrm{r}}}$ by invoking the well-known Hindley-Rosen lemma.
This is because $\to_{\rho^{\mathrm{r}}}^+$ and $\to_{\lamur}^+$ do \emph{not} commute, as the following example shows (where $\gamma\neq\eta\neq\alpha$):
$\mu\alpha.\name{\alpha}{(\mu\gamma.\name{\eta}{x})1}
 \,\,{}_{\mu^\mathrm{r}}\!\!\leftarrow^{\!\!\!\!\!\!\!\!\!\!\!+}
 \,\,\,\,(\mu\alpha.\name{\alpha}{\mu\gamma.\name{\eta}{x}})1
 \to^+_{\rho^\mathrm{r}}
 (\mu\alpha.\name{\eta}{x})1
 \to^+_{\rho^\mathrm{r}}
 \mu\alpha.\name{\eta}{x}$,
but $\mu\alpha.\name{\alpha}{(\mu\gamma.\name{\eta}{x})1} \not\to^+_{\rho^\mathrm{r}} \mu\alpha.\name{\eta}{x}$.
In the previous ``non-reduction'', the blocked $\rho$-redex can be unblocked by performing a $\mu$-reduction (and the diagram closes).
O. Laurent suggests (private communication) that we could still use the confluence of $(\de\lamu,\to_\partial)$ in order to obtain the confluence of $\tor^+ \,\, = \,\, \to^+_{\lamu^{\mathrm{r}}} \cup \to^+_{\rho^{\mathrm{r}}}$ passing through a factorization lemma: if $t\msto[\mathrm{r}]^+\Sum T$ then $t\msto[\lamu^\mathrm{r}]^+\Sum T' \msto[\rho^\mathrm{r}]^+\Sum T$, for some $\Sum T'$.
\end{Remark}

\subsubsection{Confluence of $(\N\langle\lamur\rangle,\tor^+)$}\label{lamu-sec:+confl}

We present here a proof which essentially consists in closing the diagrams of all the possible critical pairs.

\begin{Remark}\label{lamu-rm:easyTakeOutSum}
We can extend the definition of linear substitution and linear named application to sums by linearity.
Analogously, the renaming of a sum $\Sum T \set{\alpha/\beta}$ is defined component-wise.
With these definitions in place one checks that base-step-reduction lifts to sums\index{Reduction!-for $\lamu$-sums}, i.e.
$\left(\mu\alpha.\name{\beta}{\Sum{T}}\right)[\vec{\Sum{U}}]\msto[\mu^\mathrm{r}]^+ \mu\alpha.\left\langle \Sum{T} \right\rangle^+_\alpha[\vec{\Sum{U}}]$ and analogously for $\left(\lam x.\Sum T\right)[\vec{\Sum{U}}]$ and $\mu\alpha.\name{\beta}{\mu\gamma.\name{\eta}{\Sum T}}$.
One can also check that $\tor^+$ on $\N\langle\lamur\rangle$ is contextual.
\end{Remark}

\begin{Notation}\label{notation:delta}
In this section we will sometimes use the following notation: for $\alpha,\beta,\eta$ names, we set $\delta_\eta^\alpha(\beta)$ to be $\alpha$ if $\beta=\eta$, or $\eta$ otherwise.
\end{Notation}

The following is the crucial technical lemma.

\begin{Lemma}\label{lamu-lm:ForLocalConfluence}
Let $t,s\in\lamu^\mathrm{r}$, $x$ a variable, $\alpha,\beta$ names and $[\vec{u}]$ a bag. 
If $s\rightarrow^+_\mathrm{r} \Sum S$ then:
\begin{enumerate}
    \item $s\set{\alpha/\beta}\msto[\mathrm{r}]^+ \Sum S\set{\alpha/\beta}$
    \item $t\langle [s,\vec{u}]/x\rangle^+\msto[\mathrm{r}]^+ t\langle [\Sum S,\vec{u}]/x\rangle^+$
	\item $s\langle [\vec{u}]/x\rangle^+\msto[\mathrm{r}]^+ \Sum S\langle [\vec{u}]/x\rangle^+$
	\item $\langle t\rangle^+_\alpha [s,\vec{u}] \msto[\mathrm{r}]^+ \langle t \rangle^+_\alpha [\Sum S,\vec{u}]$
	\item $\mu\alpha.\langle \name{\beta}{t}\rangle^+_\alpha [s,\vec{u}] \msto[\mathrm{r}]^+ \mu\alpha. \langle \name{\beta}{t} \rangle^+_\alpha [\Sum S,\vec{u}]$
	\item $\langle s\rangle^+_\alpha [\vec{u}] \msto[\mathrm{r}]^+ \langle \Sum S\rangle^+_\alpha [\vec{u}]$.
	\item $\mu\alpha.\langle \name{\beta}{s}\rangle^+_\alpha [\vec{u}] \msto[\mathrm{r}]^+ \mu\alpha. \langle \name{\beta}{\Sum S}\rangle^+_\alpha [\vec{u}]$.
\end{enumerate}
\end{Lemma}

Before proving it, let us remark that, in the qualitative setting, it is false.
For instance, if $s\tor s'$, then $(x[x])\langle [s,s]/x\rangle = s[s] \not\msto[\mathrm{r}] s[s']+s'[s] =(x[x])\langle [s,s']/x\rangle$.
In the quantitative case, instead, $(x[x])\langle [s,s]/x\rangle^+= 2\, s[s] \msto[\mathrm{r}]^+ s[s']+s'[s] =(x[x])\langle [s,s']/x\rangle^+$.

\begin{proof}[Proof sketch of Lemma~\ref{lamu-lm:ForLocalConfluence}]
1).
Induction on $s$.
The only interesting cases are:

- Case $s=\mu\gamma.\name{\eta}{s'}$: we have two subcases:
Subcase $s\to_{\mathrm{r}}^+ \Sum S$ is performed by reducing $s'$: easy by inductive hypothesis.
Subcase $s'=\mu\gamma'.\name{\eta'}{s''}$ and $s\to_{\mathrm{r}}^+ \Sum S$ is performed by reducing its leftmost $\rho$-redex:
then $s=\mu\gamma.\name{\eta}{\mu\gamma'.\name{\eta'}{s''}}$, $\Sum S=\mu\gamma.\name{\eta'}{s''}\set{\eta/\gamma'}$ and we have the four sub-subcases $\eta=\beta$ and $\eta'=\beta$, or $\eta=\beta$ and $\eta'\neq\beta$, or $\eta\neq\beta$ and $\eta'=\beta$, or $\eta\neq\beta$ and $\eta'\neq\beta$.
They are all similar, let us only show the second one, for which we have:
\[
	\begin{array}{rl}
	\Sum S\set{\alpha/\beta}=&\mu\gamma.\name{\eta'}{s''}\set{\alpha/\beta,\alpha/\gamma'}=\mu\gamma.\name{\delta_{\eta'}^\alpha(\gamma')}{s''\set{\alpha/\beta,\alpha/\gamma'}} \\
	s\set{\alpha/\beta}=&\mu\gamma.\name{\alpha}{\mu\gamma'.\name{\eta'}{s''\set{\alpha/\beta}}}\to_\rho^+ \mu\gamma.\name{\eta'}{s''\set{\alpha/\beta}}\set{\alpha/\gamma'}\\
	=&\mu\gamma.\name{\delta_{\eta'}^\alpha(\gamma')}{s''\set{\alpha/\beta,\alpha/\gamma'}}=\Sum S\set{\alpha/\beta}.
	\end{array}
\]
- Case $s=s'[\vec{v}]$: we have four subcases depending on how the reduction $s\to_{\mathrm{r}}^+ \Sum S$ is performed. The only interesting one is the subcase $s'=\mu\gamma.\name{\eta}{s''}$ and $s\to_{\mathrm{r}}^+ \Sum S$ is performed by reducing the $\mu$-redex $s$, for which we have:
$\Sum S=\mu\gamma.\coefflnamedapp{\name{\eta}{s''}}{\gamma}{[\vec{v}]}$ and
$
s\set{\alpha/\beta} 
= 
(\mu\gamma.\name{\delta_\eta^\alpha(\beta)}{s''\set{\alpha/\beta}})[\vec{v}\set{\alpha/\beta}]$
$\to_{\mu^{\mathrm{r}}} $-reduces to
$\mu\gamma.\coefflnamedapp{\name{\delta_\eta^\alpha(\beta)}{s''\set{\alpha/\beta}}}{\gamma}{[\vec{v}\set{\alpha/\beta}]}$
which in turn coincides with the sum $\mu\gamma.\coefflnamedapp{\name{\eta}{s''}\set{\alpha/\beta}}{\gamma}{[\vec{v}\set{\alpha/\beta}]}=\Sum S\set{\alpha/\beta}.
$

(2).
Induction on $t$.
The only non-trivial case is when $t$ is $v_0[v_1,\dots,v_n]$.
In this case we can write $t\langle [s,\vec{u}]/x\rangle^+ $ as:
\[
 \sum\limits_{W}\,
 \sum\limits_{j=0}^n
 (\, v_0\langle [\vec{w}^{\,0}]*[s]_0^j/x \rangle \,)\,[\dots,v_i\langle [\vec{w}^{\,i}]*[s]_i^j/x \rangle,\dots]
\]
where $W:(\vec{u})\longto\set{1,\dots,n}$ and we put $[s]_i^j$ to be the singleton multiset $[s]$ if $i=j$, and the empty mulitset $1$ if $i\neq j$..
Fix now a $W:(\vec{u})\longto\set{1,\dots,n}$ (together with its generated w.c.) and consider each of the $n+1$ elements of the sum on $j$. 
We write the case for $j=0$, but the other cases are exactly the same. 
Since $j=0$, the element is $(\, v_0\langle [\vec{w}^{\,0}]*[s]/x \rangle \,)\,[\dots,v_i\langle [\vec{w}^{\,i}]/x \rangle,\dots]$ and by inductive hypothesis it $\msto[\mathrm{r}]^+$-reduces to
$
 (\, v_0\langle [\vec{w}^{\,0}]*[\Sum S]/x \rangle \,)\,[\dots,v_i\langle [\vec{w}^{\,i}]/x \rangle,\dots].
$
Now summing up all the elements for $j=0,\dots,n$ and $W:(\vec{u})\longto\set{1,\dots,n}$ we obtain the following sum:
\[
 \sum\limits_{W}\,
 \sum\limits_{j=0}^n
  (\, v_0\langle [\vec{w}^{\,0}]*[\Sum S]_0^j/x \rangle \,)[\dots,v_i\langle [\vec{w}^{\,i}]*[\Sum S]_i^j/x \rangle,\dots].
\]
which can be shown to be the desired
$
 (\,v_0[v_1,\dots,v_n]\,)\,\langle [\Sum S,\vec{u}]/x\rangle^+ .
$

(3).
Induction on $s$. We only show the case $s=\mu\alpha.\name{\beta}{s'}$, which splits in two subcases:
the subcase where $s\to_{\mathrm{r}}^+ \Sum S$ is performed by reducing $s'$ is immediate.
The subcase where $s'=\mu\gamma.\name{\eta}{s''}$ and $s\to_{\mathrm{r}}^+ \Sum S$ is performed by reducing its leftmost $\rho$-redex goes as follows:
we have $\Sum S=\mu\alpha.\name{\eta}{s''}\set{\beta/\gamma}$ and\\
$
  s\langle [\vec{u}]/x\rangle^+ 
  = 
  \mu\alpha.\name{\beta}{\mu\gamma.\name{\eta}{s''\langle [\vec{u}]/x \rangle^+}} $\\
  $\to_{\rho^{\mathrm{r}}} 
  \mu\alpha.\name{\eta}{s''\langle [\vec{u}]/x \rangle^+}\set{\beta/\gamma}$
 $= 
 \mu\alpha.\name{\eta}{s''}\langle [\vec{u}]/x \rangle^+\set{\beta/\gamma} $\\
 $= 
 \mu\alpha.\name{\eta}{s''}\set{\beta/\gamma}\langle [\vec{u}]/x \rangle^+ 
  = 
  \Sum S\langle [\vec{u}]/x \rangle^+.
$

(4).
Induction on $t$.
Similar to point (2).

(5).
It is easy discriminating the cases $\alpha=\beta$ and $\alpha\neq\beta$ and concluding by point (4).

(6).
Induction on $s\in\lamu$.
The only interesting cases are:

- Case $s=\mu\beta.\name{\gamma}{s'}$.
We have two subcases:
the subcase where $s\to_{\mathrm{r}}^+ \Sum S$ is performed by reducing $s'$, so $\Sum S=\mu\beta.\name{\gamma}{\Sum S'}$ with $s'\to_\mathrm{r} \Sum S'$, is easy by inductive hypothesis (however remark that we \emph{cannot} immediately apply the inductive hypothesis on $\name{\gamma}{s'}$, simply because the named term $\name{\gamma}{s'}\notin\lamu^{\mathrm{r}}$).
The subcase where $s'=\mu\gamma'.\name{\eta}{s''}$ (with $\gamma\neq\gamma'$) and $s\to_{\mathrm{r}}^+ \Sum S$ is performed by reducing its leftmost $\rho$-redex goes as follows:
we have $\Sum S=\mu\beta.\name{\eta}{s''}\set{\gamma/\gamma'}$ and we split in two sub-subcases depending whether $\alpha\neq\gamma$ or $\alpha=\gamma$.
Let us only show this last sub-subcase:
We have (putting $W:(\vec{u})\longto\set{1,2}$):
\[\begin{array}{rllr}
 \coefflnamedapp{s}{\alpha}{[\vec{u}]}
 & = &
 \sum\limits_{W}
 \mu\beta.\name{\alpha}{
  \,(\,
   \mu\gamma'.
   \coefflnamedapp{
    \name{\eta}{s''}
   }{\alpha}{[\vec{w}^{\,1}]}
  \,)\,
  [\vec{w}^{\,2}]
 }
 &
 \\
 & \mstor{\mu}^+ &
 \sum\limits_{W}
 \mu\beta.\name{\alpha}{
  \mu\gamma'.\coefflnamedapp{
   \coefflnamedapp{\name{\eta}{s''}}{\alpha}{[\vec{w}^{\,1}]}
  }{\gamma'}{[\vec{w}^{\,2}]}
 } &
 \\
 & \mstor{\rho}^+ &
 \sum\limits_{W}
 \mu\beta.
 \coefflnamedapp{
  \coefflnamedapp{
   \name{\eta}{s''}
  }{\alpha}{[\vec{w}^{\,1}]}
 }{\gamma'}{[\vec{w}^{\,2}]}
 \, \set{\alpha/\gamma'} &
 \\
 & = &
 \sum\limits_{W}
 \coefflnamedapp{
  \coefflnamedapp{
   \mu\beta.\name{\eta}{s''}
  }{\alpha}{[\vec{w}^{\,1}]}
 }{\gamma'}{[\vec{w}^{\,2}]}
 \, \set{\alpha/\gamma'}
 \\
 & = &
 \coefflnamedapp{\mu\beta.\name{\eta}{s''}\,\set{\alpha/\gamma'}}{\alpha}{[\vec{u}]}
 &
 \\
 & = &
 \coefflnamedapp{\Sum S}{\alpha}{[\vec{u}]}.
 &
\end{array}\]
- Case $s=s'[v_1,\dots,v_n]$: we have four subcases depending on how the reduction $s\to_{\mathrm{r}}^+ \Sum S$ is performed.
We only show the one in which $s'=\mu\gamma.\name{\eta}{s''}$ (with $\gamma\neq\alpha$) and $s\to_{\mathrm{r}}^+ \Sum S$ is performed by reducing the $\mu$-redex $s$.
In this subcase we have $\Sum S = \mu\gamma.\coefflnamedapp{\name{\eta}{s''}}{\gamma}{[\vec{v}]}$
and (putting $W:(\vec{u})\longto\set{0,\dots,n}$):
\[\small\begin{array}{rcl}
 \coefflnamedapp{s}{\alpha}{[\vec{u}]}
 & = &
 \sum\limits_{W}
 (\coefflnamedapp{\mu\gamma.\name{\eta}{s''}}{\alpha}{[\vec{w}^{\,0}]})\,[\dots,\coefflnamedapp{v_i}{\alpha}{[\vec{w}^{\,i}]},\dots]
 \\
 & \mstor{\mu}^+ &
 \mu\gamma.
 \sum\limits_{W}
 \coefflnamedapp{
 \coefflnamedapp{\name{\eta}{s''}}{\alpha}{[\vec{w}^{\,0}]}}
 {\gamma}
 {[\dots,\coefflnamedapp{v_i}{\alpha}{[\vec{w}^{\,i}]},\dots]}
 \\
 & = &
 \mu\gamma.
 \coefflnamedapp{
 \coefflnamedapp{
 \name{\eta}{s''}}{\gamma}{[\vec{v}]}
 }{\alpha}{[\vec{u}]}
 \\
 & = &
 \coefflnamedapp{\Sum S}{\alpha}{[\vec{u}]}.
\end{array}\]

(7).
It is immediate by discriminating the cases $\alpha=\beta$ and $\alpha\neq\beta$ and then concluding by point (6).\qedhere
\end{proof}

\begin{Proposition}\label{lamu-prop:lamu^+Confl}
The reduction $\tor^+$ is locally confluent in $\N\langle \lamur \rangle$.
\end{Proposition}
\begin{proof}[Proof sketch]
We show, by induction on a single-hole resource context $c$, that if
$t\rightarrow^+_{\mathrm{base}^\mathrm{r}}\Sum T$ and $c\hole{t} \tor^+ \Sum{T}_2$,
then there is $\Sum T'\in\N\langle\lamur\rangle$ s.t.\ $c\hole{\Sum T} \msto[\mathrm{r}]^+ \Sum T'\, {}^+_{\mathrm{r}}\!\!\!\twoheadleftarrow \Sum{T}_2$.
The proof crucially uses Lemma~\ref{lamu-lm:ForLocalConfluence} and Remark~\ref{lamu-rm:easyTakeOutSum}.
All the cases of the induction are either easy by induction, or they reduce to the case $c=\xi$, so this is the only one we sketch below.

We have $c\hole{t}=t\rightarrow^+_{\mathrm{base}^\mathrm{r}} \Sum T$ and we only have the three base-cases of Definition~\ref{lamu-def:qualitlamur}.

Case $t=(\lam x.s)[\vec{u}]$ and $\Sum T = s\langle[\vec{u}]/x \rangle^+$. 
Then $c\hole{t}=t\tor^+\Sum T_2$ (on a different redex than $t$) can only be performed either by reducing $s$, or by reducing an element $w$ of $[\vec{u}]$.
We have thus the two easy respective diagrams.

Case $t=(\mu\alpha.\name{\beta}{s})[\vec{u}]$ and $\Sum T = \mu\alpha.\langle \name{\beta}{s} \rangle^+_\alpha[\vec{u}]$.
Then $c\hole{t}=t\tor^+\Sum T_2$ (on a different redex than $t$) can only be performed either by reducing $s$, giving rise to an easy diagram, or by reducing an element $w$ of $[\vec{u}]$, giving rise to an easy diagram, or if $s=\mu\gamma.\name{\eta}{s'}$ and we reduce the $\rho$-redex $...\name{\beta}{\mu\gamma.\,...}$. In the latter case we split into the case $\alpha\neq\beta$, the case $\alpha=\beta,\gamma\neq\eta,\eta=\alpha$, the case $\alpha=\beta,\gamma\neq\eta,\eta\neq\alpha$, and the case $\alpha=\beta,\eta=\gamma$ (with necessarily $\gamma\neq\alpha$).
These four cases respectively correspond to four non-trivial (but similar) diagrams, of which we only show the one corresponding to the case $\alpha=\beta,\gamma\neq\eta,\eta=\alpha$:
$(\mu\alpha.\name{\alpha}{\mu\gamma.\name{\alpha}{s'}})[\vec{u}]$ reduces both to $\Sum U:=\mu\alpha.\coefflnamedapp{\name{\alpha}{\mu\gamma.\name{\alpha}{s'}}}{\alpha}{[\vec{u}]}$ and to $v:=(\mu\alpha.\name{\alpha}{s'\set{\alpha/\gamma}})[\vec{u}]$.
Now, $v \tor^+ \mu\alpha.\coefflnamedapp{\name{\alpha}{s'\set{\alpha/\gamma}}}{\alpha}{[\vec{u}]} = \sum\limits_{W}
\mu\alpha.\name{\alpha}{(\coefflnamedapp{s'\set{\alpha/\gamma}}{\alpha}{[\vec{w}^{\,0}]})[\vec{w}^{\,1}]}$ $=:\Sum V$ (with $W:(\vec{u})\to\set{1,2}$),
while it is easy to see that (with $W:(\vec{u})\to\set{1,2}, D:(\vec{w}^{\,0})\to\set{1,2}$) $\Sum U$ $\msto[\mathrm{r}]^+$-reduces to $\sum\limits_{W}
\sum\limits_{D}
\mu\alpha.\name{\alpha}{\mu\gamma.\name{\alpha}{\coefflnamedapp{(\coefflnamedapp{s'}{\alpha}{[\vec{d}^{\,0}]})[\vec{d}^{\,1}]}{\gamma}{[\vec{w}^{\,1}]}}}$ which in turn $\msto[\mathrm{r}]^+$-reduces to $\sum\limits_{W}
\sum\limits_{D}
\mu\alpha.\name{\alpha}{\coefflnamedapp{(\coefflnamedapp{s'}{\alpha}{[\vec{d}^{\,0}]})[\vec{d}^{\,1}]}{\gamma}{[\vec{w}^{\,1}]}}\set{\alpha/\gamma} =: \Sum U'$.
We can show that $\Sum V = \Sum U'$, so the diagram is closed.

Case $t=\mu\gamma.\name{\alpha}{\mu\beta.\name{\eta}{s}}$ and $\Sum T = \mu\gamma.\name{\eta}{s}\set{\alpha/\beta}$.
Then $c\hole{t}=t\tor^+\Sum T_2$ (on a different redex than $t$) can be only performed either by reducing $s$, which gives an easy diagram, or if $s=\mu\gamma'.\name{\eta'}{s'}$ and we reduce the $\rho$-redex $...\name{\eta}{\mu\gamma'.\,...}$.
Putting $\delta_0:=\delta^\alpha_\eta(\beta)$, $\delta'_1:=\delta^\alpha_{\eta'}(\beta)$, $\delta_1:=\delta^{\delta_0}_{\delta'_1}(\gamma')$, $\delta'_2:=\delta^{\eta}_{\eta'}(\gamma')$ and $\delta_2:=\delta^{\alpha}_{\delta'_2}(\beta)$, the latter case gives the following diagram:
$\mu\gamma.\name{\alpha}{\mu\beta.\name{\eta}{\mu\gamma'.\name{\eta'}{s'}}}$ reduces both to $\Sum U:=\mu\gamma.\name{\delta_0}{\mu\gamma'.\name{\delta'_1}{s'\set{\alpha/\beta}}}$ and to $\Sum V:=\mu\gamma.\name{\alpha}{\mu\beta.\name{\delta'_2}{s'\set{\eta/\gamma'}}}$, while $\Sum U$ $\msto[\mathrm{r}]^+$-reduces to $\mu\gamma.\name{\delta_1}{s'\set{\alpha/\beta}\set{\delta_0/\gamma'}}$ and
$\Sum V$ $\msto[\mathrm{r}]^+$-reduces to
$\mu\gamma.\name{\delta_2}{s'\set{\eta/\gamma'}\set{\alpha/\beta}}$.
We can show those sums equal, so the diagram is closed.
\qedhere
\end{proof}

\begin{Corollary}\label{lamu-cor:UNF+}
The reduction $\tor^+$ is confluent in $\N\langle \lamur \rangle$.
\end{Corollary}
\begin{proof}
By the well-known Newman Lemma, thanks to~\autoref{lamu-rmk:SN} and~\autoref{lamu-prop:lamu^+Confl}.
\end{proof}

\subsubsection{From the confluence of $(\N\langle\lamur\rangle,\tor^+)$ to the confluence of $(2\langle\lamur\rangle,\tor)$}\label{lamu-sec:from+tor}

\begin{Definition}\label{lamu-def:bigstep+}
The reduction $\Rightarrow \subseteq \N\langle\lamur\rangle \times \N\langle\lamur\rangle$ is defined as the contextual closure of the relation: \[\set{(m \, t+\Sum S , m \, \Sum T+\Sum S) \mid m\in\N, \, t \tor^+ \Sum T, \, t \notin \supp{\Sum S}}.\]
\end{Definition}

We have $\Rightarrow \, \subseteq \, \twoheadrightarrow^+_{\mathrm{r}}$.
It is also easily seen that if $\Sum T \tor \Sum S$ in $2\langle\lamur\rangle$, then for all $m_t\in\N$ (with $t\in\Sum T$), we have $\sum\limits_{t\in\Sum T} m_t t \Rightarrow \Sum S'$ (in $\N\langle\lamur\rangle$), with $\supp{\Sum S'}=\Sum S$.

\begin{Corollary}\label{lamu-cor:lamuConfl}
 The reduction $\tor$ in $2\langle\lamur\rangle$ is locally confluent.
\end{Corollary}
\begin{proof}
Let $\Sum T_1 \,\, {}_{\mathrm{r}}\!\!\leftarrow t \tor \Sum T_2$ in $2\langle\lamur\rangle$.
Since we know that $\tor$ is strongly normalising (\autoref{lamu-cor:SN}), there are (in $2\langle\lamur\rangle$) reductions $\Sum T_1 \msto[\mathrm{r}] \Sum S_1$ and $\Sum T_2 \msto[\mathrm{r}] \Sum S_2$, with $\Sum S_1,\Sum S_2$ $\mathrm{r}$-normal.
Therefore we have (in $\N\langle\lamur\rangle$) reductions $t \Rightarrow\cdots\Rightarrow \Sum S'_1$ and $t \Rightarrow\cdots\Rightarrow \Sum S'_2$, for some $\Sum S'_1,\Sum S'_2\in\N\langle\lamur\rangle$ s.t.\ $\supp{\Sum S'_i}=\Sum S_i$.
But then, since $\Sum S_i$ is $\mathrm{r}$-normal, $\Sum S'_i$ must be $\tor^+$-normal.
Now because of~\autoref{lamu-cor:UNF+}, it must be $\Sum S'_1=\Sum S'_2$, and therefore $\Sum S_1=\supp{\Sum S'_1}=\supp{\Sum S'_2}=\Sum S_2$.
Hence, we found a common reduct of $\Sum T_1,\Sum T_2$.
\end{proof}

\begin{Corollary}[Confluence]\label{lamu-cor:confluence}
The reduction $\tor$ is confluent on $2\langle\lamur\rangle$.
\end{Corollary}
\begin{proof}
By Newman Lemma, thanks to~\autoref{lamu-cor:SN} and~\autoref{lamu-cor:lamuConfl}.
\end{proof}

\section{Qualitative Taylor expansion}

\subsection{Crucial properties}

The calculus and its resource sensitive version are almost the same; the Taylor expansion map makes us pass from one to the other.

\begin{Definition}
The (qualitative) Taylor expansion is the map $\mathcal{T}:\lamu\to\mathcal{P}(\lamu^\mathrm{r})$ defined as:
\begin{enumerate}
	\item[] $\Te{x}:=\set{x}$ \qquad $\Te{\lam x.M}:=\set{\lam x.t \st t\in\Te{M}}$
	\item[] $\Te{\mu\alpha.\name{\beta}{M}}:=\set{\mu\alpha.\name{\beta}{t} \st t\in\Te{M}}$
	\item[] $\Te{MN}:=\set{t[\vec{u}] \st t\in\Te{M},\,[\vec{u}]\in\,\fmsets{\Te{N}}}$.
\end{enumerate}
\end{Definition}

Since $\rightarrow_\mathrm{r}$ is confluent and strongly normalising, all resource terms $t$ have a unique $r$-normal form $\nf[\mathrm{r}]{t}\in 2\langle\mathfrak{\lamu^\mathrm{r}}\rangle$ (it can be $0$).
Therefore, for all $M\in\lamu$ there always exists $\NFT{M}:=\bigcup\limits_{t\in\Te{M}}\nf[\mathrm{r}]{t}\subseteq\lamu^\mathrm{r}$ (in general infinite, thus not a sum).
This allows to endow $\lamu$ with a preorder:
\[
 M\leq N\textit{ iff }\,\NFT{M}\subseteq\NFT{N}.
\]

\begin{theorem}[Monotonicity]\label{thm: Monotonicity}
For $C$ a context, the map $C\hole{\cdot}:\lamu\rightarrow\lamu$ is monotone w.r.t. $\leq$.
\end{theorem}
\begin{proof}
Induction on $C$, as in \cite{DBLP:journals/pacmpl/BarbarossaM20}.
\end{proof}

The following technical lemma says that Taylor expansion behaves well w.r.t. substitutions.

\begin{Lemma}\label{lm:TaylorBehavesSubst}
One has:
\begin{enumerate}
\item 
$
\Te{M\set{\alpha/\beta}}=\Te{M}\set{\alpha/\beta}
$
\item
$
\Te{M\set{N/x}}=\bigcup\limits_{t\in\Te{M}}\bigcup\limits_{\vec{u}\in\,\fmsets{\Te{N}}} t\langle[\vec{u}]/x\rangle
$
\item
$
\Te{\namedapp{M}{\alpha}{N}}=\bigcup\limits_{t\in\Te{M}}\bigcup\limits_{\vec{u}\in\,\fmsets{\Te{N}}} \langle t\rangle_\alpha[\vec{u}].
$
\end{enumerate}
\end{Lemma}
\begin{proof}[Proof sketch]
(1). Straightforward induction on $M$.\\
(2). Induction on $M$ as one does for $\lam$-calculus. The only new case is $M=\mu\beta.\name{\alpha}{P}$ but it is done straightforwardly exactly as the case $M=\lam x.P$.\\
(3). Induction on $M$. Not more difficult than (2).
\end{proof}

The following important ``simulation property'' says in which sense the elements of $\Te{M}$ approximate $M$.

\begin{Proposition}\label{prop:LiftAssSimulation}
If $M\rightarrow_\mathrm{base} N$, then:
\begin{enumerate}\label{supposition}
	\item for all $s\in\Te{M}$ there exist $\Sum{T}\subseteq\Te{N}$ s.t.\ $s\msto[\mathrm{r}]\Sum{T}$
	\item for all $s'\in\Te{N}$ there is $s\in\Te{M}$ s.t.\ $s\msto[\mathrm{r}]s'+\Sum{T}$ for some sum $\Sum{T}\subseteq\Te{N}$.
\end{enumerate}
Furthermore, the same property lifts to all $\to$, that is, if $M\to N$ then point (1) and (2) hold.
\end{Proposition}
\begin{proof}[Proof sketch]
Points 1) and 2) are easy using Lemma \ref{lm:TaylorBehavesSubst}.
The ``furthermore'' part is by induction on the single-hole context $C$ s.t.\ $M=C\hole{M'}$, $N=C\hole{N'}$ and $M'\rightarrow_\mathrm{base} N'$.
\end{proof}

The following technical lemma is an adaptation of~\cite[Theorem 20]{DBLP:journals/tcs/EhrhardR08}.

\begin{Lemma}\label{lm:ForInjectivity}
Let $P,Q$ be $\lamu$-terms, $p,p'\in\Te{P}$ and $[\vec{d}],[\vec{d}\,']\in\,\fmsets{\Te{Q}}$.
Then $p=p'$ and $[\vec{d}]=[\vec{d}\,']$ are entailed by any of the following three\footnote{
Remark that point 3. (used in the proof of Theorem \ref{thm:Injectivity}) is \emph{not} an inductive step of point 2., simply because $\name{\eta}{p}\notin\lam\mu^{\mathrm{r}}$.
Therefore we treat is separately.
This is due to the fact that we are in $\lam\mu$-calculus and not in Saurin's $\Lambda\mu$-calculus.
} conditions:
 \begin{enumerate}
  \item if $\,p\lsubst{x}{\vec{d}}\cap p'\lsubst{x}{\vec{d}\,'}\neq\emptyset$
  \item if $\,\lnamedapp{p}{\gamma}{[\vec{d}]}\cap \lnamedapp{p'}{\gamma}{[\vec{d}\,']}\neq\emptyset$
  \item if $\lnamedapp{\name{\eta}{p}}{\gamma}{[\vec{d}]}\cap \lnamedapp{\name{\eta}{p'}}{\gamma}{[\vec{d}\,']}\neq\emptyset$.
 \end{enumerate}
\end{Lemma}

The following ``non-interference property'' (Theorem \ref{thm:Injectivity}) was first proved by Ehrhard and Regnier in \cite[Theorem 22]{DBLP:journals/tcs/EhrhardR08} for the $\lam$-calculus. It is known that it fails in MELL. A natural question, to which we do not have an answer yet, is what is the threshold, between $\lam$-calculus and MELL, where this property starts failing.
It is important also because somehow it is linked to the possibility of defining a \emph{coherence} on resource terms for which Taylor expansion is a maximal clique.
We show below that the result holds in $\lamu$-calculus.

\begin{theorem}\label{thm:Injectivity}
If $t,s\in\Te{M}$, $t\neq s$, then $\nf[\mathrm{r}]{t}\cap\nf[\mathrm{r}]{s}=\emptyset$.
\end{theorem}
\begin{proof}
By induction on $\textbf{ms}(t)$ we prove that for all $s\in\lamu^\mathrm{r}$, if $t,s\in\Te{M}$ for some $M\in\lamu$, and if there is $h\in\nf[\mathrm{r}]{t}\cap\nf[\mathrm{r}]{s}$, then $t=s$.

 Case $\textbf{ms}(t)=(1,0,1)$.
Then (Corollary \ref{lamu-cor:MSMeasure}) $t$ is a variable, thus $M$ is the same variable and therefore $s=t$.

 Case $\textbf{ms}(t)>(1,0,1)$.
By Lemma \ref{lamu-lm:writingM}, $M$ has shape:
\[
	M=\lam\vec{x}_1\mu\alpha_1.\name{\beta_1}{\dots\lam\vec{x}_{k}\mu\alpha_{k}.\name{\beta_k}{RQ_1\dots Q_n}}
\] for $R$ either a variable, or a $\lam$-redex or a $\mu$-redex. 
Since the series of $\lam$ and $\mu$ abstractions (with their respective namings) will play no role in the following, in this proof we shorten  $\lam\vec{x}_1\mu\alpha_1.\name{\beta_1}{\dots\lam\vec{x}_{k}\mu\alpha_{k}.\name{\beta_k}{\dots}}$ to just $\vec{\lamu}\name{}{\dots}$.
So:
$
 t=\vec{\lamu}\name{}{t'[\vec{u}^{\,1}]\dots[\vec{u}^{\,n}]}
$ and 
$
 s=\vec{\lamu}\name{}{s'[\vec{v}^{\,1}]\dots[\vec{v}^{\,n}]}
$
for $t',s'\in\Te{R}$ and $[\vec{u}^{\,i}],[\vec{v}^{\,i}]\in\,\fmsets{\Te{Q_i}}$.
We have now three subcases depending on the shape of $R$.

Subcase $R$ variable, say $R=x$.
Then $t'=s'=x$.
\emph{W.l.o.g.} $n\geq1$, otherwise it is trivial that $t=s$.
Now say $[\vec{u}^{\,i}]=:[u^{\,i}_{1},\dots,u^{\,i}_{m_i}]$ and $[\vec{v}^{\,i}]=:[v_{1}^{\,i},\dots,v_{m'_i}^{\,i}]$ for $i=1,\dots,n$. 
By confluence we have
$
 h \in \nf[\mathrm{\mathrm{r}}]{\vec{\lamu}\name{}{x \,\nf[\mathrm{r}]{[\vec{u}^{\,1}]}\dots\nf[\mathrm{r}]{[\vec{u}^{\,n}]}}}
$,
so 
$
h\in \nf[\mathrm{\mathrm{r}}]{\vec{\lamu}\name{}{x [\vec{d}^{\,1}]\dots[\vec{d}^{\,n}]}}
$ for some $d_{j}^{\,i}\in\nf[\mathrm{r}]{u_{j}^{\,i}}$.
Similarly, we get:
$
 h \in
\nf[\mathrm{\mathrm{r}}]{\vec{\lamu}\name{}{x [\vec{d}^{'\,1}]\dots[\vec{d}^{'\,n}]}}
$ for some $d_{j}^{'\,i}\in\nf[\mathrm{r}]{v_{j}^{\,i}}$.
So it must be $m_i=m'_i$ ($i=1,\dots,n$) and:
\[
 h=\vec{\lamu}'\name{}{x [d_{1}^{1},\dots,d_{m_1}^{1}]\cdots[d_{1}^{n},\dots,d_{m_n}^{n}]}
\]
for some head $\vec{\lamu}'$, some $d_{j}^{i}\in\nf[\mathrm{r}]{u_{j}^{i}}\cap\nf[\mathrm{r}]{v_{\sigma_i(j)}^{i}}$ and permutations $\sigma_i$ on $m_i$ elements.
But $u_{j}^{i},v_{j}^{i}\in\Te{Q_i}$ and $\textbf{ms}(u_{j}^{i})<\textbf{ms}(t)$ since $u_{j}^{i}$ is a strict subterm of $t$.
So we can apply the inductive hypothesis to each $u_{j}^{i}$ and obtain $u_{j}^{i}=v_{\sigma_i(j)}^{i}$.
Hence, $t=s$.

Subcase $R=(\lam y.P)N$.
It is the same argument as the following subcase, so we skip it.

Subcase $R=(\mu\gamma.\name{\eta}{P})N$.
Then $t'=(\mu\gamma.\name{\eta}{p})[\vec{d}]$ and $s'=(\mu\gamma.\name{\eta}{p'})[\vec{d}\,']$ for $p,p'\in\Te{P}$ and $[\vec{d}],[\vec{d}\,']\in\,\fmsets{\Te{N}}$.
By confluence on $\lamu^\mathrm{r}$ we have:
\[
 \nf[\mathrm{r}] t=\nf[\mathrm{r}]{\vec{\lamu}\name{}{(\mu\gamma.\lnamedapp{\name{\eta}{p}}{\gamma}{[\vec{d}]})[\vec{u}^{\,1}]\dots[\vec{u}^{\,n}]}}
.\]
So there is $h_1\in \mu\gamma.\lnamedapp{\name{\eta}{p}}{\gamma}{[\vec{d}]}$ s.t.\ $h\in\nf[\mathrm{r}]{\widetilde{h}_1}$
where:
$
 \widetilde{h}_1:=\vec{\lamu}\name{}{h_1[\vec{u}^{\,1}]\dots[\vec{u}^{\,n}]}.
$
Analogously we find a $h_2\in\mu\gamma.\lnamedapp{\name{\eta}{p'}}{\gamma}{[\vec{d}\,']}$ s.t.\ $h\in\nf[\mathrm{r}]{\widetilde{h}_2}$, where:
$
 \widetilde{h}_2:=\vec{\lamu}\name{}{h_2[\vec{v}^{\,1}]\dots[\vec{v}^{\,n}]}.
$
By Lemma \ref{lm:TaylorBehavesSubst} we have $h_1,h_2\in\Te{\mu\gamma.\namedapp{\name{\eta}{P}}{\gamma}{N}}$ and so
$
 \widetilde{h}_1,\widetilde{h}_2$ belong to $\Te{\vec{\lamu}\name{}{(\mu\gamma.\namedapp{\name{\eta}{P}}{\gamma}{N})Q_1\cdots Q_n}}.
$
This and the fact that $h\in\nf[\mathrm{r}]{\widetilde{h}_1}\cap\nf[\mathrm{r}]{\widetilde{h}_2}$ mean that $\widetilde{h}_1$ satisfies both the hypotheses of the inductive hypothesis.
Moreover, since $t'\to_{\mu^\mathrm{r}} h_1+\Sum T$ for some sum $\Sum T$, then $\ms{h_1}<\ms{t'}$. And since the number of $\mu$'s is constant under $\mu$-reduction, $\dg{\mu}{t'}=\dg{\mu}{h_1}$. Therefore we can apply Lemma \ref{lamu-lm:mOnCT}(2) and obtain:
$
 \ms{\widetilde{h}_1}
 =
 \ms{\vec{\lamu}\name{}{h_1[\vec{u}^{\,1}]\dots[\vec{u}^{\,n}]}}
 <
 \ms{\vec{\lamu}\name{}{t'\,[\vec{u}^{\,1}]\dots[\vec{u}^{\,n}]}}$
 $=
 \ms{t}.
$
So $\textbf{ms}(\widetilde{h}_1)<\textbf{ms}(t)$ and we can safely apply the inductive hypothesis obtaining $\widetilde{h}_1=\widetilde{h}_2$.
Looking at the definition of $\widetilde{h}_1,\widetilde{h}_2$, we get $h_1=h_2$ as well as $[\vec{u}^{\,i}]=[\vec{v}^{\,i}]$ ($i=1,\dots,n$).
But now we have:
\[
\mu\gamma.\lnamedapp{\name{\eta}{p}}{\gamma}{[\vec{d}]}\ni h_1=h_2\in \mu\gamma.\lnamedapp{\name{\eta}{p}}{\gamma}{[\vec{d}\,']}
\]
so Lemma \ref{lm:ForInjectivity} gives $p=p'$ and $[\vec{d}]=[\vec{d}\,']$, i.e. $t'=s'$.
If we look at the shape of $t,s$, this last information together with $[\vec{u}^{\,1}]=[\vec{v}^{\,1}],\dots,[\vec{u}^{\,n}]=[\vec{v}^{\,n}]$, precisely means $t=s$.\qedhere
\end{proof}

We conclude with a useful property (Corollary~\ref{cor:Ass3+Prop+Cor}).
It follows from the following proposition, which in turn easily follows by Lemma~\ref{lm:TaylorBehavesSubst}.

\begin{Proposition}\label{lamu-prop:FIRSTPART}
If $\Te{M}\ni t\rightarrow_{\mathrm{base}} \Sum{T}'$ then there is $N\in\lamu$ s.t.\ $\Sum{T}'\subseteq\Te{N}$ and $M\to N$.
\end{Proposition}

\begin{Corollary}\label{cor:Ass3+Prop+Cor}
For all $\Sum{T}\subseteq\Te{M}$, there exist $N\in\lamu$ s.t.\ $M\msto N$ and $\nf[\mathrm{r}]{\Sum{T}}\subseteq\Te{N}$.
\end{Corollary}
\begin{proof}[Proof sketch]
One first generalises Proposition~\ref{lamu-prop:FIRSTPART} to sums (instead of a term $t$ in the statement); then, we prove the desired result by induction on the length of a maximal reduction $\Sum{T}\msto[\mathrm{r}]\nf[\mathrm{r}]{\Sum T}$. 
\end{proof}

\subsection{The $\lamu$-theory $\NFTeq$}

Mimicking the definitions for $\lam$-calculus we say that:
\begin{Definition}
\begin{enumerate}
\item An equivalence $\mathcal{R}$ on $\lamu$ is a \emph{congruence} iff $\mathcal{R}$ is contextual.
\item A congruence $\mathcal{R}$ is a \emph{$\lamu$-theory} iff $\mathcal{R}\supseteq\,=_{\lamu\rho}$.
\item The \emph{term algebra} of a $\lamu$-theory $\mathcal{R}$ is the quotient $\lamu/_{\mathcal{R}}$. A $\lamu$-theory $\mathcal{R}$ is \emph{non-trivial} iff $\lamu/_{\mathcal{R}}\neq\set{*}$.
\end{enumerate}
\end{Definition}

It is clear that $=_{\lamu\rho}$ is a $\lamu$-theory. Now fix the equivalence $M \NFTeq N$ iff $\NFT{M}=\NFT{N}$.
Actually, $\NFTeq$ is a non-trivial $\lamu$-theory.
In fact, the contextuality follows immediately from the Theorem \ref{thm: Monotonicity}; the fact that it contains $=_{\lamu\rho}$ easily follows from confluence and Proposition \ref{prop:LiftAssSimulation}; and it is clearly non-trivial: $\lam x.x\not\NFTeq\emptyset\NFTeq(\lam x.xx)(\lam x.xx)=:\Omega$.

In  $\lam$-calculus, $\NFTeq$ is the $\lam$-theory equating B\"ohm trees. In particular, it is sensible (i.e. it equates all unsolvables). We will see (Corollary \ref{lamu-cor:sensible!}) that in our case it is still sensible.

\begin{Definition}\label{lamu-def:hnf-ToHaveAhnf}
 A $\lamu$-term $M$ is a \emph{head normal form}\index{Term!Head normal-!$\lamu$-} (\emph{hnf} for short) iff there are no $\rho$-redexes in its head (remember Lemma~\ref{lamu-lm:writingM}) and it has a head variable.
 We define the exact same notion for $\lamu^\mathrm{r}$.
\end{Definition}

\begin{Definition}\label{lamu-def:headReduction}
 The \emph{head reduction}\index{Reduction!Head- for $\lamu$-term} is the partial function $\mathrm{H}:\lamu\to\lamu$ obtained defining $\mathrm{H}(M)$ via the following algorithm:
 \begin{enumerate}
  \item
  $\rho$-reduce the leftmost $\rho$-redex in the head of $M$, if any
  \item otherwise, $\lamu$-reduce the head redex of $M$, if any
  \item otherwise, $\mathrm{H}(M)$ is not defined.
 \end{enumerate}
 $\mathrm{H}(M)$ is \emph{not} defined iff $M$ is a hnf.
 We say that \emph{head reduction starting on $M$ terminates} iff there is a (necessarily unique) $n\geq0$ s.t.\ $\mathrm{H}^n(M)$ is a hnf.
Here we mean as usual that $\mathrm{H}^0(M):=M$.

 We extend the same definitions to resource terms\index{Reduction!Head- for $\lamu$-resource-term}, and set $\mathrm{H}(t):=\emptyset$ whenever $t$ is a hnf.
Moreover, we set $\mathrm{H}^0(t):=t \in 2\langle \Lamr \rangle$ and, for $n\geq 0$:
 \[\mathrm{H}^{n+1}(t):=\sum\limits_{t_1\in \mathrm{H}(t)} \sum\limits_{t_2\in \mathrm{H}(t_1)} \cdots \sum\limits_{t_{n+1}\in \mathrm{H}(t_n)} t_{n+1} \, \, \, \in 2\langle \Lamr \rangle.\]
We have $\mathrm{H}^{1}(t)=\mathrm{H}(t)$ and $\mathrm{H}^{n+1}(t)=\sum\limits_{t' \in \mathrm{H}(t)} \mathrm{H}^{n}(t')$.
\end{Definition}

\begin{Lemma}\label{lamu-lm:Reviewer}
 If $s$ only contains empty bags (if any) and $s\in\nf[\mathrm{r}]{t}$, then $s\in H^n(t)$ for some~$n\geq 0$.
\end{Lemma}
\begin{proof}[Proof sketch]
If $t$ is a hnf, $s\in\nf[\mathrm{r}]{t}$ entails that $t$ already contains only empty bags, as any eventual bag of $s$ is empty and reductions cannot erase non-empty bags;
but in a hnf the reduction can only take place \emph{inside} some bag, so it must be $s=t$ and we are done.
If $t$ is \emph{not} hnf, by confluence $t \to_{\mathrm{r}} \mathrm{H}(t) \msto[\mathrm{r}] \nf[\mathrm{r}]{t}\ni s$.
So there is a $t_1\in\mathrm{H}(t)$ s.t.\  $s\in\nf[\mathrm{r}]{t_1}$.
Now we reason as in the beginning: if $t_1$ is hnf we are done; if $t_1$ is not, we repeat the argument finding some $t_2$.
By the  well-foundedness  of $\SNm{.}$, we cannot repeat the argument forever and we must end on a hnf, so we conclude.
\end{proof}

Set $\mathrm{H}(\Te{M}):=\bigcup\limits_{t\in\Te{M}} \mathrm{H}(t)\subseteq \lam\mu^\mathrm{r}$.
The following lemma is easy using Definition~\ref{lamu-def:hnf-ToHaveAhnf} and Lemma~\ref{lm:TaylorBehavesSubst}.

\begin{Lemma}\label{lamu-lm:TaylorCommutesHead}
 If $M\in\lamu$ with $\mathrm{H}(M)$ defined, we have: \[\Te{\mathrm{H}(M)}=\mathrm{H}(\Te{M}).\]
\end{Lemma}

The following proposition shows that $\lamu$-calculus enjoys a notion of \emph{solvability} analogue to the one of $\lam$-calculus.

\begin{Proposition}\label{lamu-prop:solvability!}
 For $M\in\lamu$, the following are equivalent:
 \begin{enumerate}
  \item $M=_{\lamu\rho} H$ with $H$ hnf
  \item Head reduction starting on $M$ terminates
  \item $\NFT{M}\neq\emptyset$.
 \end{enumerate}
We call $M\in\lamu$ \emph{solvable}\index{Term!Solvable-!$\lamu$-} iff it satisfies any of the previous equivalent conditions. Otherwise, $M$ is called \emph{unsolvable}.
\end{Proposition}
\begin{proof}
(1$\Rightarrow$2).
  By confluence $M$ and $H$ have a common $\lamu\rho$-redex $M_0$. Since $H$ is a hnf, $M_0$ is too. Let $s_0$ be the unique resource $\lamu$-term in $\Te{M_0}$ s.t.\ all its bags (if any) are empty. This term clearly exists. Note that, by construction, $s_0$ is $\mathrm{r}$-normal. By repeatedly applying Proposition~\ref{prop:LiftAssSimulation} one can check that we obtain an $s\in\Te{M}$ s.t.\ $s_0\in\nf[\mathrm{r}]{s}$. Now, by Lemma \ref{lamu-lm:Reviewer}, $s_0\in\mathrm{H}^n(s)$ for some $n\geq 0$. By repeatedly applying Lemma \ref{lamu-lm:TaylorCommutesHead}, we find that $s_0\in\mathrm{H}^n(\Te{M})=\Te{\mathrm{H}^n(M)}$.
Finally, since $s_0$ is a hnf, so it must be $\mathrm{H}^n(M)$.

(2$\Rightarrow$3).
Easy.

(3$\Rightarrow$1).
If $\NFT{M}\neq 0$ there is $t\in\Te{M}$ s.t.\ $\nf[\mathrm{r}]{t} \neq 0$.
By Corollary~\ref{cor:Ass3+Prop+Cor}, $M\msto N$ for some $N\in\lamu$ s.t.\ $\nf[\mathrm{r}]{t}\subseteq\Te{N}$. So $\Te{N}$ contains at least a hnf, and thus $N$ must be a hnf too.\qedhere
\end{proof}

\begin{Corollary}\label{lamu-cor:sensible!}
 The $\lamu$-theory $\NFTeq$ is \emph{sensible} (that is, it equates all unsolvable terms).
\end{Corollary}

\section{Applying the approximation theory}

\subsection{Stability}

The Stability Property gives sufficient conditions for a context to commute with intersections in $\lamu/_{\NFTeq}$, i.e. (the intersections are defined below):
\[C\hole{\,\bigcap_{i_1}{N_{i_1}} \ ,\dots,\bigcap_{i_n}{N_{i_n}}}\NFTeq\bigcap_{i_1\dots,i_n} C\hole{N_{i_1},\dots,N_{i_n}}.\]
Given a non-empty subset $\cX\subseteq\lamu$, call its \emph{$\mathcal{T}$-infimum} the set $\bigcap\cX := \bigcap\limits_{M\in\cX} \NFT{M}\subseteq\lamu^\mathrm{r}$.
We say that $\cX$ is \emph{bounded} iff there exists an $L\in\lamu$ such that $M\leq L$ for all $M\in\cX$.
Write $M\NFTeq\bigcap\cX$ instead of $\NFT{M}=\bigcap\cX$.
Observe that (in case it exists) an $M$ s.t.\ $M\NFTeq\bigcap\cX$ need \emph{not} to be unique, so $\bigcap\cX$ does not identify a unique $\lam$-term.

\begin{theorem}[Stability]\label{lamu-th:TeStability}
Let $C$ be an $n$-ary $\lamu$-context and fix non empty bounded $\cX_1,\dots,\cX_n\subseteq\lamu^\mathrm{r}$.
For all $M_1,\dots,M_n\in\lamu$ s.t.\ $M_i\NFTeq\bigcap{\cX_i}$ ($i=1,\dots,n$) we have:
\[
C\hole{M_1,\dots,M_n}\NFTeq\bigcap_{\underset{N_n\in\cX_n}{\underset{\dots}{N_1\in\cX_1}}} C\hole{N_1,\dots,N_n}.
\]
\end{theorem}
\begin{proof}
 Non-trivial, but exactly as done in \cite{DBLP:journals/pacmpl/BarbarossaM20} for $\lam$-calculus (using Theorem \ref{thm:Injectivity}).
\end{proof}

Using the usual encoding of booleans and pairs ($\prog{True}:=\lam x y.x$, $\prog{False}:=\lam x y.y$, $\langle M,N \rangle:= \lam z.zMN$) we have the non-implementability of the following parallel-or.

\begin{Corollary}\label{lamu-cor:Por}
There is \emph{no} $\prog{Por}\in\lamu$ s.t.\ for all $M,N\in\lamu$,
\[
\left\{
\begin{array}{rlll}
\prog{Por}\langle M,N \rangle & \NFTeq & \prog{True} & \mathrm{if \ }M\not\NFTeq\Omega\mathrm{ \ or \ }N\not\NFTeq\Omega \\
\prog{Por}\langle M,N \rangle & \NFTeq & \Omega & \mathrm{if \ }M\NFTeq N\NFTeq\Omega.
\end{array}
\right.
\]
\end{Corollary}
\begin{proof}
Otherwise, for the context $C:=\prog{Por}\,\xi$, by Theorem \ref{lamu-th:TeStability} we would have the contradiction:
$\prog{True} \NFTeq C\hole{\langle \prog{True},\Omega \rangle} \cap C\hole{\langle \Omega,\prog{True} \rangle}
 \NFTeq C\hole{\langle \prog{True},\Omega \rangle \cap \langle \Omega,\prog{True} \rangle} 
 \NFTeq C\hole{\langle \Omega,\Omega \rangle}
 \NFTeq \Omega$.
\end{proof}

\subsection{The perpendicular Lines Property}

The perpendicular lines Property (PLP for short) states that, fixed a term $\lam z_1\dots z_n.F\in\lamu$, if the function $\vec{M}\in\lamu^n/_{\NFTeq}\to(\lam\vec{z}.F)\vec{M}\in\lamu/_{\NFTeq}$ is constant on $n$ ``perpendicular lines'' (in the sense of the statement, Theorem~\ref{lamu-thm:PLL}), then it is constant everywhere.
Lemma \ref{lamu-lm:ClaimPLL} below is the crucial ingredient for the proof of PLP, and we use in it all the strong properties of resource approximation (linearity, SN and confluence).

\begin{Lemma}\label{lamu-lm:ClaimPLL} 
Fix $\vec{z}:=z_1,\dots z_n$ distinct variables and let $t\in\lamu^\mathrm{r}$. Suppose that:
\begin{enumerate}
	\item[i.] $\nf[\mathrm{r}]{t}\neq 0$
	\item[ii.] there is $F\in\lamu$ s.t.\ $t\in\Te{F}$
	\item[iii.] there are $\set{M_{ij}}_{1\leq i\neq j\leq n}\subseteq\lamu$ s.t.\ the function mapping $\vec{M}\in\lamu^n/_{\NFTeq}$ to $(\lam\vec{z}. F)\vec{M}\in\lamu/_{\NFTeq}$ is constant on the following ``perpendicular lines'' of $\lamu^n/_{\NFTeq}$:
	\begin{equation}\label{lamu-eq:Claim}
	\begin{array}{l}
	\mathit{l}_1= \set{(Z, \ \ M_{12}, \ \dots\dots, \ M_{1 n}) \mid Z\in\lamu}\\
	\mathit{l}_2= \set{(M_{21}, \ Z, \ \ \dots\dots, \ M_{2 n}) \mid Z\in\lamu}\\
	\phantom{ \ \ \ \ \ \ (\lam\vec{z}.F)M_{21} \dots} \ddots\ \!\\
	\mathit{l}_n= \set{(M_{n1}, \ \dots, \ M_{n (n-1)}, \ Z) \mid Z\in\lamu}.
	\end{array}
\end{equation}
\end{enumerate}
Then $\dg{z_1}{t}=\dots=\dg{z_n}{t}=0$.
\end{Lemma}
\begin{proof}
Induction on the size $\textbf{ms}(t)$ of $t\in\lamu^\mathrm{r}$.

- Case $\textbf{ms}(t)=(1,0,1)$.
Then $t$ is a variable (Corollary \ref{lamu-cor:MSMeasure}).
If $t=z_i$ for some $i$ then the $i$-th line of (\ref{lamu-eq:Claim}) gives the contradiction:
\[N_i\NFTeq (\lam\vec{z}.z_i)M_{i1}\cdots M_{i(i-1)}ZM_{i(i+1)}\cdots M_{in}=_\tau Z\] for all $Z\in\lamu$.
Hence, it must be $\dg{z_1}{t}=\dots=\dg{z_n}{t}~=~0$.

- Case $\textbf{ms}(t)>(1,0,1)$.
By $(i)$ there is $u\in\nf[\mathrm{r}]{t}$.
As $u$ is normal, it has shape:
$
 u=\vec{\lamu}\name{}{y[\vec{u}^{\,1}]\dots[\vec{u}^{\,m}]}
$
for some $m\geq 0$, some variable $y$, some normal bags $[\vec{u}^{\,j}]$, and where we have shorten, as before, a series $\lam\vec{x}_1\mu\alpha_1.\name{\beta_1}{\dots\lam\vec{x}_{k}\mu\alpha_{k}.\name{\beta_k}{\dots}}$ of $\lam$ and $\mu$ abstraction by just $\vec{\lamu}\name{}{\dots}$.
By $(ii)$ $t\in\Te{F}$, so that by Corollary~\ref{cor:Ass3+Prop+Cor} there is $Q\in\lamu$ s.t.\ $F\msto[] Q$ and $u\in\Te{Q}$.
So $Q$ must have shape:
$
 Q=\vec{\lamu}\name{}{yQ_1\cdots Q_m}
$
for some $Q_j$'s in $\lamu$ s.t.\ $[\vec{u}_j]\in\,\fmsets{\Te{Q_j}}$ for all $j=1,\dots ,m$.
Now there are two possibilities: either $y=z_i$ for some $i=1,\dots ,n$, either $y\neq z_i$ for all $i$.

Suppose $y=z_i$.
Then, for $\vec{q}:=q_1,\dots,q_m$ fresh variables, we can chose $Z:=\lam\vec{q}.\prog{True}\in\lamu$ (or $Z:=\prog{True}$ if $m=0$) in the $i$-th line $\mathit{l}_i$ of (\ref{lamu-eq:Claim}), and since by $(iii)$ $\lam\vec{z}.F$ is constant (mod $\NFTeq$) on $\mathit{l}_i$, we can compute its value as:
\[\small\begin{array}{rl}
 & (\lam\vec{z}.F)M_{i1}\cdots M_{i(i-1)}\,(\lam\vec{q}.\prog{True})\,M_{i(i+1)}\cdots M_{in} \\
 \NFTeq &
 Q\set{M_{i1}/z_1,\dots,(\lam\vec{q}.\prog{True})/z_i,\dots,M_{in}/z_n}
 \\
 \NFTeq &
 \vec{\lamu}\name{}{(\lam\vec{q}.\prog{True})\widetilde{Q_{i1}} \cdots \widetilde{Q_{im}}}
 \\
 \NFTeq &
 \vec{\lamu}\name{}{\prog{True}}
\end{array}\]
where we set $\widetilde{Q_{ij}}:=Q_j\set{M_{i1}/z_1,\dots,(\lam\vec{q}.\prog{True})/z_i,\dots,M_{in}/z_n}$. The first equality holds because $F\msto[] Q$ and $\NFTeq$ is finer than $=_{\lamu\rho}$, and the second equality holds because $y=z_i$.
In the same way, choosing $Z:=\lam\vec{q}.\prog{False}\in\lamu$ in $\mathit{l}_i$ we find that the value (mod $\NFTeq$) of $\lam\vec{z}.F$ on $\mathit{l}_i$ is $\vec{\lamu}\name{}{\prog{False}}$.
But this is impossible because $\prog{True}\not\NFTeq \prog{False}$.

Therefore, it must be $y\neq z_i$ for all $i$.
Note that \emph{w.l.o.g.} $m\geq 1$ (indeed if $m=0$, from the fact that $y\neq z_i$ for all $i$ we already get $\dg{z_i}{u}=0$ and, as $u\in\nf[\mathrm{r}]{t}$ and in $\lamu^\mathrm{r}$ one cannot erase non-empty bags, we are done).
Now fix $i\in\set{1,\dots,n}$ and $Z',Z''\in\lamu$.
Similarly as before, choosing $Z:=Z'$ in $\mathit{l}_i$ and using what we found so far, putting $Q'_{ij}:=Q_j\set{M_{i1}/z_1,\dots,Z'/z_i,\dots,M_{in}/z_n}$, since $\lam\vec{z}.F$ is constant (mod $\NFTeq$) on $\mathit{l}_i$, we can compute its value as:
\[\begin{array}{rl}
 & (\lam\vec{z}.F)M_{i1}\dots M_{i(i-1)}Z'M_{i(i+1)}\dots M_{in} \\
 \NFTeq &
 Q\set{M_{i1}/z_1,\dots,Z'/z_i,\dots,M_{in}/z_n} 
 \\
 \NFTeq &
 \vec{\lamu}\name{}{yQ'_{i1} \dots Q'_{im}}
\end{array}\]
where the last equality holds since $y$ is \emph{not} one of the $z_i$'s.
Choosing $Z''$ instead of $Z'$ and putting $Q''_{ij}$ the same as $Q'_{ij}$ but with $Z''$ instead of $Z'$, one has that the value (mod $\NFTeq$) of $\lam\vec{z}.F$ on $\mathit{l}_i$ is $\vec{\lamu}\name{}{yQ''_{i1} \dots Q''_{im}}$.
So we have $ \vec{\lamu}\name{}{yQ'_{i1} \dots Q'_{im}}= \vec{\lamu}\name{}{yQ''_{i1} \dots Q''_{im}}$, which easily entails:
$Q'_{i1}\NFTeq Q''_{i1},
 \,
 \dots,
 \, Q'_{im}\NFTeq Q''_{im}$.
But by construction we have: 
\[\begin{array}{c}
Q'_{ij}\NFTeq (\lam\vec{z}. Q_j)M_{i1}\dots M_{i(i-1)}Z'M_{i(i+1)}\dots M_{in}
\\
Q''_{ij}\NFTeq (\lam\vec{z}. Q_j)M_{i1}\dots M_{i(i-1)}Z''M_{i(i+1)}\dots M_{in}.
\end{array}\]
So if we remember that $Z',Z''$ were generic in $\lamu$, the previous equalities $Q'_{ij}=Q''_{ij}$ precisely say that $\lam\vec{z}. Q_j$ is constant on the line $\mathit{l}_i$.
And since this holds for all $i=1,\dots,n$, we have just found that $\lam\vec{z}.Q_j$ satisfies $(iii)$. And since we have equalities $Q'_{ij}=Q''_{ij}$ for all $j=1,\dots,m$, we have that each $\lam\vec{z}.Q_1,\dots,\lam\vec{z}.Q_k$ satisfies $(iii)$.
We can now comfortably apply the induction hypothesis on any $s\in [\vec{u}^{\,j}]$. In fact, as $[\vec{u}^{\,j}]$ is normal, $\nf[\mathrm{r}]{s}\neq 0$, i.e.\ $(i)$; as $[\vec{u}^{\,j}]\in\,\fmsets{\Te{Q_j}}$, we have $s\in\Te{Q_j}$, i.e.\ $(ii)$; and we just found that $\lam\vec{z}.Q_j$ satisfies $(iii)$; finally, $s$ is a strict subterm of $u\in\nf[\mathrm{r}]t$, thus (Corollary \ref{lamu-cor:MSMeasure}) $\textbf{ms}(s)<\textbf{ms}(u)\leq\textbf{ms}(t)$.
Therefore, the inductive hypothesis gives $\dg{z_1}{s}=\dots=\dg{z_n}{s}=0$. 
Since this is true for all $s$ in all $[\vec{u}^{\,j}]$, $j=1,\dots,m$, we get $\dg{z_1}{u}=\dots=\dg{z_n}{u}=0$.
And now $u\in\nf[\mathrm{r}]{t}$ entails $\dg{z_1}{t}=\dots=\dg{z_n}{t}=0$.\qedhere
\end{proof}

\begin{theorem}[Perpendicular Lines Property\index{Property!Perpendicular Lines-!-for $\lamu/_{\NFTeq}$}]\label{lamu-thm:PLL}
Suppose that for some fixed $\set{M_{ij}}_{1\leq i\neq j\leq n}$, $\set{N_{i}}_{1\leq i\leq n}\subseteq\lamu$, the system of equations:
\[
	\left\{\begin{array}{llcr}
	(\lam z_1\dots z_n.F) \ \ Z \  M_{12} \ \dots\dots \ M_{1 n} & =_\tau & N_{1} \\
	(\lam z_1\dots z_n.F) \ M_{21} \ Z \ \ \dots\dots \ M_{2 n} & =_\tau & N_{2} \\
	\phantom{(\lam z_1\dots z_n.F)M_{21} \ \ \ \dots} \ \ddots&\ \vdots& \!\\
	(\lam z_1\dots z_n.F) \ M_{n1} \ \dots \ M_{n (n-1)} \ Z & =_\tau & N_{n}
	\end{array}\right.
\]
holds for all $Z\in\lamu$.
Then: \[(\lam z_1\dots z_n.F)Z_1\dots Z_n =_\tau N_{1}\]
for all $Z_1,\dots,Z_n\in\lamu$.
\end{theorem}
\begin{proof}
It follows from Lemma \ref{lamu-lm:ClaimPLL} as done in \cite{DBLP:journals/pacmpl/BarbarossaM20}.
\end{proof}

\begin{Corollary}\label{lamu-cor:Por'}
 There is \emph{no} $\prog{Por}'\in\lamu$ s.t.\ for all $Z\in\lamu$ one has at the same time $\prog{Por}'\,\prog{True}\,Z \NFTeq\prog{True}$,
 $\prog{Por}'\,Z\,\prog{True} \NFTeq\prog{True}$,
 $\prog{Por}'\,\prog{False}\,\prog{False} \NFTeq\prog{False}$.
\end{Corollary}

The non existence of parallelism in $\lamu$-calculus is known as folklore via arguments involving stable models: here we proved it solely via Taylor expansion.

\section{Conclusions and Future Works}
In \cite{Laurent04anote} Laurent studies the mathematics of (untyped) $\lamu$-calculus via its denotational semantics;
this paper does it by developing a theory of program approximation
based on Linear Logic resources.
In particular, we proved that the approximation theory satisfies strong normalisation and confluence (non-trivial results in this setting), the Monotonicity Property, the Non-Interference Property, that it induces a sensible $\lam\mu$-theory, and that it can be used as a tool in order to obtain the Stability Property and the Perpendicular Lines Property, and thus the impossibility of parallel computations in the language.
A first natural question immediately arises:

1- Does Taylor expansion allow to find interesting properties \emph{not} satisfied by $\lambda$-calculus, but that are enjoyed by the $\lamu$-calculus due to the presence of \prog{callcc}?

For future investigations, we believe that it would be interesting to integrate this approach with the \emph{differential} extension of $\lamu$-calculus defined in~\cite{DBLP:journals/tcs/Vaux07}, via the mentioned translation $\de{(\cdot)}$, in order to explore quantitative properties.

The following two questions are maybe the most significant ones:

2- Does it makes sense to introduce B\"ohm trees for the $\lamu$-calculus? 
For instance, for the call-by-value $\lam$-calculus, the Taylor expansion has provided in \cite{DBLP:journals/corr/abs-1809-02659} invaluable guidance for finding a meaningful notion of trees satisfying Ehrhard and Regnier's \emph{commutation formula}; the same methodology could maybe be applied here.
However, in \cite{DBLP:journals/jsyml/DavidP01} it is shown that $\lamu$-calculus does \emph{not} enjoy B\"ohm's separation property. 
David and Py's counterexample could hence be an indication that, instead, B\"ohm trees are not a ``good'' notion for $\lamu$-calculus.
The best way of proceeding would be, in that case, to consider the natural extension of $\lamu$-calculus given by Saurin's $\Lam\mu$-calculus \cite{SAURIN2012106}.
It was introduced precisely to satisfy B\"ohm's property and, as a matter of fact, in~\cite{SAURIN2012106} Saurin proposes a definition of B\"ohm trees for his $\Lambda\mu$-calculus.

3- Does $\Lam\mu$-calculus enjoys the same approximation theory developed in this paper? 
On one hand, many constructions we did in this chapter seem possible also in Saurin's calculus, on the other hand we used the fact that the number of $\mu$'s in a term is the same as the named subterms, e.g.\, in Remark~\ref{rm:depthDegree}.
In general, one could wonder which one, between $\lamu$ and $\Lam\mu$, should be the ``canonical lambda-mu-calculus'': from a proof-theoretical perspective $\lamu$-calculus precisely corresponds to Parigot's CD-derivations, but $\Lam\mu$-calculus satisfies more desirable properties (B\"ohm separation).
Moreover, in~\cite{DBLP:conf/flops/Saurin10}, Saurin adapts usual techniques of $\lam$-calculus to $\Lam\mu$-calculus: he studies the notion of solvability, proves a standardization theorem and studies more in detail the notion of B\"ohm trees.
A very interesting future direction of research would be, hence, to develop our theory of resource approximation for Saurin's calculus, and study its relation with his theory of B\"ohm approximation.
In any case, we look at the fact that the Taylor expansion works so nicely in $\lamu$-calculus -- and this regardless of a notion of B\"ohm trees -- as a \emph{a posteriori} confirmation of the high power of this form of approximation.

There are at least three other interesting points in relation with strictly related areas:

3- The $\lamu$-calculus can be translated in the $\lam$-calculus via the CPS-translations (see e.g. de Groote's one in~\cite{DBLP:conf/caap/Groote94}).
What does our theory of approximation correspond to under this translation?

4- In order to perform a deeper logical analysis, one should consider translations into Linear Logic. 
It is known from \cite{DBLP:journals/tcs/Laurent03} that $\lamu$-calculus translates into polarized proof nets.
Taylor expansion for proof-nets is possible, but the construction can be complex: in fact one of the interests in \emph{directly} defining a Taylor expansion for a certain ``$\lam$-calculus style'' programming language (as we did for $\lamu$-calculus, and as one does for $\lam$-calculus) is precisely to avoid that complexity. 
In our case we have just shown that, at the end of the day, the theory of resource approximation for $\lamu$-calculus can be developed with essentially the same methodology as in $\lam$-calculus. This leads to asking what makes a Taylor expansion ``easy'', and should be considered in relation to the already mentioned possibility of the existence of a coherence relation for which $\Te{M}$ is a clique. 
This motivates an investigation of the complexity of the definition of a Taylor expansion of a programming language/proof system, which may be related to the notion of connectedness of proof-nets, whose study starts in~\cite{DBLP:conf/rta/GuerrieriPF16}. 
Such question should be considered in relation with the so-called problem of the ``inversion of Taylor expansion''~\cite{DBLP:conf/wollic/GuerrieriPF19,DBLP:conf/csl/GuerrieriPF20} and the problem of ``injectivity'' of denotational models (in particular, the relational one) for Linear Logic.

5- The $\lamu$-calculus is not the only way of extending the Curry-Howard correspondence to classical logic. 
Another notable one is the already mentioned \emph{Krivine's classical realizability}, which is a ``machine to extract computational content from proofs + axioms'' (for almost all mathematics, such as the one formalizable in ZF+AC, see \cite{krivine2020program}).
There are translations between $\lamu$-calculus and Krivine's calculus, and \emph{vice-versa}.
What does our work say about Krivine's realisability?

\begin{acks}
Great help was given by Giulio Manzonetto and Lorenzo Tortora de Falco;
crucial remarks were given by the anonymous reviewers;
we had useful discussions with Olivier Laurent, Lionel Vaux and the DIAPASoN project.
Thanks to all of them, as well as to Federico Olimpieri, who first questioned us about Taylor expansion for $\lamu$-calculus.
\end{acks}

\bibliographystyle{ACM-Reference-Format}
\bibliography{MyBibTeX}

\newpage
\phantom{altra pagina}
\newpage
\appendix

\section{TECHNICAL APPENDIX}

Below, the technical appendix, where we report all the proofs not given in the main paper.
The appendix is organised in five sections (A.1 to A.6).
At the beginning of each of them, we indicate which are the results of the main paper that we are going to prove in it.
Remark that at the very end (after the Appendix A.6) there are some figures to which we refer during the appendices.

\subsection{APPENDIX - SECTION 2.1}

We give proofs of Proposition~\ref{lamu-prop:LamNormalizes!}, Proposition~\ref{lamu-prop:RhoNormalizes!} and Corollary~\ref{lamu-cor:SN}.

\begin{proof}[Proof of Proposition~\ref{lamu-prop:LamNormalizes!}]
 If $t\to_{\lam^\mathrm{r}} t'+\Sum T$ then $t=c\hole{(\lam x.s)b_0}$ and $t'=c\hole{h}$ with $h=s\set{u_{\sigma(1)}/x^{(1)},\dots,u_{\sigma(n)}/x^{(n)}}\in s\langle b_0/x\rangle$, for $c$ a single-hole resource context, $\sigma$ a permutation and $b_0=[u_1,\dots,u_n]$.
 Call $k:=\dg{\mu}{t}=\dg{\mu}{t'}$, and let us use the same notation $B_{t}^c$ as in the proof of Proposition~\ref{lamu-prop:MuNormalizes!}.
 Then we have that $\ms{t'} $ is the bag:
 \[ 
  k 
  - 
  B_{t'}^c
  *
  B_{t'}^s
  *
  [\,d_{t'}(b) \mid b \textit{ in some }u_i\textit{ in }b_0\,]
 \]
 and $\ms{t} $ is the bag:
 \[
  \ms{t} 
  = 
  k 
  - 
  B_{t}^c
  *
  B_{t}^s
  *
  [\,d_{t}(b) \mid b \textit{ in some }u_i\textit{ in }b_0\,]
  *
  [\,d_{t}(b_0)].
 \]
 Now, it is easily understood that if $b$ occurs in $c$, or $b$ occurs in $s$, then $d_{t'}(b) = d_{t}(b)$. 
 Furthermore, for all $b$ occurrence of bag in some $u_{\sigma(i)}$ belonging to $b_0$, one has: 
 \[
  d_{t'}(b)=d_{c}(\xi)+d_{s}(x^{(i)})+d_{u_{\sigma(i)}}(b)\geq d_{c}(\xi)+d_{u_{\sigma(i)}}(b)=d_{t}(b).
 \]
 Thus:
 $
 \ms{t} \geq \ms{t'}*[\,k-d_t(b_0)\,]>\ms{t'}.
 $
\end{proof}

\vspace{0.3cm}

\begin{proof}[Proof of Proposition~\ref{lamu-prop:RhoNormalizes!}]
 If $t\to_{\rho^\mathrm{r}} t'+\Sum T$ then $t=c\hole{h}$ and $t'=c\hole{h'}$, with $h=\mu\gamma.\name{\alpha}{\mu\beta.\name{\eta}{s}}$ and $h'=\mu\gamma.\name{\eta}{s}\set{\alpha/\beta}$ and $c$ a single-hole resource context.
 Therefore:
 \[
  \ms{t'}
  =
  [\,\dg{\mu}{t'}-d_{t'}(b) \mid b \textit{ in }c\,]
  *
  [\,\dg{\mu}{t'}-d_{t'}(b) \mid b \textit{ in }s\,].
 \]
 \[
  \ms{t}
  =
  [\,\dg{\mu}{t}-d_{t}(b) \mid b \textit{ in }c]
  *
  [\,\dg{\mu}{t}-d_{t}(b) \mid b \textit{ in }s].
 \]
 First, remark that $\dg{\mu}{t'}=1+\dg{\mu}{c}+\dg{\mu}{s}$ and $\dg{\mu}{t}=2+\dg{\mu}{c}+\dg{\mu}{s}$.
 
 Also, notice as usual that if $b$ occurs in $c$ then $d_{t'}(b) = d_{c}(b) = d_{t}(b)$.
 Putting these things together we have that, if $c$ \emph{contains at least one bag}:
 \[
  [\,\dg{\mu}{t'}-d_{t'}(b) \mid b \textit{ in }c\,]
  =
  (1+\dg{\mu}{s})+\ms{c}\]
 \[
  <
  (2+\dg{\mu}{s})+\ms{c}
  =
  [\,\dg{\mu}{t}-d_{t}(b) \mid b \textit{ in }c\,].
 \]
 On the other hand, if $c$ does \emph{not} contain any bag, the previous multisets are both empty, thus equal.
 
 Now let's see what happens if $b$ occurs in $s$.
 We have: $d_{t'}(b)=1+d_c(\xi)+d_s(b)$ and $d_{t}(b)=2+d_c(\xi)+d_s(b)$.
 Therefore:
 \[
  [\,\dg{\mu}{t'}-d_{t'}(b) \mid b \textit{ in }s\,]
  =
  (\dg{\mu}{c}-d_c(\xi))+\ms{s}\]
 \[
  =
  [\,\dg{\mu}{t}-d_{t}(b) \mid b \textit{ in }s\,].
 \]
 These facts immediately imply that $\ms{t}\geq\ms{t'}$.
\end{proof}

\vspace{0.3cm}

\begin{proof}[Proof of Corollary~\ref{lamu-cor:SN}]
 The only case in which $\ms{\cdot}$ may remain constant is along a $\rho^{\mathrm{r}}$-reduction, but in this case $\dg{\mu}{t}$ strictly decreases, so $t\to_{\mathrm{r}} t'+\Sum T$ entails $\SNm{t} > \SNm{t'}$.
Now one concludes as usual: if we associate the multiset $[\SNm{t} \stt t\in\Sum T]$ with any sum $\Sum T$ (in particular, a single element sum), it is immediate to see that if $\Sum T\to_\mathrm{r} \Sum S$, then the multiset associated with $\Sum T$ is strictly smaller, in the multiset order, than the one associated with $\Sum S$ (the empty multiset $1$ being associated with the empty sum $0$).
Therefore the well-foundeness of the multiset order (since the order on $\SNm{\cdot}$ is well-founded) gives the non-existence of infinite (non-trivial) reductions.
\end{proof}

\subsection{APPENDIX - SECTION 2.2}

We give complete proofs of Lemma~\ref{lamu-lm:ForLocalConfluence} and Proposition~\ref{lamu-prop:lamu^+Confl}.

In order to prove Lemma~\ref{lamu-lm:ForLocalConfluence}, we need a number of technical lemmas handling the interaction of two successive substitutions.
The proofs of these lemmas are tedious and long inductions and we do not report them here.
In the following, we will use Notation~\ref{notation:delta}.

\begin{Lemma}\label{Applamu-lm:nameCommutes}
Let $t\in\lamu^\mathrm{r}$ and $\alpha,\beta,\gamma,\eta$ names.

\noindent If $\alpha\neq\eta$, $\beta\neq\eta$ and $\beta\neq\gamma$, then $t\set{\alpha/\beta}\set{\gamma/\eta}=t\set{\gamma/\eta}\set{\alpha/\beta}$.
\end{Lemma}

The following lemma says how a renaming behaves with respect to the linear substitution and the linear named application.

\begin{Lemma}\label{Applamu-lm:outNames}
 Let $t\in\lamu^\mathrm{r}$ and $[\vec{u}]$ a bag. 
 Then:
\begin{enumerate}
 \item
  $
   t\langle[\vec{u}]/x\rangle^+ \set{\alpha/\beta}
   =
   t\set{\alpha/\beta}\langle[\vec{u}\set{\alpha/\beta}]/x\rangle^+.
  $
 \item If $\alpha\neq\gamma\neq\beta$ then:

  $
   (\,\coefflnamedapp{t}{\gamma}{[\vec{u}]}\,) \set{\alpha/\beta}
   =
   \coefflnamedapp{t\set{\alpha/\beta}}{\gamma}{[\vec{u}\set{\alpha/\beta}]}.
  $
 \item If $\alpha\neq\gamma\neq\beta$ then:

  $
   (\,\coefflnamedapp{\name{\eta}{t}}{\gamma}{[\vec{u}]}\,)\set{\alpha/\beta}
   =
   \coefflnamedapp{\name{\eta}{t}\set{\alpha/\beta}}{\gamma}{[\vec{u}\set{\alpha/\beta}]}.
  $
\end{enumerate}
\end{Lemma}

The next lemma says how two linear substitutions operate when applied consecutively.

\begin{Lemma}\label{Applamu-lm:1}.
 Let $t\in\lamu^{\mathrm{r}}$, $[\vec{v}]=:[v_1,\dots,v_n]$ and $[\vec{u}]$ bags and $y\neq x$ variables with $y$ \emph{not} occurring in $[\vec{u}]$. Then the sum $t\lsubst{y}{\vec{v}}\lsubst{x}{\vec{u}}$ is:
 \[\begin{small}
  \sum\limits_{W}
  t\lsubst{x}{\vec{w}^{\,0}}
  \lsubst{y}{v_1\lsubst{x}{\vec{w}^{\,1}},\dots,v_n\lsubst{x}{\vec{w}^{\,n}}}
 \end{small}\]
 with $W:(\vec{u})\longto\set{0,\dots,n}$.
\end{Lemma}

The condition $x\neq y$ in the previous lemma cannot be suppressed: for instance for $x=y$, $t=x$, $[\vec{v}]=1$ and $[\vec{u}]=[z]$, with $z\neq y$, the left hand-side of the equality becomes $x\langle 1/x\rangle^+ \langle [z]/x\rangle^+=0$ while the right hand-side becomes $x \langle [z]/x\rangle^+ \langle 1/x\rangle^+ =z$.

The next lemma says how a linear substitution operates on linear named application.

\begin{Lemma}\label{Applamu-lm:2}.
 Let $t\in\lamu^{\mathrm{r}}$, $[\vec{v}]=:[v_1,\dots,v_n]$ and $[\vec{u}]$ bags and $\alpha$ a name with\footnote{By ``$\dg{\beta}{[\vec{u}]}=0$'' we mean that $\dg{\beta}{u}=0$ for all $u\in [\vec{u}]$.} $\dg{\alpha}{[\vec{u}]}=0$. Then the sum $(\,\coefflnamedapp{t}{\alpha}{[\vec{v}]}\,)\,\lsubst{x}{\vec{u}}$ is:
 \[
  \sum\limits_{W}
  \coefflnamedapp{t\lsubst{x}{\vec{w}^{\,0}}}{\alpha}
  {[
  v_1\lsubst{x}{\vec{w}^{\,1}},
  \dots,
  v_n\lsubst{x}{\vec{w}^{\,n}}
  ]}
 \]
 with $W:(\vec{u})\longto\set{0,\dots,n}$.
\end{Lemma}

The following two remarks are easy:

\begin{Remark}\label{Applamu-rmk:inNextLemma}
 Let $\alpha$ be a name, $t\in\lamur$ and $[\vec{u}],[\vec{v}]$ bags. 
 If $\dg{\alpha}{[\vec{v}]}=0$ then one has: \[\coefflnamedapp{t[\vec{v}]}{\alpha}{[\vec{u}]}=(\coefflnamedapp{t}{\alpha}{[\vec{u}]}\,)\,[\vec{v}].\]
\end{Remark}

\begin{Remark}\label{Applamu-rmk:inNextLemma2}
 Let $\alpha\neq\beta$ be names, $t\in\lamur$ and $[\vec{v}]=:[v_1,\dots,v_n],[\vec{u}]$ bags. 
 If $\dg{\beta}{t}=0$ then one has: \[\coefflnamedapp{\coefflnamedapp{t}{\alpha}{[\vec{v}]}}{\beta}{[\vec{u}]}
 =
 \sum\limits_{W}
 \coefflnamedapp{t}{\alpha}{[\coefflnamedapp{v_1}{\beta}{[\vec{w}^{\,1}]},\dots,\coefflnamedapp{v_n}{\beta}{[\vec{w}^{\,n}]}]}.\]
with $W:(\vec{u})\longto\set{1,\dots,n}$.
\end{Remark}

The next lemma says in which sense and when, in some cases, one can swap the order of two linear applications.

\begin{Lemma}\label{Applamu-lm:commutelnamedapp}
 Let $t\in\lamur$, $[\vec{u}],[\vec{v}]$ bags and $\alpha\neq\beta$ names.
 \begin{enumerate}
  \item If $\dg{\alpha}{[\vec{v}]}=0$ and $\dg{\beta}{[\vec{u}]}=0$, then:
  \[\coefflnamedapp{\coefflnamedapp{t}{\alpha}{[\vec{u}]}}{\beta}{[\vec{v}]}=\coefflnamedapp{\coefflnamedapp{t}{\beta}{[\vec{v}]}}{\alpha}{[\vec{u}]}.\]
  \item If $\dg{\alpha}{[\vec{v}]}=0$ then, taking $\delta$ a fresh name, one has:
  \[\coefflnamedapp{\coefflnamedapp{t}{\alpha}{[\vec{u}]}}{\beta}{[\vec{v}]}
  =
  \sum\limits_{W}
  \coefflnamedapp{
    \coefflnamedapp{
      \coefflnamedapp{t}{\beta}{[\vec{w}^{\,1}]}
    }{\alpha}{[\vec{u}\set{\delta/\beta}]}
  }{\delta}{[\vec{w}^{\,2}]}
  \,\,\set{\beta/\delta}\]
with $W:(\vec{v})\longto\set{1,2}$.
  \item If $\dg{\beta}{[\vec{u}]}=0$ then, taking $\delta$ a fresh name, one has:
  \[
   \coefflnamedapp{\coefflnamedapp{t}{\alpha}{[\vec{u}]}}{\beta}{[\vec{v}]}
   =
   \coefflnamedapp{\coefflnamedapp{t}{\beta}{[\vec{v}\set{\delta/\alpha}]}}{\alpha}{[\vec{u}]}\,\,\set{\alpha/\delta}.
  \]
 \end{enumerate}
\end{Lemma}

The next lemma says how a linear named applications on a name operates on a renaming involving the same name.

\begin{Lemma}\label{Applamu-lm:3}
Let $t\in\lamu^{\mathrm{r}}$, $[\vec{u}]$ a bag and $\alpha\neq\beta$ names with $\dg{\beta}{[\vec{u}]}=0$. Then:
 \[
  \coefflnamedapp{t\set{\alpha/\beta}}{\alpha}{[\vec{u}]}
  =
  \sum\limits_{W}
  \coefflnamedapp{
  \coefflnamedapp{t}{\alpha}{[\vec{w}^{\,1}]}
  }{\beta}{[\vec{w}^{\,2}]}
  \,\set{\alpha/\beta}
 \]
with $W:(\vec{u})\longto\set{1,2}$.
\end{Lemma}

The condition $\alpha\neq\beta$ in the previous lemma cannot be suppressed: for instance, for $t=\mu\gamma.\name{\alpha}{x}$ and $[\vec{u}]=1$, the left hand-side of the equality becomes $\mu\gamma.\name{\alpha}{x1}$ while the right hand-side becomes:
\[
	\sum\limits_{W:()\longto\set{1,2}}\mu\gamma.\name{\alpha}{x\,[\vec{w}^{\,1}]\,[\vec{w}^{\,2}]}=\mu\gamma.\name{\alpha}{x\,1\,1}.
\]
Also the condition $\dg{\beta}{[\vec{u}]}=0$ cannot be suppressed: for instance, for $t=\mu\gamma.\name{\alpha}{x}$ and $[\vec{u}]=[\mu\gamma'.\name{\beta}{y}]$, the left hand-side becomes $\mu\gamma.\name{\alpha}{x[\mu\gamma'.\name{\beta}{y}]}$ while the right hand-side becomes $\mu\gamma.\name{\alpha}{x[\mu\gamma'.\name{\beta}{y1}]}$.

The following lemma says how a linear named application operates on a linear substitution.

\begin{Lemma}\label{Applamu-lm:4}
Let $t\in\lamur$, $[\vec{v}]=:[v_1,\dots,v_n],[\vec{u}]$ bags, $x$ a variable and $\alpha$ a name s.t.\ $\dg{x}{[\vec{u}]}=0$.
Then $\coefflnamedapp{
 t\lsubst{x}{\vec{v}}
 }{\alpha}{[\vec{u}]}$ is:
 \[
 \sum\limits_{W}
 (\coefflnamedapp{t}{\alpha}{[\vec{w}^{\,0}]})\lsubst{x}{\coefflnamedapp{v_1}{\alpha}{[\vec{w}^{\,1}]},\dots,\coefflnamedapp{v_n}{\alpha}{[\vec{w}^{\,n}]}}
 \]
 with $W:(\vec{u})\longto\set{0,\dots,n}$.
\end{Lemma}

The condition $\dg{x}{[\vec{u}]}=0$ in the previous lemma cannot be suppressed: for instance, for $t=\mu\gamma.\name{\alpha}{y}$ (with $y\neq x$), $[\vec{v}]=1$ and $[\vec{u}]=[x]$, the left hand-side becomes $\mu\gamma.\name{\alpha}{y[x]}$ while the right hand-side becomes $0$.

The next lemma says how two linear named application operate consecutively.

\begin{Lemma}\label{Applamu-lm:5}
Let $t\in\lamur$, $\alpha\neq\gamma$ names and $[\vec{v}]=:[v_1,\dots,v_n]$,\\$[\vec{u}]$ bags s.t.\ $\dg{\gamma}{[\vec{u}]}=0$.
Then:
 \[
 \coefflnamedapp{
 \coefflnamedapp{
 t}{\gamma}{[\vec{v}]}
 }{\alpha}{[\vec{u}]}
 =
 \sum\limits_{W}
 \coefflnamedapp{
 \coefflnamedapp{t}{\alpha}{[\vec{w}^{\,0}]}}
 {\gamma}
 {[\coefflnamedapp{v_1}{\alpha}{[\vec{w}^{\,1}]},\dots,\coefflnamedapp{v_n}{\alpha}{[\vec{w}^{\,n}]}]}
 \]
 with $W:(\vec{u})\longto\set{0,\dots,n}$.
 
 Furthermore, the same holds also for named terms. That is, under the same hypothesis and for $\eta$ a name, we have\footnote{The following is just an equality between sets of words -- the named terms.} that the sum $\coefflnamedapp{
 \coefflnamedapp{
 \name{\eta}{t}}{\gamma}{[\vec{v}]}
 }{\alpha}{[\vec{u}]}$ is:
 \[
 \sum\limits_{W}
 \coefflnamedapp{
 \coefflnamedapp{\name{\eta}{t}}{\alpha}{[\vec{w}^{\,0}]}}
 {\gamma}
 {[\coefflnamedapp{v_1}{\alpha}{[\vec{w}^{\,1}]},\dots,\coefflnamedapp{v_n}{\alpha}{[\vec{w}^{\,n}]}]}
 \]
 with $W:(\vec{u})\longto\set{0,\dots,n}$.
\end{Lemma}

The condition $\alpha\neq\gamma$ in the previous lemma cannot be suppressed: for instance, for $t=\mu\delta.\name{\alpha}{x}$, $[\vec{u}]=[\mu\delta.\name{\delta}{x}]$ and $[\vec{v}]=[\mu\delta.\name{\delta}{y}]$, the left hand-side becomes\\$\mu\delta.\name{\alpha}{x[\mu\delta.\name{\delta}{y}][\mu\delta.\name{\delta}{x}]}$ while the right hand-side becomes $\mu\delta.\name{\alpha}{x[\mu\delta.\name{\delta}{x}][\mu\delta.\name{\delta}{y}]}$.

Also the condition $\dg{\gamma}{[\vec{u}]}=0$ in the previous lemma cannot be suppressed: for instance, for $t=\mu\delta.\name{\alpha}{x}$, $[\vec{u}]=[\mu\delta.\name{\gamma}{x}]$ and $[\vec{v}]=[\mu\delta.\name{\delta}{x}]$, the left hand-side becomes $0$ while the right hand-side becomes $\mu\delta.\name{\alpha}{x[\mu\delta.\name{\gamma}{x[\mu\delta.\name{\delta}{x}]}]}$.

\begin{Remark}\label{Applamu-rmk:WCofSU}
Let $[p,\vec{q}]$ be a bag.
Then, every function $W':\set{p,u_1,\dots,u_k}\longto\set{0,\dots,n}$ is uniquely determined by the choice of $W'(p)\in\set{0,\dots,n}$ plus the choice of a function $W:\set{u_1,\dots,u_k}\longto\set{0,\dots,n}$.
This is reflected in the equality $(n+1)^{k+1}=(n+1)^{k}\cdot(n+1)$.
Now, for $0\leq i,j\leq n$, let us set $[p]_i^j$ the singleton multiset $[p]$ if $i=j$, and the empty mulitset $1$ if $i\neq j$.
Therefore, the w.c.\ of $[p,\vec{q}]$ generated by a $W':\set{p,u_1,\dots,u_k}\longto\set{0,\dots,n}$ is of shape $([\vec{w}^{\,0}]*[p]_0^{W'(p)},\dots,[\vec{w}^{\,n}]*[p]_n^{W'(p)})$, for $([\vec{w}^{\,0}],\dots,[\vec{w}^{\,n}])$ a w.c.\ of $[\vec{q}]$ generated by a $W:\set{u_1,\dots,u_k}\longto\set{0,\dots,n}$.
\end{Remark}

\vspace{0.2cm}

In the main paper, there is a sketch of the proof of Lemma~\ref{lamu-lm:ForLocalConfluence} with the most important cases.
As declared at the beginning of this Appendix section, we give below its complete proof.

\begin{proof}[Proof of Lemma~\ref{lamu-lm:ForLocalConfluence}]

(1).
Induction on $s$:

Case $s$ variable: impossible.

Case $s=\lam y.s'$: straightforward by inductive hypothesis.

Case $s=\mu\gamma.\name{\eta}{s'}$: we have two subcases\footnote{In the following we do not explicitly say it, but of course we will take bound names different from $\alpha$ and $\beta$, as well as from other bound names.}:

Subcase $s\to_{\mathrm{r}}^+ \Sum S$ is performed by reducing $s'$:\\
then $\Sum S=\mu\gamma.\name{\eta}{\Sum S'}$ with $s'\tor^+ \Sum S'$, and so (regardless of whether $\eta$ equals $\gamma$ or not):
\[\begin{array}{rcl}
s\set{\alpha/\beta}=\mu\gamma.\name{\delta_\eta^\alpha(\beta)}{s'\set{\alpha/\beta}}&\tor^+&\mu\gamma.\name{\delta_\eta^\alpha(\beta)}{\Sum S'\set{\alpha/\beta}}\\
&=&\mu\gamma.\name{\eta}{\Sum S'}\set{\alpha/\beta}=\Sum S\set{\alpha/\beta}
\end{array}\]
where we used the inductive hypothesis.

Subcase $s'=\mu\gamma'.\name{\eta'}{s''}$ and $s\to_{\mathrm{r}}^+ \Sum S$ is performed by reducing its leftmost $\rho$-redex:\\
Then $s=\mu\gamma.\name{\eta}{\mu\gamma'.\name{\eta'}{s''}}$, $\Sum S=\mu\gamma.\name{\eta'}{s''}\set{\eta/\gamma'}$ and we have four sub-subcases:

Sub-subcase $\eta=\beta$ and $\eta'=\beta$: then \[\Sum S\set{\alpha/\beta}=\mu\gamma.\name{\beta}{s''}\set{\alpha/\beta,\alpha/\gamma'}=\mu\gamma.\name{\alpha}{s''\set{\alpha/\beta,\alpha/\gamma'}}\] and:
\[
	\begin{array}{rcl}
s\set{\alpha/\beta}&=&\mu\gamma.\name{\alpha}{\mu\gamma'.\name{\alpha}{s''\set{\alpha/\beta}}}\\
&\to_\rho^+&\mu\gamma.\name{\alpha}{s''\set{\alpha/\beta}}\set{\alpha/\gamma'}\\
	&=&\mu\gamma.\name{\alpha}{s''\set{\alpha/\beta,\alpha/\gamma'}}=\Sum S\set{\alpha/\beta}.
	\end{array}
\]

Sub-subcase $\eta=\beta$ and $\eta'\neq\beta$: then 
\[
	\Sum S\set{\alpha/\beta}=\mu\gamma.\name{\eta'}{s''}\set{\alpha/\beta,\alpha/\gamma'}=\mu\gamma.\name{\delta_{\eta'}^\alpha(\gamma')}{s''\set{\alpha/\beta,\alpha/\gamma'}}
	\]
	 and:
\[
	\begin{array}{rcl}
	s\set{\alpha/\beta}&=&\mu\gamma.\name{\alpha}{\mu\gamma'.\name{\eta'}{s''\set{\alpha/\beta}}}\\
&\to_\rho^+&\mu\gamma.\name{\eta'}{s''\set{\alpha/\beta}}\set{\alpha/\gamma'}\\
	&=&\mu\gamma.\name{\delta_{\eta'}^\alpha(\gamma')}{s''\set{\alpha/\beta,\alpha/\gamma'}}=\Sum S\set{\alpha/\beta}.
	\end{array}
\]
Sub-subcase $\eta\neq\beta$ and $\eta'=\beta$: then 
\[
\Sum S\set{\alpha/\beta}=\mu\gamma.\name{\beta}{s''}\set{\eta/\gamma'}\set{\alpha/\beta}=\mu\gamma.\name{\alpha}{s''\set{\eta/\gamma'}\set{\alpha/\beta}}
\] 
and:
\[\begin{array}{rcl}
s\set{\alpha/\beta}&=&\mu\gamma.\name{\eta}{\mu\gamma'.\name{\alpha}{s''\set{\alpha/\beta}}}\\
&\to_\rho^+&\mu\gamma.\name{\alpha}{s''\set{\alpha/\beta}}\set{\eta/\gamma'}\\
&=&\mu\gamma.\name{\alpha}{s''\set{\alpha/\beta}\set{\eta/\gamma'}}
\end{array}\]
and the result follows by Lemma \ref{Applamu-lm:nameCommutes}.

Sub-subcase $\eta\neq\beta$ and $\eta'\neq\beta$: then $\Sum S\set{\alpha/\beta}$ is the term: 
\[\mu\gamma.\name{\delta_{\eta'}^{\eta}(\gamma')}{s''\set{\eta/\gamma'}}\set{\alpha/\beta}=\mu\gamma.\name{\delta_{\eta'}^{\eta}(\gamma')}{s''\set{\eta/\gamma'}\set{\alpha/\beta}}
\]
and:
\[\begin{array}{rcl}
s\set{\alpha/\beta}&=&\mu\gamma.\name{\eta}{\mu\gamma'.\name{\eta'}{s''\set{\alpha/\beta}}}\\
&\to_\rho^+&\mu\gamma.\name{\eta'}{s''\set{\alpha/\beta}}\set{\eta/\gamma'}\\
&=&\mu\gamma.\name{\delta_{\eta'}^{\eta}(\gamma')}{s''\set{\alpha/\beta}\set{\eta/\gamma'}}
\end{array}\]
and the result follows by Lemma \ref{Applamu-lm:nameCommutes}.

Case $s=s'[\vec{v}]$: we have four subcases:

Subcase $s\to_{\mathrm{r}}^+ \Sum S$ is performed by reducing the $s'$: straightforward by inductive hypothesis.

Subcase $s\to_{\mathrm{r}}^+ \Sum S$ is performed by reducing one $v_i\in[\vec{v}]$: straightforward by inductive hypothesis.

Subcase $s'=\lam x.s''$ and $s\to_{\mathrm{r}}^+ \Sum S$ is performed by reducing the $\lam$-redex $s$: then $\Sum S=s''\langle [\vec{v}]/x \rangle^+$ and:
\[\begin{array}{rcl}
s\set{\alpha/\beta}&=&(\lam x.s''\set{\alpha/\beta})[\vec{v}\set{\alpha/\beta}]\\
&\to_{\lam^{\mathrm{r}}}& s''\set{\alpha/\beta}\langle [\vec{v}\set{\alpha/\beta}]/x \rangle^+\\
&=&s''\langle [\vec{v}]/x \rangle^+\set{\alpha/\beta}=\Sum S\set{\alpha/\beta}
\end{array}\]
where the second-last equality holds thanks to Lemma \ref{Applamu-lm:outNames}.

Subcase $s'=\mu\gamma.\name{\eta}{s''}$ and $s\to_{\mathrm{r}}^+ \Sum S$ is performed by reducing the $\mu$-redex $s$: then $\Sum S=\mu\gamma.\coefflnamedapp{\name{\eta}{s''}}{\gamma}{[\vec{v}]}$ and:
\[\begin{array}{rll}
s\set{\alpha/\beta} &
= &
(\mu\gamma.\name{\delta_\eta^\alpha(\beta)}{s''\set{\alpha/\beta}})[\vec{v}\set{\alpha/\beta}] \\
&
\to_{\mu^{\mathrm{r}}} &
\mu\gamma.\coefflnamedapp{\name{\delta_\eta^\alpha(\beta)}{s''\set{\alpha/\beta}}}{\gamma}{[\vec{v}\set{\alpha/\beta}]} \\
& = &
\mu\gamma.\coefflnamedapp{\name{\eta}{s''}\set{\alpha/\beta}}{\gamma}{[\vec{v}\set{\alpha/\beta}]}
\end{array}\]
and the equality with $\Sum S\set{\alpha/\beta}$ follows from Lemma \ref{Applamu-lm:outNames}.

(2).
Induction on $t$.
The only non-trivial case is\footnote{In the case $t=y\neq x$, and if $\Sum S =0$ and $\vec{u}$ is empty, note that $t\langle [\Sum S,\vec{u}]/x\rangle^+=y\langle [0]/x\rangle^+=\sum\limits_{s'\in 0} y\langle [s']/x\rangle^+ = 0$ (and \emph{not} ``$y\langle [0]/x\rangle^+ = y\langle 1/x\rangle^+=y$''), so the result still holds in this case.}  $t=v_0[v_1,\dots,v_n]$.
In this case, by Remark \ref{Applamu-rmk:WCofSU}, we have that $(\,v_0[v_1,\dots,v_n]\,)\,\langle [s,\vec{u}]/x\rangle^+ $ is:
\[
 \sum\limits_{W}\,
 \sum\limits_{j=0}^n
 (\, v_0\langle [\vec{w}^{\,0}]*[s]_0^j/x \rangle \,)\,[\dots,v_i\langle [\vec{w}^{\,i}]*[s]_i^j/x \rangle,\dots].
\]
where $W:(\vec{u})\longto\set{1,\dots,n}$.
Fix now a $W:(\vec{u})\longto\set{1,\dots,n}$ (together with its generated w.c.) and consider each of the $n+1$ elements of the sum on $j$. 
We write the case for $j=0$, but the other cases are exactly the same. 
Since $j=0$, the element is $(\, v_0\langle [\vec{w}^{\,0}]*[s]/x \rangle \,)\,[\dots,v_i\langle [\vec{w}^{\,i}]/x \rangle,\dots]$ and by inductive hypothesis we have:
\[\begin{array}{l}
 (\, v_0\langle [\vec{w}^{\,0}]*[s]/x \rangle \,)\,[\dots,v_i\langle [\vec{w}^{\,i}]*[s]/x \rangle,\dots]\\
 \msto[\mathrm{r}]^+
 (\, v_0\langle [\vec{w}^{\,0}]*[\Sum S]/x \rangle \,)\,[\dots,v_i\langle [\vec{w}^{\,i}]/x \rangle,\dots].
\end{array}\]
Now summing up all the elements for $j=0,\dots,n$ and $W:(\vec{u})\longto\set{1,\dots,n}$ we obtain the following sum:
\[
 \sum\limits_{W}\,
 \sum\limits_{j=0}^n
  (\, v_0\langle [\vec{w}^{\,0}]*[\Sum S]_0^j/x \rangle \,)[\dots,v_i\langle [\vec{w}^{\,i}]*[\Sum S]_i^j/x \rangle,\dots].
\]
Using again Remark \ref{Applamu-rmk:WCofSU}, the above sum becomes:
\[
 (\,v_0[v_1,\dots,v_n]\,)\,\langle [\Sum S,\vec{u}]/x\rangle^+ 
\]
which is the desired result.

(3).
Induction on $s$.

Case $s$ variable: impossible.

Case $s=$ $\lam$-abstraction: straightforward by inductive hypothesis.

Case $s=\mu\alpha.\name{\beta}{s'}$: we have two subcases:

Subcase $s\to_{\mathrm{r}}^+ \Sum S$ is performed by reducing $s'$: immediate.

Subcase $s'=\mu\gamma.\name{\eta}{s''}$ and $s\to_{\mathrm{r}}^+ \Sum S$ is performed by reducing its leftmost $\rho$-redex:\\
then $\Sum S=\mu\alpha.\name{\eta}{s''}\set{\beta/\gamma}$ and (call $[\vec{u}]=:[u_1,\dots,u_k]$)
\[
 \begin{array}{rllr}
  s\langle [\vec{u}]/x\rangle^+ &
  = &
  \mu\alpha.\name{\beta}{\mu\gamma.\name{\eta}{s''\langle [\vec{u}]/x \rangle^+}} 
  &\\
   &
  \to_{\rho^{\mathrm{r}}} &
  \mu\alpha.\name{\eta}{s''\langle [\vec{u}]/x \rangle^+}\set{\beta/\gamma} 
  & \\
   & =
  & \mu\alpha.\name{\eta}{s''}\langle [\vec{u}]/x \rangle^+\set{\beta/\gamma} & \\
  & =
  & \mu\alpha.\name{\eta}{s''}\set{\beta/\gamma}\langle [\vec{u}]/x \rangle^+ & \\
   &
  = &
  \Sum S\langle [\vec{u}]/x \rangle^+. 
  & 
\end{array}
\]
Case $s=s'[\vec{v}]$: we have four subcases:

Subcase $s\to_{\mathrm{r}}^+ \Sum S$ is performed by reducing $s'$: straightforward by inductive hypothesis.

Subcase $s\to_{\mathrm{r}}^+ \Sum S$ is performed by reducing the $v_i$ in $[\vec{v}]$: straightforward by inductive hypothesis.

Subcase $s'=\lam y.s''$ and $s\to_{\mathrm{r}}^+ \Sum S$ is performed by reducing the $\lam$-redex $s$: follows by Lemma \ref{Applamu-lm:1} (where $x\neq y$ because $y$ is bound and $x$ is fixed).

Subcase $s'=\mu\alpha.\name{\beta}{s''}$ and $s\to_{\mathrm{r}}^+ \Sum S$ is performed by reducing the $\mu$-redex $s$:
follows by Lemma \ref{Applamu-lm:2} (where $\dg{\alpha}{[\vec{u}]}=0$ because $\alpha$ is bound and $[\vec{u}]$ is fixed).

(4).
Induction on $t$.
The non-trivial cases are $t=v_0[v_1,\dots,v_n]$ and $t=\mu\gamma.\name{\eta}{t'}$.
Both are done following the same argument as we did in point (2); apply Remark \ref{Applamu-rmk:WCofSU}, then the inductive hypothesis and thus close the argument by applying Remark \ref{Applamu-rmk:WCofSU}.

(5).
It is immediate discriminating the cases $\alpha=\beta$ and $\alpha\neq\beta$, and concluding by point (4).

(6).
Induction on $s\in\lamu$.

Case $s=$ variable. Impossible.

Case $s=$ $\lam$-abstraction. Straightforward by inductive hypothesis.

Case $s=\mu\beta.\name{\gamma}{s'}$: we have two subcases:

Subcase $s\to_{\mathrm{r}}^+ \Sum S$ is performed by reducing $s'$: then $\Sum S=\mu\beta.\name{\gamma}{\Sum S'}$ with $s'\to_\mathrm{r} \Sum S'$. We have $\coefflnamedapp{s}{\alpha}{[\vec{u}]}=\mu\beta.\coefflnamedapp{\name{\gamma}{s'}}{\alpha}{[\vec{u}]}$. Remark that we \emph{cannot} immediately apply the inductive hypothesis on $\name{\gamma}{s'}$, simply because $\name{\gamma}{s'}\notin\lamu^{\mathrm{r}}$ (it is a named term). However we can split in the two subcases whether $\gamma=\alpha$ or $\gamma\neq\alpha$ and now in both subcases we can conclude straightforwardly by inductive hypothesis.

Subcase $s'=\mu\gamma'.\name{\eta}{s''}$ (with $\gamma\neq\gamma'$) and $s\to_{\mathrm{r}}^+ \Sum S$ is performed by reducing its leftmost $\rho$-redex:
then $\Sum S=\mu\beta.\name{\eta}{s''}\set{\gamma/\gamma'}$.
We split in two sub-subcases:

Sub-subcase $\alpha\neq\gamma$: 
\[\begin{array}{rlll}
 \coefflnamedapp{s}{\alpha}{[\vec{u}]}
 & = &
 \mu\beta.\name{\gamma}{\mu\gamma'.\coefflnamedapp{\name{\eta}{s''}}{\alpha}{[\vec{u}]}} & 
 \\
 & \mstor{\rho}^+ &
 \mu\beta.\,
 \coefflnamedapp{\name{\eta}{s''}}{\alpha}{[\vec{u}]}\,
 \set{\gamma/\gamma'} &
 \\
 & = &
 \coefflnamedapp{\mu\beta.\name{\eta}{s''}\set{\gamma/\gamma'}}{\alpha}{[\vec{u}]}
 & \textit{(by Lemma \ref{Applamu-lm:outNames}}
\\
& \phantom{=} & \phantom{\coefflnamedapp{\mu\beta.\name{\eta}{s''}\set{\gamma/\gamma'}}{\alpha}{[\vec{u}]}} & \textit{\,plus $\dg{\gamma'}{[\vec{u}]}=0$)}
 \\
 & = &
 \coefflnamedapp{\Sum S}{\alpha}{[\vec{u}]}. &
\end{array}\]
Sub-subcase $\alpha=\gamma$.
We have:
\[\begin{array}{rlll}
 \coefflnamedapp{s}{\alpha}{[\vec{u}]}
 & = &
 \sum\limits_{W}
 \mu\beta.\name{\alpha}{
  \,(\,
   \mu\gamma'.
   \coefflnamedapp{
    \name{\eta}{s''}
   }{\alpha}{[\vec{w}^{\,1}]}
  \,)\,
  [\vec{w}^{\,2}]
 }
 & \\
& & \textit{where } W:(\vec{u})\longto\set{1,2} &
 \\ \\
 & \mstor{\mu}^+ &
 \sum\limits_{W}
 \mu\beta.\name{\alpha}{
  \mu\gamma'.\coefflnamedapp{
   \coefflnamedapp{\name{\eta}{s''}}{\alpha}{[\vec{w}^{\,1}]}
  }{\gamma'}{[\vec{w}^{\,2}]}
 } &
 \\
 & \mstor{\rho}^+ &
 \sum\limits_{W}
 \mu\beta.
 \coefflnamedapp{
  \coefflnamedapp{
   \name{\eta}{s''}
  }{\alpha}{[\vec{w}^{\,1}]}
 }{\gamma'}{[\vec{w}^{\,2}]}
 \, \set{\alpha/\gamma'} &
 \\
 & = &
 \sum\limits_{W}
 \coefflnamedapp{
  \coefflnamedapp{
   \mu\beta.\name{\eta}{s''}
  }{\alpha}{[\vec{w}^{\,1}]}
 }{\gamma'}{[\vec{w}^{\,2}]}
 \, \set{\alpha/\gamma'}
 \\
 & = &
 \coefflnamedapp{\mu\beta.\name{\eta}{s''}\,\set{\alpha/\gamma'}}{\alpha}{[\vec{u}]} \textit{ \ (by Lemma \ref{Applamu-lm:3})} &
 \\ 
 & = &
 \coefflnamedapp{\Sum S}{\alpha}{[\vec{u}]}\quad \textit{(since }\alpha=\gamma \textit{).} &
\end{array}\]
Case $s=s'[v_1,\dots,v_n]$: we have four subcases:

Subcase $s\to_{\mathrm{r}}^+ \Sum S$ is performed by reducing $s'$: straightforward by inductive hypothesis.

Subcase $s\to_{\mathrm{r}}^+ \Sum S$ is performed by reducing the $v_i$ in $[\vec{v}]$: straightforward by inductive hypothesis.

Subcase $s'=\lam x.s''$ and $s\to_{\mathrm{r}}^+ \Sum S$ is performed by reducing the $\lam$-redex $s$.
Then $\Sum S = s''\lsubst{x}{\vec{v}}$ and we have:
\[\begin{array}{rcl}
 \coefflnamedapp{s}{\alpha}{[\vec{u}]}
 & = &
 \sum\limits_{W}
 (\lam x. \coefflnamedapp{s''}{\alpha}{[\vec{w}^{\,0}]})\,[\dots,\coefflnamedapp{v_i}{\alpha}{[\vec{w}^{\,i}]},\dots]
\\
 & & \textit{ with } W:(\vec{u})\longto\set{0,\dots,n}
 \\
 \\
 & \mstor{\lam}^+ &
 \sum\limits_{W}
 (\coefflnamedapp{s''}{\alpha}{[\vec{w}^{\,0}]})\lsubst{x}{\dots,\coefflnamedapp{v_i}{\alpha}{[\vec{w}^{\,i}]},\dots}
 \\
 \\
 & = &
 \coefflnamedapp{
 s''\lsubst{x}{\vec{v}}
 }{\alpha}{[\vec{u}]} \qquad \textit{(by Lemma \ref{Applamu-lm:4})}
 \\
 \\
 & = &
 \coefflnamedapp{\Sum S}{\alpha}{[\vec{u}]}.
\end{array}\]
Subcase $s'=\mu\gamma.\name{\eta}{s''}$ (with $\gamma\neq\alpha$) and $s\to_{\mathrm{r}}^+ \Sum S$ is performed by reducing the $\mu$-redex $s$.
Then $\Sum S = \mu\gamma.\coefflnamedapp{\name{\eta}{s''}}{\gamma}{[\vec{v}]}$
and we have:
\[\small\begin{array}{rcl}
 \coefflnamedapp{s}{\alpha}{[\vec{u}]}
 & = &
 \sum\limits_{W}
 (\coefflnamedapp{\mu\gamma.\name{\eta}{s''}}{\alpha}{[\vec{w}^{\,0}]})\,[\dots,\coefflnamedapp{v_i}{\alpha}{[\vec{w}^{\,i}]},\dots]
\\
 & & \textit{with } W:(\vec{u})\longto\set{0,\dots,n}
 \\
 \\
 & \mstor{\mu}^+ &
 \mu\gamma.
 \sum\limits_{W}
 \coefflnamedapp{
 \coefflnamedapp{\name{\eta}{s''}}{\alpha}{[\vec{w}^{\,0}]}}
 {\gamma}
 {[\dots,\coefflnamedapp{v_i}{\alpha}{[\vec{w}^{\,i}]},\dots]}
 \\
 \\
 & = &
 \mu\gamma.
 \coefflnamedapp{
 \coefflnamedapp{
 \name{\eta}{s''}}{\gamma}{[\vec{v}]}
 }{\alpha}{[\vec{u}]} \qquad \textit{(by Lemma \ref{Applamu-lm:5})}
 \\
 \\
 & = &
 \coefflnamedapp{\Sum S}{\alpha}{[\vec{u}]}.
\end{array}\]

(7).
It is immediate by discriminating the cases $\alpha=\beta$ and $\alpha\neq\beta$, and then concluding by point (6).\qedhere
\end{proof}

\vspace{0.2cm}

In the main paper, there is a sketch of the proof of Proposition~\ref{lamu-prop:lamu^+Confl} with the most important cases.
As declared at the beginning of this Appendix section, we give below its complete proof.

\begin{proof}[Proof of Proposition~\ref{lamu-prop:lamu^+Confl}]
We show, by induction on a single-hole resource context $c$, that if:
\[t\rightarrow^+_{\mathrm{base}^\mathrm{r}}\Sum T \textit{ and } c\hole{t} \tor^+ \Sum{T}_2\] 
then there is $\Sum T'\in\N\langle\lamur\rangle$ s.t.\:
\[c\hole{\Sum T} \msto[\mathrm{r}]^+ \Sum T'\, {}^+_{\mathrm{r}}\!\!\!\twoheadleftarrow \Sum{T}_2.\]
In all the following diagrams we write ``$\to$'' but of course we mean ``$\tor^+$''.

(1). Case $c=\xi$.

So $c\hole{t}=t\rightarrow^+_{\mathrm{base}^\mathrm{r}} \Sum T$ and we only have the three base-cases of Definition~\ref{lamu-def:qualitlamur}.
In all the diagrams of this case, when not explicitly said differently or not clear by an easy reduction, the bottom-left reduction follows from Lemma \ref{lamu-lm:ForLocalConfluence} and the bottom-right from Remark \ref{lamu-rm:easyTakeOutSum}.

Subcase $t=(\lam x.s)[\vec{u}]$ and $\Sum T = s\langle[\vec{u}]/x \rangle^+$. 
Then $c\hole{t}=t\tor^+\Sum T_2$ (on a different redex than $t$) can only be performed either by reducing $s$, or by reducing an element $w$ of $[\vec{u}]$.
We have thus the following two diagrams:
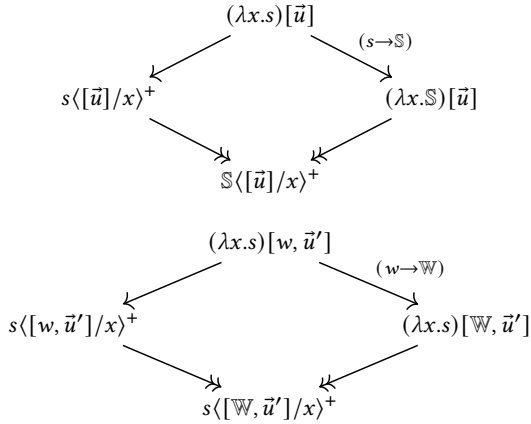
\[\small{
\begin{tikzcd}[column sep=scriptsize]
& (\lam x.s)[\vec{u}] \ar[dl] \ar[dr, "(s\rightarrow\Sum S)"] & & \\
s\langle[\vec{u}]/x \rangle^+ \ar[->>, dr] & & (\lam x.\Sum S)[\vec{u}] \ar[->>, dl] \\
& \Sum S \langle[\vec{u}]/x \rangle^+ & &
\end{tikzcd}
}\]
\[\small{
\begin{tikzcd}[column sep=scriptsize]
& (\lam x.s)[w,\vec{u}'] \ar[dl] \ar[dr, "(w\rightarrow\Sum W)"] & & \\
s\langle[w,\vec{u}']/x \rangle^+ \ar[->>, dr] & & (\lam x.s)[\Sum W,\vec{u}'] \ar[->>, dl] \\
& s\langle[\Sum W,\vec{u}']/x \rangle^+ & &
\end{tikzcd}
}\]

Subcase $t=(\mu\alpha.\name{\beta}{s})[\vec{u}]$ and $\Sum T = \mu\alpha.\langle \name{\beta}{s} \rangle^+_\alpha[\vec{u}]$.
Then $c\hole{t}=t\tor^+\Sum T_2$ (on a different redex than $t$) can only be performed either by reducing $s$, or by reducing an element $w$ of $[\vec{u}]$, or if $s=\mu\gamma.\name{\eta}{s'}$ and we reduce the $\rho$-redex $...\name{\beta}{\mu\gamma.\,...}$. In the latter case we split into the case $\alpha\neq\beta$, the case $\alpha=\beta,\gamma\neq\eta,\eta=\alpha$, the case $\alpha=\beta,\gamma\neq\eta,\eta\neq\alpha$, and the case $\alpha=\beta,\eta=\gamma$ (with necessary $\gamma\neq\alpha$).\\
The first case is given by the diagram:
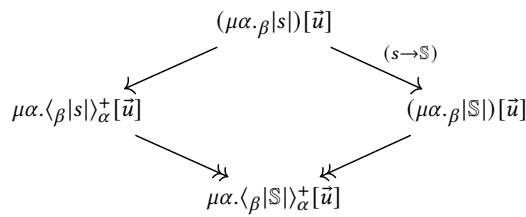
\[\small{
\begin{tikzcd}[column sep=scriptsize]
& (\mu\alpha.\name{\beta}{s})[\vec{u}] \ar[dl] \ar[dr, "(s\rightarrow\Sum S)"] & & \\
\mu\alpha.\langle \name{\beta}{s} \rangle^+_\alpha[\vec{u}] \ar[->>, dr] & & (\mu\alpha.\name{\beta}{\Sum S})[\vec{u}] \ar[->>, dl] \\
& \mu\alpha.\langle \name{\beta}{\Sum S} \rangle^+_\alpha[\vec{u}] & &
\end{tikzcd}
}\]
The second case is given by the diagram:
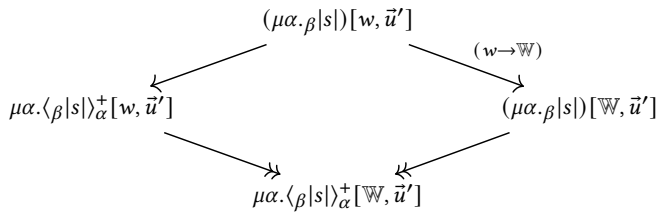
\[\small{
\begin{tikzcd}
& (\mu\alpha.\name{\beta}{s})[w,\vec{u}'] \ar[dl] \ar[dr, "(w\rightarrow\Sum W)"] & & \\
\mu\alpha.\langle \name{\beta}{s} \rangle^+_\alpha[w,\vec{u}'] \ar[->>, dr] & & (\mu\alpha.\name{\beta}{s})[\Sum W,\vec{u}'] \ar[->>, dl] \\
& \mu\alpha.\langle \name{\beta}{s} \rangle^+_\alpha[\Sum W,\vec{u}'] & &
\end{tikzcd}
}\]
The third case is given by the diagram in Figure~\ref{fig:6Diag.3}.\\
The fourth case is given by the diagram in Figure~\ref{fig:6Diag.4}.
The bottom right equality holds because, by~\autoref{Applamu-lm:3}, \autoref{Applamu-rmk:inNextLemma} and since $\dg{\gamma}{[\vec{u}]}=0$, each addend of the bottom right sum is:
\[\begin{array}{rl}
& \mu\alpha.\name{\alpha}{(\coefflnamedapp{s'\set{\alpha/\gamma}}{\alpha}{[\vec{w}^{\,0}]})[\vec{w}^{\,1}]} \\
 = &
\sum\limits_{\small{D:(\vec{w}^{\,0})\to\set{1,2}}}
\mu\alpha.\name{\alpha}{(\coefflnamedapp{\coefflnamedapp{s'}{\alpha}{[\vec{d}^{\,1}]}}{\gamma}{[\vec{d}^{\,0}]})[\vec{w}^{\,1}]}\set{\alpha/\gamma}
\\
 = &
\sum\limits_{D}
\mu\alpha.\name{\alpha}{\coefflnamedapp{(\coefflnamedapp{s'}{\alpha}{[\vec{d}^{\,1}]})[\vec{w}^{\,1}]}{\gamma}{[\vec{d}^{\,0}]}}\set{\alpha/\gamma}
\end{array}\] 
and since we are then summing up on all possible $W:(\vec{u})\to\set{1,2}$, the resulting sum is the same as the one at the bottom of the above diagram.\\
The fifth case is given by the diagram in Figure~\ref{fig:6Diag.5}.
The bottom right equality holds by~\autoref{Applamu-lm:3} and because $\dg{\gamma}{[\vec{u}]}=0$.\\
The sixth case is given by the diagram in Figure~\ref{fig:6Diag.6}.
The bottom right equality holds by~\autoref{Applamu-lm:3} and because $\dg{\gamma}{[\vec{u}]}=0$.

Subcase $t=\mu\gamma.\name{\alpha}{\mu\beta.\name{\eta}{s}}$ and $\Sum T = \mu\gamma.\name{\eta}{s}\set{\alpha/\beta}$.
Then $c\hole{t}=t\tor^+\Sum T_2$ (on a different redex than $t$) can be only performed either by reducing $s$, or if $s=\mu\gamma'.\name{\eta'}{s'}$ and we reduce the $\rho$-redex $...\name{\eta}{\mu\gamma'.\,...}$.
In the first case the diagram is:
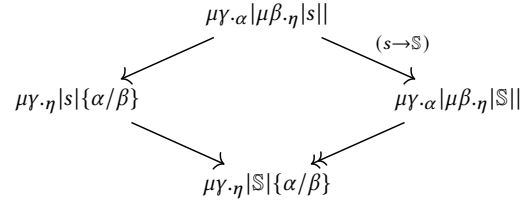
\[\small{
\begin{tikzcd}[column sep=scriptsize]\centering
& \mu\gamma.\name{\alpha}{\mu\beta.\name{\eta}{s}} \ar[dl] \ar[dr, "(s\rightarrow\Sum S)"] & & \\
\mu\gamma.\name{\eta}{s}\set{\alpha/\beta} \ar[dr] & & \mu\gamma.\name{\alpha}{\mu\beta.\name{\eta}{\Sum S}} \ar[->>, dl] \\
& \mu\gamma.\name{\eta}{\Sum S}\set{\alpha/\beta} & &
\end{tikzcd}
}\]
In the second case the diagram is given in Figure~\ref{fig:twoDiag}.
In order to prove that the second diagram holds, let us prove the equality in the last diagram: we first show that $s'\set{\alpha/\beta}\set{\delta_0/\gamma'}=s'\set{\eta/\gamma'}\set{\alpha/\beta}$ and then that $\delta_1=\delta_2$.
For the former equality, we have: 

if $\beta\neq\eta$, then $\delta_0=\eta$ and the equality follows from~\autoref{Applamu-lm:nameCommutes}; 

if $\beta=\eta$ then $\delta_0=\alpha$ and the equality holds because both renamings $\set{\alpha/\beta}\set{\delta_0/\gamma'}$ and $\set{\eta/\gamma'}\set{\alpha/\beta}$ coincide with the unique renaming $\set{\alpha/(\beta,\gamma')}$\footnote{We mean here that both $\beta$ and $\gamma'$ get renamed with $\alpha$.}.

For the latter equality, we have:

if $\gamma'=\eta'$: then $\delta'_2=\eta$ and thus $\delta_2=\delta_0$.
Remark that it must be $\beta\neq\eta'$, because otherwise $\beta=\gamma'$ which is impossible. This means that $\delta'_1=\eta'$. But then $\delta_1=\delta_0$, so we are done.

if $\gamma'\neq\eta'$: then $\delta'_2=\eta'$ and thus $\delta_2=\delta'_1$.
Now if $\beta\neq\eta'$ then $\delta'_1=\eta'$ and thus $\delta_1=\eta'$; 
if $\beta=\eta'$ then $\delta'_1=\alpha$ and thus $\delta_1=\delta_\alpha^{\delta_0}(\gamma')=\alpha$, because it cannot be $\gamma'=\alpha$.
If we read what we just found about $\delta_1$, it precisely says that $\delta_1=\delta'_1$ so we are done.

(2). Case $c=\mu\alpha.\name{\beta}{c'}$ with either $c'\hole{t}$ \emph{not} a $\mu$-abstraction, or $c=\lam x.c'$.

Then the reduction $c\hole{t}\rightarrow\Sum{T}_2$ can only be performed via a reduction $c'\hole{t}\rightarrow\Sum{T}'_2$.
Both the diagrams for the two cases of $c$ have the exact same shape; let us only give the one for $c=\mu\alpha.\name{\beta}{c'}$:
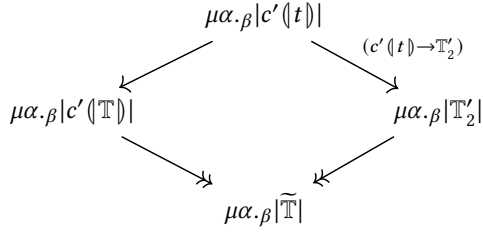
\[
\begin{tikzcd}[column sep=scriptsize]
& \mu\alpha.\name{\beta}{c'\hole{t}} \ar[dl] \ar[dr, "(c'\hole{t}\rightarrow\Sum{T}'_2)"] & & \\
\mu\alpha.\name{\beta}{c'\hole{\Sum T}} \ar[->>, dr] & & \mu\alpha.\name{\beta}{\Sum{T}'_2} \ar[->>, dl] \\
& \mu\alpha.\name{\beta}{\widetilde{\Sum T}} & &
\end{tikzcd}
\]
where a sum $\widetilde{\Sum T}$ s.t.\ $c'\hole{\Sum T}\msto \widetilde{\Sum T}\twoheadleftarrow \Sum{T}'_2$ is given by inductive hypothesis on $c'$.

(3). Case $c=\mu\alpha.\name{\beta}{c'}$, with $c'\hole{t}$ a $\mu$-abstraction.

Then either $c'=\xi$ and $t=\mu\gamma.\name{\eta}{t'_0}$, or $c'=\mu\gamma.\name{\eta}{c''}$. 

In the first case, since by hypothesis $t\rightarrow^+_{\mathrm{base}^\mathrm{r}}\Sum T$, it must be $t=\mu\gamma.\name{\eta}{\mu\gamma'.\name{\eta'}{t'}}$ and $\Sum T=\mu\gamma.\name{\eta'}{t'}\set{\eta/\gamma'}$.
Therefore, the reduction $c\hole{t} \to_\mathrm{r} \Sum{T}_2$ (on a different redex than the one of $t\rightarrow^+_{\mathrm{base}^\mathrm{r}}\Sum T$) can only be performed either by reducing $t'$, or by reducing the $\rho$-redex $...\name{\beta}{\mu\gamma.\,...}$, or if $t'=\mu\widetilde{\gamma}.\name{\widetilde{\eta}}{\widetilde{t}}$ and we reduce the $\rho$-redex $...\name{\eta'}{\mu\widetilde{\gamma}.\,...}$.
The first and second situation of the previous list have been already treated in the case $c=\xi$ (with the notation used there, it corresponds to the two diagrams of the subcase $t=\mu\gamma.\name{\alpha}{\mu\beta.\name{\eta}{s}}$ and $\Sum T = \mu\gamma.\name{\eta}{s}\set{\alpha/\beta}$.).
Also the third situation corresponds to the exact same case just mentioned, because the external $\mu\alpha.\name{\beta}{...}$ is not modified in neither reductions.

In the second case, then the reduction $c\hole{t} \to_\mathrm{r} \Sum{T}_2$ can only be performed either by reducing $c''\hole{t}$, or by reducing the $\rho$-redex $c\hole{t}$.
In the first situation the diagram follows easily by inductive hypothesis as in the previous case, and in the second one the diagram is the following:
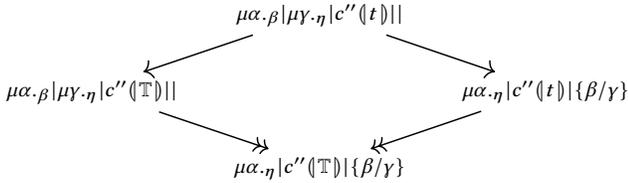
\[\begin{footnotesize}
\begin{tikzcd}[column sep=scriptsize]
& \mu\alpha.\name{\beta}{\mu\gamma.\name{\eta}{c''\hole{t}}} \ar[dl] \ar[dr] & & \\
\mu\alpha.\name{\beta}{\mu\gamma.\name{\eta}{c''\hole{\Sum T}}} \ar[->>, dr] & & \mu\alpha.\name{\eta}{c''\hole{t}}\set{\beta/\gamma} \ar[->>, dl] \\
& \mu\alpha.\name{\eta}{c''\hole{\Sum T}}\set{\beta/\gamma} & &
\end{tikzcd}
\end{footnotesize}\]
thanks to Remark \ref{lamu-rm:easyTakeOutSum} and Lemma \ref{lamu-lm:ForLocalConfluence}.

(4). Case $c=c'[\vec{u}]$.

Then $c\hole{t}=c'\hole{t}[\vec{u}]$ and the reduction $c\hole{t} \to_\mathrm{r} \Sum{T}_2$ can only be performed either by reducing $c'\hole{t}$, or reducing an element of $[\vec{u}]$, or the in case $c'\hole{t}$ is a $\lam$-abstraction or a $\mu$-abstraction and we reduce the $\lamu$-redex $c\hole{t}$. 

In the first case one can easily use inductive hypothesis.

In the second case one can easily write the diagram.

In the third case $c'\hole{t}$ is a $\lam$-abstraction.
Then either $c'=\xi$ and $t=\lam x.t'$, or $c'=\lam x.c''$.
But the first case is impossible, because by hypothesis $t\rightarrow^+_{\mathrm{base}^\mathrm{r}}\Sum T$, so $t$ cannot be a $\lam$-abstraction;
In the second case the diagram corresponds to the first diagram of the case 1 (with the notations used there, take $s:=c''\hole{t}$ and $\Sum S:=c''\hole{\Sum T}$).

In the fourth case $c'\hole{t'}$ is a $\mu$-abstraction.
Then either $c'=\xi$ and $t=\mu\alpha.\name{\beta}{t'}$, or $c'=\mu\alpha.\name{\beta}{c''}$.
But in the first situation, since by hypothesis $t\rightarrow^+_{\mathrm{base}^\mathrm{r}}\Sum T$, it must be $t=\mu\alpha.\name{\beta}{\mu\gamma.\name{\eta}{t''}}$ and $\Sum T=\mu\alpha.\name{\eta}{t''}\set{\beta/\gamma}$, and this situation has already been treated in the case $c=\xi$ (with the notations used there, it corresponds to the subcase $t=(\mu\alpha.\name{\beta}{s})[\vec{u}]$ and $\Sum T = \mu\alpha.\langle \name{\beta}{s} \rangle^+_\alpha[\vec{u}]$, where we consider the reduction of the $\rho$-redex\footnote{The diagrams are the last four of that subcase.}).
In the second case the diagram corresponds to the third diagram of the case 1 (with the notations used there, take $s:=c''\hole{t}$ and $\Sum S:=c''\hole{\Sum T}$).

(5). Case $c=v[c',\vec{u}]$.

Then $c\hole{t}=v[c'\hole{t},\vec{u}]$ and the reduction $c\hole{t} \to_r \Sum{T}_2$ can only be performed either by reducing $c'\hole{t}$, or reducing $v$, or reducing an element of $[\vec{u}]$, or in the case $v$ is a $\lam$-abstraction or a $\mu$-abstraction and we reduce the $\lamu$-redex $c\hole{t}$. 

In the first case one can trivially use the inductive hypothesis as done in the previous case.

In the second and third case one can easily write the diagram.

Let us look at the fourth case.

If $v$ is a $\lam$-abstraction, say $v=\lam x.v'$, then the diagram corresponds to the second diagram of the case 1 (with the notations used there, take $w:=c'\hole{t}$ and $\Sum W:=c'\hole{\Sum T}$).

If $v$ is a $\mu$-abstraction, say $v=\mu\alpha.\name{\beta}{v'}$, then the diagram corresponds to the fourth diagram of the case 1 (with the notations used there, take $w:=c'\hole{t}$ and $\Sum W:=c'\hole{\Sum T}$).
\qedhere
\end{proof}

\subsection{APPENDIX - SECTION 3.1}

We give proofs of Lemma~\ref{lm:TaylorBehavesSubst}, Proposition~\ref{prop:LiftAssSimulation}, Lemma~\ref{lm:ForInjectivity}, Proposition~\ref{lamu-prop:FIRSTPART}, Corollary~\ref{cor:Ass3+Prop+Cor}.

\begin{proof}[Proof of Lemma~\ref{lm:TaylorBehavesSubst}]
(1).
 Straightforward induction on $M$.

(2).
 Induction on $M$.
Nothing changes w.r.t.\ the proof one does in $\lam$-calculus, the only new case is $M=\mu\beta.\name{\alpha}{P}$ but it is done straightforwardly exactly as the case $M=\lam x.P$.

(3).
Induction on $M$.

Case $M=x$:
\[\begin{array}{rcl} 
  \Te{(M)_\alpha N}&=&\Te{(x)_\alpha N} \\
 &=&\set{x}\\
 &=&\langle x\rangle_\alpha 1\\
 &=&\bigcup\limits_{\vec{u}\in\,\fmsets{\Te{N}}} \langle x\rangle_\alpha[\vec{u}]\\
 &=&\bigcup\limits_{t\in\Te{M},\vec{u}\in\,\fmsets{\Te{N}}} \langle t\rangle_\alpha[\vec{u}].
\end{array}\]

Case $M=\lam x.P$:
\[
 \begin{array}{rcl}
  \Te{(M)_\alpha N} &
  = &
  \Te{\lam x.(P)_\alpha N} \\
   &
  = &
  \set{\lam x.s\stt s\in\Te{(P)_\alpha N}} \\
   &
  = & 
  \bigcup\limits_{p\in\Te{P}}\bigcup\limits_{[\vec{u}]\in\,\fmsets{\Te{N}}} \lam x.(\langle p\rangle_\alpha[\vec{u}]) \\
   &
  = &
  \bigcup\limits_{p\in\Te{P}}\bigcup\limits_{[\vec{u}]\in\,\fmsets{\Te{N}}} \langle \lam x.p\rangle_\alpha[\vec{u}] \\
   &
  = &
  \bigcup\limits_{t\in\Te{M}}\bigcup\limits_{[\vec{u}]\in\,\fmsets{\Te{N}}} \langle t\rangle_\alpha[\vec{u}].
\end{array}
\]
Case $M=\mu\beta.\name{\gamma}{P}$ (with $\beta,\gamma\neq\alpha$):
\[
 \begin{array}{rcl}
  \Te{(M)_\alpha N} &
  = &
  \Te{\mu\beta.\name{\gamma}{(P)_\alpha N}} \\
   & 
  = & 
  \set{\mu\beta.\name{\gamma}{s} \stt s\in\Te{(P)_\alpha N}} \\
   &
  = &
  \bigcup\limits_{p\in\Te{P}}\bigcup\limits_{[\vec{u}]\in\,\fmsets{\Te{N}}} \mu\beta.\name{\gamma}{\langle p\rangle_\alpha[\vec{u}]} \\
   &
  = &
  \bigcup\limits_{p\in\Te{P}}\bigcup\limits_{[\vec{u}]\in\,\fmsets{\Te{N}}} \langle\mu\beta.\name{\gamma}{p}\rangle_\alpha [\vec{u}] \\
   &
  = &
  \bigcup\limits_{t\in\Te{M}}\bigcup\limits_{[\vec{u}]\in\,\fmsets{\Te{N}}} \langle t\rangle_\alpha[\vec{u}].
\end{array}
\]
Case $M=\mu\beta.\name{\alpha}{P}$ (with $\beta\neq\alpha$):
\[
 \begin{array}{rcl}
  \Te{(M)_\alpha N} & 
  = &
  \Te{\mu\beta.\name{\alpha}{((P)_\alpha N)N}} \\
   & 
  = &
  \Big\{\mu\beta.\name{\alpha}{v[\vec{w}]} \stt [\vec{w}]\in\,\fmsets{\Te{N}},\\
&& \quad v\in\bigcup\limits_{p\in\Te{P},\vec{q}\in\,\fmsets{\Te{N}}} \langle p\rangle_\alpha[\vec{q}]\Big\} \\
   & 
  = &
  \Big\{\mu\beta.\name{\alpha}{v[\vec{w}]} \stt [\vec{w}]\in\,\fmsets{\Te{N}}, \ v\in \langle p\rangle_\alpha[\vec{q}], \\
&& \quad p\in\Te{P}, \ [\vec{q}]\in\,\fmsets{\Te{N}}\Big\} \\
   &
  = &
  \bigcup\limits_{p\in\Te{P}}\bigcup\limits_{[\vec{u}]\in\,\fmsets{\Te{N}}}\\
&&\quad
  \mu\beta.
  \sum\limits_{
  \begin{scriptsize}\begin{array}{c}
  ([\vec{w}],[\vec{q}])\\
  \textit{ w.c.\ of }[\vec{u}]
  \end{array}\end{scriptsize}
  }
  \name{\alpha}{(\lnamedapp{p}{\alpha}{[\vec{q}]})[\vec{w}]} \\
   &
  = &
  \bigcup\limits_{p\in\Te{P}}\bigcup\limits_{[\vec{u}]\in\,\fmsets{\Te{N}}}
  \mu\beta.\lnamedapp{\name{\alpha}{p}}{\alpha}{[\vec{u}]} \\
   &
  = &
  \bigcup\limits_{p\in\Te{P}}\bigcup\limits_{[\vec{u}]\in\,\fmsets{\Te{N}}}
  \lnamedapp{\mu\beta.\name{\alpha}{p}}{\alpha}{[\vec{u}]} \\
   &
  = &
  \bigcup\limits_{t\in\Te{M}}\bigcup\limits_{[\vec{u}]\in\,\fmsets{\Te{N}}} \langle t\rangle_\alpha[\vec{u}].
\end{array}
\]
Case $M=PQ$:
\[\small
 \begin{array}{rcl}
  \Te{(M)_\alpha N} & 
  = &
  \Te{((P)_\alpha N)((Q)_\alpha N)} \\
  \\
   & 
  = &
  \bigcup\limits_{n\in\N}
  \bigcup\limits_{p\in\Te{P}}\bigcup\limits_{[\vec{q}]\in\fmsets{\Te{Q}}}\bigcup\limits_{[\vec{s}^0],\dots,[\vec{s}^n]\in\fmsets{\Te{N}}} \\
  &&\quad \Big\{ v[w_1,\dots,w_n] \stt v\in\lnamedapp{p}{\alpha}{[\vec{s}^0]},\, w_i\in\lnamedapp{q_i}{\alpha}{[\vec{s}^i]} \Big\} 
  \\
  \\
   & 
  = &
  \bigcup\limits_{n\in\N}
  \bigcup\limits_{p\in\Te{P}}\bigcup\limits_{[\vec{q}]\in\fmsets{\Te{Q}}}
  \bigcup\limits_{[\vec{u}]\in\fmsets{\Te{N}}} \\
 &&\quad \sum\limits_{\begin{scriptsize}
  \begin{array}{c}
  ([\vec{s}^0],\dots,[\vec{s}^n])\\
  \textit{ w.c.\ of }[\vec{u}]
  \end{array}\end{scriptsize}
  }
  (\lnamedapp{p}{\alpha}{[\vec{s}^0]})[\dots,\lnamedapp{q_i}{\alpha}{[\vec{s}^i]},\dots] 
  \\ 
  \\
   &
  = &
  \bigcup\limits_{p\in\Te{P}}\bigcup\limits_{[\vec{q}]\in\fmsets{\Te{Q}}}
  \bigcup\limits_{[\vec{u}]\in\fmsets{\Te{N}}}
  \lnamedapp{p[\vec{q}]}{\alpha}{[\vec{u}]}
  \\
  \\
   &
  = &
  \bigcup\limits_{t\in\Te{M}}\bigcup\limits_{\vec{u}\in\,\fmsets{\Te{N}}} \langle t\rangle_\alpha[\vec{u}].
\end{array}\]\qedhere
\end{proof}

\vspace{0.3cm}

\begin{proof}[Proof of Proposition~\ref{prop:LiftAssSimulation}]
For the first part of the proof, we have three subcases, corresponding to the three base-cases of the reduction.

 Subcase $M=(\mu\alpha.\name{\beta}{P})Q, \ N=\mu\alpha.(\name{\beta}{P})_\alpha Q$.

 (1).
If $s\in\Te{(\mu\alpha.\name{\beta}{P})Q}$ then $s=(\mu\alpha.\name{\beta}{p})[\vec{q}]$ for $p\in\Te{P}$ and $[\vec{q}]\in\,\fmsets{\Te{Q}}$.
So $s\rightarrow_{\mathrm{r}} \mu\alpha.\langle \name{\beta}{p} \rangle_\alpha [\vec{q}]\subseteq \Te{\mu\alpha.(\name{\beta}{P})_\alpha Q}$ thanks to Lemma \ref{lm:TaylorBehavesSubst}.

(2). 
If $s'\in\Te{\mu\alpha.(\name{\beta}{P})_\alpha Q}$ then thanks to Lemma~\ref{lm:TaylorBehavesSubst}, $s'\in \mu\alpha.\langle \name{\beta}{p}\rangle_\alpha [\vec{q}]$ for $p\in\Te{P}$ and $[\vec{q}]\in\,\fmsets{\Te{Q}}$.
So $\Te{(\mu\alpha.\name{\beta}{P})Q} \ni (\mu\alpha.\name{\beta}{p})[\vec{q}]\rightarrow_{\mathrm{r}} \mu\alpha.\langle \name{\beta}{p}\rangle_\alpha [\vec{q}] \ni s'$, and $\mu\alpha.\langle \name{\beta}{p}\rangle_\alpha [\vec{q}] \subseteq \Te{\mu\alpha.(\name{\beta}{P})_\alpha Q}$.

Subcase $M=\mu\gamma.\name{\alpha}{\mu\beta.\name{\eta}{P}}, \ N=\mu\gamma.\name{\eta}{P}\set{\alpha/\beta}$.
(1) and (2) are straightforward using Lemma \ref{lm:TaylorBehavesSubst} as above.

Subcase $M=(\lam x.P)Q, \ N=P\set{Q/x}$.
(1) and (2) are straightforward using Lemma \ref{lm:TaylorBehavesSubst} as above.

The ``furthermore'' is by induction on the single-hole context $C$ s.t.\ $M=C\hole{M'}$, $N=C\hole{N'}$ and $M'\rightarrow_\mathrm{base} N'$ (such a $C$ exists because $\rightarrow$ is contextual).
The base case $C=\xi$ coincides with the above case.
The step cases are straightforwardly done  as one does in $\lam$-calculus, treating the case of $\mu$-abstraction exactly as a $\lam$-abstraction.
\end{proof}

\vspace{0.3cm}

\begin{proof}[Proof of Lemma~\ref{lm:ForInjectivity}]
Such a statement may seem strange at first sight, because one would expect (3) (which is the item used in the proof of Theorem~\ref{thm:Injectivity}, together with (1)) to be an indutive step of (2).
However, (3) is \emph{not} an inductive step of (2), simply because $\name{\eta}{p}\notin\lam\mu^{\mathrm{r}}$.
This is due to the fact that we are in $\lam\mu$-calculus and not in Saurin's $\Lambda\mu$-calculus.
One could be then tempted to state it in $\Lambda\mu$-calculus, so only with item (2).
In this case, since $p$ would be in $\Lambda\mu$ and not just in $\lamu$, (3) would in fact be an inductive step of (2), but still this is not what we need: in fact when we use (3) in the present paper, we want $p$ to be in $\lamu$, something which is not guaranteed by the statement in $\Lambda\mu$.
Let us now prove the Lemma.

 (1).
 Induction on $P$, similar to as it is done in \cite{DBLP:journals/tcs/EhrhardR08}.

(2).
 Induction on $P$. Let us see the case $P=\mu\alpha.\name{\beta}{M}$, the other cases being similar.
 Then $p=\mu\alpha.\name{\beta}{s}$ and $p'=\mu\alpha.\name{\beta}{s'}$, with $s,s'\in\Te{M}$.
 Let $h\in\,\lnamedapp{p}{\gamma}{[\vec{d}]}\cap \lnamedapp{p'}{\gamma}{[\vec{d}\,']}$.
 Choosing $\alpha\neq\gamma$, we have two subcases.

 Subcase $\beta\neq\gamma$: then $h=\mu\alpha.\name{\beta}{h_0}$ with $h_0\in \lnamedapp{s}{\gamma}{[\vec{d}]}\cap \lnamedapp{s'}{\gamma}{[\vec{d}\,']}$.
 We can easily conclude by inductive hypothesis.

 Subcase $\beta=\gamma$: then $h=\mu\alpha.\name{\gamma}{h_0[\vec{v}]}$, with $h_0\in \lnamedapp{s}{\gamma}{[\vec{w}]}\cap \lnamedapp{s'}{\gamma}{[\vec{w}\,']}$ and
 where $([\vec{v}],[\vec{w}])$ and $([\vec{v}],[\vec{w}\,'])$ are w.c.\ respectively of $[\vec{d}]$ and of $[\vec{d}\,']$.
 Thus by inductive hypothes we have $s=s'$, i.e. $p=p'$, and also $[\vec{w}]=[\vec{w}\,']$.
 Finally, $[\vec{d}]=[\vec{w}]*[\vec{v}]=[\vec{w}\,']*[\vec{v}]=[\vec{d}\,']$.
 
(3).
 Let $\alpha\neq\gamma$ and fresh. Then
 $\lnamedapp{\name{\eta}{p}}{\gamma}{[\vec{d}]}\cap \lnamedapp{\name{\eta}{p'}}{\gamma}{[\vec{d}\,']}\neq\emptyset$
 iff
 $\mu\alpha.\lnamedapp{\name{\eta}{p}}{\gamma}{[\vec{d}]}\cap \mu\alpha.\lnamedapp{\name{\eta}{p'}}{\gamma}{[\vec{d}\,']}\neq\emptyset$.
 The latter condition is
 $\lnamedapp{\mu\alpha.\name{\eta}{p}}{\gamma}{[\vec{d}]}\cap\lnamedapp{\mu\alpha.\name{\eta}{p'}}{\gamma}{[\vec{d}\,']}$, so it is of the shape considered by (2), and $\mu\alpha.\name{\eta}{p},\mu\alpha.\name{\eta}{p'}\in\lamu$. Therefore (2) gives $\mu\alpha.\name{\eta}{p}=\mu\alpha.\name{\eta}{p'}$, i.e. $p=p'$, as well as $[\vec{d}]=[\vec{d}\,']$.\qedhere
\end{proof}

\vspace{0.3cm}

\begin{proof}[Proof of Proposition~\ref{lamu-prop:FIRSTPART}]
We have the three base-case reductions:

Case $t=(\mu\alpha.\name{\beta}{p})[\vec{q}]$ and $\Sum{T}'=\mu\alpha.\langle \name{\beta}{p}\rangle_\alpha [\vec{q}]$.
So it must be $M=(\mu\alpha.\name{\beta}{P})Q$ with $p\in\Te{P}$ and $[\vec{q}]\in\,\fmsets{\Te{Q}}$. Now thanks to Lemma \ref{lm:TaylorBehavesSubst} we can take $N:=\mu\alpha.(\name{\beta}{P})_\alpha Q$.

Case $t=(\lam x.p)[\vec{q}]$ and $\Sum{T}'=p\langle[\vec{q}]/x\rangle$, or case $t=\mu\gamma.\name{\alpha}{\mu\beta.\name{\eta}{p}}$ and $\Sum T'=\mu\gamma.\name{\eta}{p}\set{\alpha/\beta}$.
Exactly as above, thanks to Lemma \ref{lm:TaylorBehavesSubst}.
\end{proof}

\vspace{0.3cm}

We now turn to Corollary~\ref{cor:Ass3+Prop+Cor}.
It relies on a result that we did not report in the main paper: Proposition~\ref{prop:ReduceSum}, which appears in the Appendix A.5.
It is the generalisation of Proposition~\ref{lamu-prop:FIRSTPART} mentioned in the proof of Corollary~\ref{cor:Ass3+Prop+Cor} in the main paper.
This result requires some discussion, so we postpone its proof and dedicate the whole Appendix A.5. to it.
In the meantime, below, we prove Corollary~\ref{cor:Ass3+Prop+Cor} using Proposition~\ref{prop:ReduceSum}.

\begin{proof}[Proof of Corollary~\ref{cor:Ass3+Prop+Cor}]
We prove, by induction on $\ell\in\N$, that if $\ell$ is the length of a maximal reduction (let's call it $\phi$) from a sum $\Sum{T}$ to its normal form $\nf[\mathrm{r}]{\Sum T}$, then the statement of Corollary \ref{cor:Ass3+Prop+Cor} holds.
If $\ell = 0$ then just take $N := M$.
If $\ell\geq1$, then $\phi$ factorizes as $\Sum T\to_{\mathrm{r}} \Sum T'\msto[\mathrm{r}] \nf[\mathrm{r}]{\Sum{T}}$ for some sum $\Sum T'$.
Since $\Sum T\in\Te{M}$ by hypothesis, by Proposition~\ref{prop:ReduceSum} there exist $N'\in\lamu$ and $\Sum T''\subseteq\Te{N'}$ such that $M\to N'$ and $\Sum T'\msto[\mathrm{r}] \Sum T''$. Let $k\geq0$ be the length of this last reduction. By confluence, we have also $\Sum T''\msto[\mathrm{r}]\nf[\mathrm{r}]{\Sum T}$. Now take the maximal reduction from $\Sum T''$ to $\nf[\mathrm{r}]{\Sum T''}=\nf[\mathrm{r}]{\Sum T}$ and let $\ell'$ be its length. Due to the maximality of $\phi$, it must be $\ell'+k+1\leq\ell$, so $\ell'<\ell$, and now we can apply the inductive hypothesis to $\ell'$ (because it is the maximal length of a reduction between a sum $\Sum{T''}$ and its normal form). Since we already found that $\Sum{T''}\subseteq\Te{N'}$, we get an $N\in\lamu$ such that $M\to N'\msto N$ and $\nf[\mathrm{r}]{\Sum T}=\nf[\mathrm{r}]{\Sum T''}\subseteq\Te{N}$. 
\end{proof}

\subsection{APPENDIX - SECTION 3.2}

We give proofs of Lemma~\ref{lamu-lm:Reviewer} and of Lemma~\ref{lamu-lm:TaylorCommutesHead}.

\begin{proof}[Proof of Lemma~\ref{lamu-lm:Reviewer}]
If $t$ is a hnf, we can easily conclude: since any eventual bag of $s$ is empty, and since reductions cannot erase non-empty bags (because \emph{linear}), the fact that $s\in\nf[\mathrm{r}]{t}$ entails that already $t$ contains only empty bags.
But in a hnf the reduction can only take place inside some bag, so it actually must be $s=t$.
Therefore, $s \in s = t = \mathrm{H}^0(t)$ and we are done taking $n:=0$.

So we are left with the case in which $t$ is \emph{not} hnf.
In this case we know that $t \to_{\mathrm{r}} \mathrm{H}(t)$.
By confluence, $\mathrm{H}(t) \msto[\mathrm{r}] \nf[\mathrm{r}]{t}$.
But since $s\in\nf[\mathrm{r}]{t}$, there is some $t_1\in\mathrm{H}(t)$ s.t.\  $s\in\nf[\mathrm{r}]{t_1}$.
Now we can reason as in the beginning, splitting in two cases:
either $t_1$ is a hnf, in which case we reason exactly as in the first four lines of the proof: by linearity we get $s = t_1 \in \mathrm{H}(t)$, and we are done taking $n:=1$.
Or $t_1$ is not a hnf.
In this case we can reason again as before, obtaining a $t_1 \to_{\mathrm{r}} \mathrm{H}(t_1)$, $\mathrm{H}(t_1) \msto[\mathrm{r}] \nf[\mathrm{r}]{t_1}$ and a $t_2\in\mathrm{H}(t_1)$ s.t.\ $s\in\nf[\mathrm{r}]{t_2}$.
We can keep going with the same splits: if $t_2$ is hnf, we are done taking $n:=2$; if $t_2$ is not hnf, we obtain a new $t_3$ as before.
Now, the generation of such a new $t_{i+1}$ from the previously generated $t_i$ cannot continue forever: we claim that there must be some $m\in\N$ for which $t_m$ is hnf. 
If this is the case, the proof is concluded because we can take $n:=m$ as already mentioned.
To see that such an $m$ does exist, one simply remarks that at each time we have $t_{i} \to_{\mathrm{r}} \mathrm{H}(t_i)$ and $t_{i+1}\in\mathrm{H}(t_i)$.
But the well-founded measure $\ms{\cdot}$ is strictly decreasing along reductions, which precisely means that we obtain the strictly decreasing sequence:
\[\ms{t} > \ms{t_1} > \ms{t_2} > \cdots \]
Therefore, it must terminate (because the order is well-founded) on some $\ms{t_m}$, for some $m\in\N$, and we are done as already explained.
\end{proof}

The previous proof could clearly be given in an inductive way, but we chose to give it in this way to (maybe) let the reader better see the reasoning.

\vspace{0.2cm}

\begin{proof}[Proof of Lemma~\ref{lamu-lm:TaylorCommutesHead}]
 We have to show that $\Te{\mathrm{H}(M)}=\bigcup\limits_{t\in\Te{M}} \mathrm{H}(t)$.
 By Lemma \ref{lamu-lm:writingM} we know that: \[M=\lam \vec{x_1}.\mu\alpha_1.\name{\beta_1}{\dots\lam \vec{x_k}.\mu\alpha_k.\name{\beta_k}{RQ_1\dots Q_n}}\]
with the condition that either there is a $\rho$-redex in the head of $M$, or there is no such $\rho$-redex and $R$ is \emph{not} a variable (thus $R$ is either a $\lam$-redex or a $\mu$-redex). 
We have just said in a different fashion that $M$ is \emph{not} a hnf.
 Now let us show the two inclusions.

($\subseteq$).
 Take $s\in\Te{\mathrm{H}(M)}$. We have three cases:
 
 Case in which there is a $\rho$-redex in the head of $M$. Therefore there is also a leftmost $\rho$-redex $\dots\name{\beta_i}{\mu\alpha_{i+1}.\dots}$ in the head of $M$.
 Then $\mathrm{H}(M)$ is the term: 
 \[\lam \vec{x_1}.\mu\alpha_1.\name{\beta_1}{\dots\lam \vec{x_i}.\mu\alpha_i.\name{\beta_{i+1}}{\dots
 \lam\vec{x}_{k}\mu\alpha_{k}.\name{\beta_k}{R\vec{Q}}
 }\set{\beta_i/\alpha_{i+1}}
 }.\]
 So $s$ is the term:
\[\begin{small}\lam \vec{x_1}.\mu\alpha_1.\name{\beta_1}{\dots\lam \vec{x_i}.\mu\alpha_i.\name{\beta_{i+1}}{\dots
 \lam\vec{x}_{k}\mu\alpha_{k}.\name{\beta_k}{r[\vec{q}^1]\dots[\vec{q}^n]}
 }\set{\beta_i/\alpha_{i+1}}
 }\end{small}\]
 for $r\in\Te{R}$ and $[\vec{q}^{\,i}]\in\fmsets{\Te{Q_i}}$.
 But $\Te{M}$ contains the element: \[\lam \vec{x_1}.\mu\alpha_1.\name{\beta_1}{\dots\lam \vec{x_k}.\mu\alpha_k.\name{\beta_k}{r[\vec{q}^1]\dots[\vec{q}^n]}} =: t\] and thus $s=\mathrm{H}(t)$, i.e. $s\in\bigcup\limits_{t\in\Te{M}} \mathrm{H}(t)$.
 
 Case there are no $\rho$-redexes and $R = (\mu \gamma.\name{\eta}{P})D$.
 Then $\mathrm{H}(M)=\vec{\lamu}.\mid (\mu\gamma.\namedapp{\name{\eta}{P}}{\gamma}{D})\vec{Q} \mid$.
 So by Lemma \ref{lm:TaylorBehavesSubst} $s\in\vec{\lamu}.\mid (\mu\gamma.\lnamedapp{\name{\eta}{p}}{\gamma}{[\vec{d}]})[\vec{q}^1]\dots[\vec{q}^n] \mid$ for a $p\in\Te{P}$, $[\vec{d}]\in\fmsets{\Te{D}}$ and $[\vec{q}^{\,i}]\in\,\fmsets{\Te{Q_i}}$.
 So $s\in\mathrm{H}(t)$, with 
 $t$ begin the term: \[\lam \vec{x_1}.\mu\alpha_1.\name{\beta_1}{\dots\lam \vec{x_k}.\mu\alpha_k.\name{\beta_k}{(\mu \gamma.\name{\eta}{p})[\vec{d}][\vec{q}^1]\dots[\vec{q}^n]}}\] which belongs to $\Te{M}$, i.e. $s\in\bigcup\limits_{t\in\Te{M}} \mathrm{H}(t)$.

 Case there are no $\rho$-redexes and $R = (\lam x.P)D$.
 Exactly as above using Lemma \ref{lm:TaylorBehavesSubst}.
 
($\supseteq$).
 One can follow the exact same kind of argument as before: the fact that Taylor expansion preserves the structure of the term, plus Lemma \ref{lm:TaylorBehavesSubst}, is what makes us able to transport one step of the head reduction from terms to resource terms.\qedhere
\end{proof}

\subsection{APPENDIX - Rigids associated with a resource term}\label{sec:rigid}

This appendix section does not correspond to a section in the main paper.
Its aim is to prove Proposition~\ref{prop:ReduceSum}, a result that is needed in the proof of Corollary~\ref{cor:Ass3+Prop+Cor}, as previously explained in Appendix A.3.

We report here a proof involving the ``rigid resource terms'': words built in the same way as resource terms but taking lists instead of bags (multiset of terms).

We present these constructions in detail also because they are needed in the proof of Stability (Theorem~\ref{lamu-th:TeStability}), which is done exactly as in \cite[Theorem 5.11]{DBLP:journals/pacmpl/BarbarossaM20}, and whose detailed proof is given in the next Appendix A.6 for the seek of completeness; in particular, the following Lemma \ref{lm:NeedInStability} and Lemma \ref{lm:Te_of_contexts}, involving these constructions, are used in the proof of Theorem~\ref{lamu-th:TeStability}.

\begin{Definition}
The set of \emph{rigid terms} is defined by:
\[
 t::=
 x \mid
 \lam x.t \mid
 \mu\alpha.\name{\beta}{t} \mid
 t\langle t\dots,t\rangle \qquad \textit{(here $\langle \dots \rangle$ means a list).}
\]
The set of rigid $k$-context is defined as expected adding the clause ``$\xi_1 \mid \dots \mid \xi_k$'' for the holes.
\end{Definition}

\begin{Definition}
Let $c$ be a resource-$k$-context. We define a set $\mathrm{Rigid}(c)$ of \emph{rigid $k$-contexts}, whose elements are called the \emph{rigids of $c$}, by induction on $c$ as follows:
\begin{enumerate}
 \item $\mathrm{Rigid}(\xi_i)=\set{\xi_i}$
 \item $\mathrm{Rigid}(x)=\set{x}$
 \item $\mathrm{Rigid}(\lam x.c_0)=\set{\lam x.\linc_0\mid \linc_0\in\mathrm{Rigid}(c_0)}$
 \item $\mathrm{Rigid}(\mu\alpha.\name{\beta}{c_0})=\set{\mu\alpha.\name{\beta}{\linc_0}\mid \linc_0\in\mathrm{Rigid}(c_0)}$
 \item $\mathrm{Rigid}(c_0[c_1,\dots,c_k])=\set{\,\linc_0\langle\linc_{\sigma(1)}\dots,\linc_{\sigma(k)}\rangle
 \mid
 \linc_i\in\mathrm{Rigid}(c_i)\textit{ and }
 \sigma \textit{ permutation on }k \textit{ elements }}$.
\end{enumerate}
\end{Definition}

The rigids of a resource $k$-context $c$ are the rigid $k$-contexts which can be ``canonically associated'' with $c$.
Rigid-resource-calculus has been considered, e.g., in \cite{DBLP:journals/corr/abs-2008-02665} and its study sheds light on the combinatorial role of the factorial coefficients in the full (i.e. quantitative) Taylor expansion of a term.
Our set $\mathrm{Rigid}(c)$ is the preimage $F^{-1}(c)$ of $c$ under the surjection $F$\footnote{In \cite{DBLP:journals/corr/abs-2008-02665} $F$ is called $\parallel .\parallel$, but here write $F$, for ``forgetful''.} from rigid contexts to resource context simply forgetting the order of the lists.
Its graph is what, in \cite{DBLP:journals/corr/abs-2008-02665}, is called the ``representation relation''.
In the following definition we precisely operate such a forgetful operation, but in addition we consider terms filling the holes.

\begin{Definition}\label{def:rigidtores}
Let $\linc$ be a rigid of a resource-$k$-context $c$ and, for $i=1,\dots,k$, let $\vec{v}^{\,i}:=\langle v_1^i,\dots,v_{\dg{\xi_i}{c}}^i\rangle$ be a list\footnote{If $\dg{\xi_i}{c}=0$ we mean the empty list.} of resource terms.
We define, by induction on $c$, a resource term $\linc\hole{\vec{v}^1,\dots,\vec{v}^k}$ as follows:
\begin{enumerate}
 \item If $c=\xi_i$ then $\linc=\xi_i$; we set $\linc\hole{\langle\rangle,\dots,\langle\rangle,\langle v^i_1\rangle,\langle\rangle,\dots,\langle\rangle}$ to be $v^i_1$.
 \item If $c=x$ then $\linc=x$; we set $\linc\hole{\langle\rangle,\dots,\langle\rangle}:=x$.
 \item If $c=\lam x.c_0$ then $\linc=\lam x.\linc_0$ where $\linc_0$ is a rigid of $c_0$; we set $\linc\hole{\vec{v}^1,\dots,\vec{v}^k}=\lam x.\linc_0\hole{\vec{v}^1,\dots,\vec{v}^k}$.
 \item If $c=\mu\alpha.\name{\beta}{c_0}$ then $\linc=\mu\alpha.\name{\beta}{\linc_0}$ where $\linc_0$ is a rigid of $c_0$; we set $\linc\hole{\vec{v}^1,\dots,\vec{v}^k}=\mu\alpha.\name{\beta}{\linc_0\hole{\vec{v}^1,\dots,\vec{v}^k}}$.
 \item If $c=c_0[c_1,\dots,c_n]$, then $\linc=\linc_0\langle \linc_{\sigma(1)},\dots,\linc_{\sigma(n)} \rangle$ where $\linc_i$ is a rigid of $c_i$.
 So each list $\vec{v}^{\,i}$ factorizes as a concatenation $\vec{w}^{\,i0}\vec{w}^{\,i1}\cdots\vec{w}^{\,in}$ of lists where $\vec{w}^{\,ij}$ has exactly $\dg{\xi_i}{c_j}$ elements\footnote{This is a concise way of saying that we take $\vec{w}^{\,i1}:=\langle v^i_1,\dots,v^i_{\dg{\xi_i}{c_1}} \rangle$, $\vec{w}^{\,i2}:=\langle v^i_{1+\dg{\xi_i}{c_1}},\dots,v^i_{\dg{\xi_i}{c_2}+\dg{\xi_i}{c_1}} \rangle$ etc.}; we set
$\linc\hole{\vec{v}^1,\dots,\vec{v}^k}$ to be the term:
\[\begin{array}{c}
\linc_0\hole{\vec{w}^{10},\dots,\vec{w}^{k0}}[\linc_{\sigma(1)}\hole{\vec{w}^{\,11},\dots,\vec{w}^{\,k1}},\\
\qquad\dots,\linc_{\sigma(n)}\hole{\vec{w}^{\,1n},\dots,\vec{w}^{\,kn}}].
 \end{array}\]
\end{enumerate}
\end{Definition}

\begin{Remark}
One clearly has that if $v\msto[r]\Sum V$ then:
 \[\linc\hole{\dots,\langle \dots,v,\dots \rangle,\dots}\]
  $\msto[r]$-reduces to:
  \[\sum\limits_{w\in\Sum V}
  \linc\hole{\dots,\langle \dots,w,\dots \rangle,\dots}=:\linc\hole{\dots,\langle \dots,\Sum V,\dots \rangle,\dots}.
 \]
\end{Remark}

Let us extend the definition of Taylor expansion to resource $k$-contexts by adding, in its definition, the clause:
\[
 \Te{\xi_i}:=\set{\xi_i}.
\]
It is clear that is $C$ is a $k$-context then all elements of $\Te{C}$ are resource $k$-contexts.

In the following, if $\vec{v}$ is a list, we denote with $[\vec{v}]$ the multiset associated with $\vec{v}$ (same elements but disordered).

\begin{Lemma}\label{lm:NeedInStability}
 Let $C$ be a $k$-context and $c_1,c_2\in\Te{C}$.
 Let $\linc_1$ and $\linc_2$ rigids respectivly of $c_1$ and $c_2$.
 For $i=1\dots,k$, let $\vec{v}^i=\langle v^i_1,\dots,v^i_{\dg{\xi_i}{c_1}} \rangle$ and $\vec{u}^i=\langle v^i_1,\dots,v^i_{\dg{\xi_i}{c_2}} \rangle$ be lists of resource terms.
 If $\linc_1\hole{\vec{v}^1,\dots,\vec{v}^k}=\linc_2\hole{\vec{u}^1,\dots,\vec{u}^k}$ then $c_1=c_2$ and $[\vec{v}^{\,i}]=[\vec{u}^{\,i}]$ for all $i$.
\end{Lemma}
\begin{proof}
 Induction on $C$.
 
 Case $C=\xi_i$.
 Then $c_1=\xi_i=c_2$, and $\vec{v}^{\,i}=\langle v^{i1}\rangle$, $\vec{u}^{\,i}=\langle u^{i1}\rangle$ and $\vec{v^j}=\langle\rangle=\vec{u^j}$ for $j\neq i$.
 So $v^{i1}=\linc_1\hole{\vec{v}^1,\dots,\vec{v}^k}=\linc_2\hole{\vec{u}^1,\dots,\vec{u}^k}=u^{i1}$.
 
 Case $C=x$.
 Trivial.
 
 Case $C=\lam x.C_0$ and case $C=\mu\alpha.\name{\beta}{C_0}$.
 Trivial by inductive hypothesis.
 
 Case $C=C'C''$.
 Then, for $i=1,2$, one has $c_i=c_{i0}[c_{i1},\dots,c_{in_{i}}]$ with $c_{i0}\in\Te{C'}$ and $c_{ij}\in\Te{C''}$ for $j\geq 1$.
 So $\linc_i=\linc_{i0}\langle \linc_{i\sigma_i{(1)}},\dots,\lam x.\linc_{i\sigma_i{(n_{i})}}\rangle$ where $\sigma_i$ is a permutation on $n_{i}$ elements.
 So $\linc_1\hole{\vec{v}^1,\dots,\vec{v}^k}$ is the term:
 \[\begin{array}{c}
 \linc_{10}\hole{\vec{w}^{110},\dots,\vec{w}^{1k0}}[\,\linc_{1\sigma_1{(1)}}\hole{\vec{w}^{111},\dots,\vec{w}^{1k1}},\\
\qquad\dots,\linc_{1\sigma_1{(n_1)}}\hole{\vec{w}^{11{n_1}},\dots,\vec{w}^{1k{n_1}}}]
\end{array}\]
and $\linc_2\hole{\vec{v}^1,\dots,\vec{v}^k}$ is the term:
\[\begin{array}{c}
 \linc_{20}\hole{\vec{w}^{210},\dots,\vec{w}^{2k0}}[\,\linc_{2\sigma_2{(1)}}\hole{\vec{w}^{211},\dots,\vec{w}^{2k1}},\\
\qquad\dots,\linc_{2\sigma_2{(n_2)}}\hole{\vec{w}^{21{n_2}},\dots,\vec{w}^{2k{n_2}}}]
\end{array}\]
where the concatenation $\vec{w}^{1j1}\cdots\vec{w}^{1j{n_1}}$ gives $\vec{v}^{j}$ and the concatenation $\vec{w}^{2j1}\cdots\vec{w}^{2j{n_2}}$ gives $\vec{u}^{j}$.
 From $\linc_1\hole{\vec{v}^1,\dots,\vec{v}^k}=\linc_2\hole{\vec{u}^1,\dots,\vec{u}^k}$ we get that: $n_1=n_2=:n$, that:
 \begin{equation}\label{eq:c0}
 \linc_{10}\hole{\vec{w}^{110},\dots,\vec{w}^{1k0}}=\linc_{20}\hole{\vec{w}^{210},\dots,\vec{w}^{2k0}}
 \end{equation}
 and that there exist a permutation $\rho$ on $n$ elements which identifies each term of the writte bag of $\linc_1\hole{\vec{v}^1,\dots,\vec{v}^k}$ with the respective one of the writtten bag of $\linc_2\hole{\vec{u}^1,\dots,\vec{u}^k}$. That is, for all $h=1,\dots,n$, one has:
 \begin{equation}\label{eq:ci}
 \linc_{1\sigma_1(h)}\hole{\vec{w}^{\,1\,1\,
 {h}
 },\dots,\vec{w}^{\,1\,k\,{h}}}
 =
 \linc_{2\sigma_2(\rho(h))}\hole{\vec{w}^{\,2\,1\,{\rho(h)}},\dots,\vec{w}^{\,2\,k\,{\rho(h)}}}.
 \end{equation}
 Now the inductive hypothesis on (\ref{eq:c0}) gives $\linc_{10}=\linc_{20}$ as well as $\vec{w}^{\,1\,i\,0
 }=\vec{w}^{\,2\,i\,0}$, and the inductive hypothesis on (\ref{eq:ci}) gives, at the end of the day, $[\linc_{11},\dots,\linc_{1n}]=[\linc_{21},\dots,\linc_{2n}]$ as well as $[\vec{w}^{1j1}\cdots\vec{w}^{1j{n}}]=[\vec{w}^{2j1}\cdots\vec{w}^{2j{n}}]$ for $j\geq 1$.
Putting these things together, we have the desired result.
\end{proof}

\begin{Lemma}\label{lm:Te_of_contexts}
\begin{enumerate}
  \item\label{lem:Te_of_contexts1} 
  Let $C$ be a $k$-context.\\
  Then $\Te{C\hole{M_1,\dots,M_k}}$ is the set:
  \[\begin{array}{c}
\set{\linc\hole{\vec{s}^{\,1},\dots,\vec{s}^{\,k}} \stt c\in\Te{C}, \ \linc \textit{ rigid of }c\\
\qquad\textit{ and }\vec{s}^{\,i}\textit{ list of elements of }\Te{M_i}}.
\end{array}\]
  \item\label{lem:Te_of_contexts3}
  Let $c=c\hole{\xi}$ single-hole resource context, $M\in\lamu$ and $s_0\in\lamu^{\mathrm{r}}$.
  If $c\hole{s_0}\in\Te{M}$, then there is a context $C=C\hole{\xi}$, an $N\in\lamu$, a resource context $\widetilde{c}\in\Te{C}$, a rigid $\widetilde{c}^{\,\bullet}$ of $c$ and $s_1\dots,s_{\mathrm{deg}_{\xi}{\widetilde{c}}-1}\in\Te{N}$ s.t.\
  \begin{enumerate}
   \item $M=C\hole{N}$
   \item $s_0\in\Te{N}$
   \item $c\hole{t}=\widetilde{c}^{\,\bullet}\hole{\langle t,s_1\dots,s_{\dg{\xi}{\widetilde{c}}-1}\rangle }$ for all $t\in\lamu^{\mathrm{r}}$.
  \end{enumerate}
\end{enumerate}
\end{Lemma}
\begin{proof}
(1).
 Straightforward induction on $C$.
 
\noindent
(2). Induction on the single-hole resource context $c$.

\noindent Case $c=\xi$. Take $C:=\xi$, $N:=M$, $\widetilde{c}:=\xi$, $\widetilde{c}^{\,\bullet}:=\xi$ (there are no $s_i$'s because $\dg{\xi}{\widetilde{c}}-1=0$).

\noindent Case $c=\mu\alpha.\name{\beta}{c_1}$.
Since $c\hole{s_0}\in\Te{M}$ then $M=\mu\alpha.\name{\beta}{M_1}$ with $c_1\hole{s_0}\in\Te{M_1}$. We can thus take $C:=\mu\alpha.\name{\beta}{C_1}$, $\widetilde{c}:=\mu\alpha.\name{\beta}{\widetilde{c_1}}$, $\widetilde{c}^{\,\bullet}:=\mu\alpha.\name{\beta}{\widetilde{c_1}^\bullet}$, 
where $C_1,N,\widetilde{c_1},\widetilde{c}_1^{\,\bullet}$ and $s_1,\dots,s_{\dg{\xi}{\widetilde{c}_1}}$ are given by the inductive hypothesis.

\noindent Case $c=\lam x.c_1$. Exactly as the case of $\mu$-abstraction.

\noindent Case $c=c'[\vec{u}]$. Analogous as above.

\noindent Case $c=u[c',u_1,\dots,u_n]$.
Since $u[c'\hole{s_0},\vec{u}]=c\hole{s_0}\in\Te{M}$ then $M$ must have shape $M=PQ$ with $u\in\Te{P}$ and $c'\hole{s_0},u_i\in\Te{Q}$.
By induction hypothesis, we can write $Q=C_0\hole{N}$ for an appropriate context $C_0$ and $N\in\lamu$ s.t.\
$s_0\in\Te{N}$, and also there is a resource context $c_0\in\Te{C_0}$ and $\linc_0$ a rigid of $c_0$, together with a list $\vec{s}^{\,0}:=\langle s_1^0,\dots,s_{\dg{\xi}{c_0}-1}^0\rangle$ of elements of $\Te{N}$, such that $c'\hole{t}=\linc_0\hole{\langle t,\vec{s}^{\,0}\rangle}$ for all $t\in\lamu^{\mathrm{r}}$.
But $u_i\in\Te{Q}=\Te{C_0\hole{N}}$, so by the point (1), each $u_i$ have shape $u_i=\linc_i\hole{\vec{s}^{\,i}}$ for an appropriate $\linc_i$ rigid of some $c_i\in\Te{C_0}$ and some $\vec{s}^{\,i}$ list of elements of $\Te{N}$. 
Now remark that $M=C\hole{N}$ for $C:=PC_0$.
Moreover, putting $\widetilde{c}:=u[c_0,c_1,\dots,c_n]\in\Te{C}$ and choosing its rigid $\widetilde{c}^{\,\bullet}:=u^\bullet\langle \linc_0,\linc_1,\dots,\linc_k\rangle$ (which rigid $u^\bullet$ of $u$ one choses does not matter), we have:
$c\hole{t}=u[c'\hole{t},u_1\dots,u_n]
=
u[
\linc_0\hole{\langle t,\vec{s}\rangle}
,\linc_1\hole{\vec{s}^{\,1}}
,\dots,
\linc_k\hole{\vec{s}^{\,k}}
]=
\widetilde{c}^{\,\bullet}\hole{\langle t,\vec{s}\rangle}$ for all $t\in\lam^{\mathrm{r}}$, where the list $\vec{s}$ is the concatenation $\vec{s}^{\,0}\vec{s}^{\,1}\cdots\vec{s}^{\,n}$ is a list of $\dg{\xi}{\widetilde{c}}-1$ elements of $\Te{N}$.
\end{proof}

We can now finally prove the full Proposition~\ref{prop:ReduceSum}.

\begin{Proposition}\label{prop:ReduceSum}
If $\Te{M}\supseteq\Sum{T}\rightarrow_\mathrm{r}\Sum{T}'$ then there is $N\in\lamu$ and a sum $\tilde{\Sum{T}}\subseteq\Te{N}$ s.t.\ $M\rightarrow N$ and $\Sum{T}'\msto[\mathrm{r}]\tilde{\Sum{T}}$.
\end{Proposition}
\begin{proof}
Saying that $\Sum{T}\rightarrow_{\mathrm{r}}\Sum{T}'$ means that $\Sum T$ has shape $\Sum T=\sum\limits_{i}t_i+c\hole{h}$ and $\Sum{T}'$ has shape $\Sum{T}'=\sum\limits_{i}t_i+c\hole{\mathcal{H}}$, for some single-hole resource context $c$, a resource term $h$ and a sum $\mathcal{H}$ s.t.\ $h\to_{\mathrm{base}}\mathcal{H}$.
But since $c\hole{h}\in\Sum T\subseteq\Te{M}$, by Lemma \ref{lm:Te_of_contexts}(\ref{lem:Te_of_contexts3}) we get a context $C_0$, a term $N'\in\lamu$, a resource context $c_0\in\Te{C_0}$, a rigid $c_0^{\,\bullet}$ of $c_0$ and resource terms $\vec{s}\in\Te{N'}$ s.t.\ $M=C_0\hole{N'}$, $h\in\Te{N'}$ and $c_0\hole{u}=c_0^{\,\bullet}\hole{u,\vec{s}}$ for all $u\in\lamu^{\mathrm{r}}$.
Now we can apply Proposition~\ref{lamu-prop:FIRSTPART} to $h\in\Te{N'}$ obtaining an $N''\in\lamu$ s.t.\ $N'\rightarrow N''$ and $\mathcal{H}\subseteq\Te{N''}$.
Set $N:=C_0\hole{N''}$, so that $M=C_0\hole{N'}\rightarrow N$.
Now:
every $t_i\in\Te{M}=\Te{C_0\hole{N'}}$, so by Lemma \ref{lm:Te_of_contexts}(\ref{lem:Te_of_contexts1}) it must have shape $t_i=c_i^{\,\bullet}\hole{\vec{v}_i}$ for some resource terms $v_{ij}\in\Te{N'}$, a context $c_i\in\Te{C_0}$ and a rigid $c_i^{\,\bullet}$ of $c_i$. But since $N'\rightarrow N''$ we can apply Proposition~\ref{prop:LiftAssSimulation}(1) on $v_{ij}$ and obtain that $v_{ij}\msto[r]\Sum{V}_{ij}$ for some sum $\Sum{V}_{ij}\subseteq\Te{N''}$. So $t_i\msto[r] c_i^{\,\bullet}\hole{\vec{\Sum{V}}_i}$. Let's call $\Sum{T}_i:=c_i^{\,\bullet}\hole{\vec{\Sum{V}}_i}$. Using again Lemma \ref{lm:Te_of_contexts}(\ref{lem:Te_of_contexts1}) one has that $\Sum{T}_i\subseteq\Te{N}$. Now, let's use again Proposition~\ref{prop:LiftAssSimulation}(1), this time on $s\in\Te{N'}$. Since $N'\rightarrow N''$ we obtain sums $\Sum{S}_i\subseteq{\Te{N''}}$ s.t.\ $s_i\msto[r]\Sum{S}_i$. So we have: $c\hole{\mathcal{H}}=c_0^{\,\bullet}\hole{\mathcal{H},\vec{s}}\msto[r] c_0^{\,\bullet}\hole{\mathcal{H},\vec{\Sum{S}}}=:\Sum{U}$. But since $\mathcal{H}\subseteq\Te{N''}$ and every $\Sum{S}_i\subseteq\Te{N''}$, again thanks to Lemma \ref{lm:Te_of_contexts}(\ref{lem:Te_of_contexts1}) one has $\Sum U\subseteq\Te{C_0\hole{N''}}=\Te{N}$. This ends the proof, since letting $\tilde{\Sum{T}}:=\sum\limits_i \Sum{T}_i+\Sum U\subseteq\Te{N}$ one has $\Sum{T}'\msto[r]\tilde{\Sum{T}}$.
\end{proof}

\subsection{APPENDIX - SECTION 4.1}

We give the proof of Theorem~\ref{lamu-th:TeStability}.
The proof is taken from~\cite{DBLP:journals/pacmpl/BarbarossaM20}.

\begin{proof}[Proof of Theorem~\ref{lamu-th:TeStability}]
Since every $\cX_i$ is $\mathcal{T}$-bounded, for $i=1,\dots,n$ there exists $L_i\in\Lam$ s.t.\ $\bigcup_{N\in\cX_i}\NFT{N}\subseteq\NFT{L_i}.$
Fix now $M_1,\dots,M_n\in\Lam$ s.t.\ $\NFT{M_i} = \bigcap\limits_{N\in\cX_i}\NFT{N}$. 
We have to show that:
\[
	\NFT{C\hole{M_1,\dots,M_n}} = \bigcap_{N_1\in\cX_1}\dots\bigcap_{N_n\in\cX_n}\NFT{C\hole{N_1,\dots,N_n}}.
\]

$(\subseteq)$.
 Clearly, for all $i=1,\dots,n$ and $N_i\in\cX_i$, we have $\NFT{M_i}\subseteq \NFT{N_i}$, therefore we conclude that: \[\NFT{C\hole{M_1,\dots,M_n}}\subseteq\NFT{C\hole{N_1,\dots,N_n}}\] by Monotonicity (Theorem~\ref{thm: Monotonicity}).

$(\supseteq)$.
 Let $t\in\bigcap\limits_{\vec N\in\vec \cX}\NFT{C\hole{N_{1},\dots,N_{n}}}$ (where we put $\vec N:=(N_1,\dots,N_n)$ and $\vec{\cX}:=(\cX_1,\dots,\cX_n)$).
 For every $\vec N\in\vec\cX$, by Lemma~\ref{lm:Te_of_contexts} there exist an $n$-resource-context $c_{\vec N}\in\Te{C}$ and, for every $i=1,\dots,n$, a list $\vec v^{\,i}_{\vec N} = \langle v^{i1}_{\vec N},\dots,v^{id_i}_{\vec N} \rangle$ (where $d_i:=\dg{\xi_i}{c_{\vec N}}$) of elements of $\Te{N_{i}}$ and such that $t\in\nf[\mathrm{r}]{\linc_{\vec N}\hole{\vec v^{\,1}_{\vec N},\dots,\vec v^{\,n}_{\vec N}}}$, for $\linc_{\vec{N}}$ a rigid of $c_{\vec{N}}$.
 Fix any reduction from $\linc_{\vec N}\hole{\vec v^{\,1}_{\vec N},\dots,\vec v^{\,n}_{\vec N}}$ to its normal form, and confluence allows to factorize it as follows:
 \[\begin{array}{l}
  \linc_{\vec N}\hole{\nf[\mathrm{r}]{v^{11}_{\vec N}},\dots,\nf[\mathrm{r}]{v^{1d_1}_{\vec N}},\dots,\nf[\mathrm{r}]{v^{n1}_{\vec N}},\dots,\nf[\mathrm{r}]{v^{nd_n}_{\vec N}}}
  \\
  \msto[\mathrm{r}]
  \\
  \nf[\mathrm{r}]{\linc_{\vec N}\hole{\vec v^{\,1}_{\vec N},\dots,\vec v^{\,n}_{\vec N}}}
  \ni t.
 \end{array}\]
 So for all $i=1,\dots,n$ and $j=1,\dots,d_i$, there exist $w^{ij}_{\vec N}\in\nf[\mathrm{r}]{v^{ij}_{\vec N}}$ such that:
 \begin{equation}\label{P1lam-eq:tindoppiavvu}
  \nf[\mathrm{r}]{\linc_{\seq N}\hole{\vec w^{\,1}_{\vec N},\dots,\vec w^{\,n}_{\vec N}}} \ni t
 \end{equation} 
 and being $N_i\in\cX_i$ which is bounded by $L_i$, we have $w^{\,ij}_{\vec N}\in\nf[\mathrm{r}]{v^{\,ij}_{\vec N}}\subseteq\NFT{N_{i}}\subseteq\NFT{L_i}$.
 From each inclusion $w^{\,ij}_{\vec N}\in\NFT{L_i}$ we obtain a resource term $u^{\,ij}_{\vec N}\in\Te{L_i}$ such that:
 \begin{equation}\label{P1lam-eq:crucialforfinite}
  w^{\,ij}_{\vec N}\in\nf[\mathrm{r}]{u^{\,ij}_{\vec N}}
 \end{equation}
 By composing thus a reduction from $u^{ij}_{\vec N}$ to $w^{ij}_{\vec N}$ with a reduction from $\linc_{\seq N}\hole{\vec w^{\,1}_{\vec N},\dots,\vec w^{\,n}_{\vec N}}$ to $t$, we find that $t$ belongs to $\nf[\mathrm{r}]{\linc_{\vec N}\hole{\vec u^{\,1}_{\vec N},\dots,\vec u^{n}_{\vec N}}}$. 
 This holds for all $\vec N\in\vec\cX$, i.e.:
 \begin{equation}\label{P1lam-eq:tininterparagraphu}
  t\in\bigcap_{\vec N\in\vec\cX} \nf[\mathrm{r}]{c_{\vec N}\hole{\vec u^{\,1}_{\vec N},\dots,\vec u^{n}_{\vec N}}}.
 \end{equation}
 Now, Lemma~\ref{lm:Te_of_contexts} gives $\linc_{\vec N}\hole{\vec u^{\,1}_{\vec N},\dots,\vec u^{n}_{\vec N}}\in\Te{C\hole{L_1,\dots,L_n}}$.
 But since the $L_i$'s are independent from $N_1,\dots,N_n$, and thanks to \eqref{P1lam-eq:tininterparagraphu}, we can apply Lemma~\ref{thm:Injectivity}, and obtain that the set 
 $\set{\linc_{\vec N}\hole{\vec u^{\,1}_{\vec N},\dots,\vec u^{n}_{\vec N}}
 \stt 
 \vec{N} \in \vec{\cX}
 }$ 
 is actually a singleton.
 Therefore, Lemma~\ref{lm:Te_of_contexts}(2) tells us that also the terms $c_{\vec N}$ and the bags $[\vec u_{\vec N}^{\,i}]$ are independent from $\vec N\in\vec \cX$.
 The unique element of the previous singleton has hence shape $\linc\hole{\vec u^{\,i},\dots,\vec u^{n}}$, with $\linc$ a rigid of a $c\in\Te{C}$, and $\vec{u}^{\,i}$ a list of elements of $\Te{L_i}$.
 Recalling now that $\sum\limits_j u^{i}_j \subseteq\Te{L_i}$, we can apply Corollary~\ref{cor:Ass3+Prop+Cor} in order to obtain, for each $i=1,\dots,n$, an $L'_i\in\Lam$ s.t.\ $L_i\msto L'_i$ and, using \eqref{P1lam-eq:crucialforfinite}, $\sum\limits_j w^{\,ij} \subseteq\Te{L'_i}$.
 Thus Lemma~\ref{lm:Te_of_contexts} tells us that, for every $\vec N\in\vec\cX$, we have:
 \begin{equation}\label{P1lam-eq:TaylorLprimes}
  \linc\hole{\vec w^{\,1}_{\vec N},\dots,\vec w^{n}_{\vec N}}\in\Te{C\hole{L'_1,\dots,L'_n}}.
 \end{equation}
 But now thanks to \eqref{P1lam-eq:TaylorLprimes} and \eqref{P1lam-eq:tindoppiavvu} (which holds for all $\vec N\in\vec\cX$), we can apply again Lemma~\ref{thm:Injectivity} in order to find that the set $\set{\linc\hole{\vec w^{\,1}_{\vec N},\dots,\vec w^{n}_{\vec N}}\stt \vec{N}\in\vec{\cX}}$ 
 is a singleton.
 Again by Lemma~\ref{lm:Te_of_contexts}(2), we have that all the bags $[ \vec w^{\,1}_{\vec N}],\dots,[\vec w^{n}_{\vec N} ]$ for $\vec{N}\in\vec{\cX}$, coincide respectively to some bags $[ \vec w^{\,1}],\dots,[\vec w^{n} ]$ which are independent from $\vec N\in\vec\cX$.
 So the only element of the previous singleton has shape $\linc\hole{\vec w^{\,1},\dots,\vec w^{n}}$, and \eqref{P1lam-eq:tindoppiavvu} becomes:
 \begin{equation}\label{P1lam-eq:tindoppiavvubis}
  t\in\nf[\mathrm{r}]{\linc\hole{\vec w^{\,1},\dots,\vec w^{n}}}.
 \end{equation}
 Now for all $i=1,\dots,n$, we already know that $[\vec{w}^{\,i}]=[\vec{w}^{\,i}_{\vec N}]$ which is a bag of elements of $\NFT{N}$, and this holds for all $N\in\cX_i$.
 That is, we have:
 \begin{equation}\label{P1lam-eq:endstability}
  \sum\limits_j \vec{w}^{\,ij}\subseteq\bigcap_{N\in\cX_i} \NFT{N} =\NFT{M_i}
 \end{equation}
 where we finally used the hypothesis.
 From \eqref{P1lam-eq:tindoppiavvubis}, \eqref{P1lam-eq:endstability} and Lemma~\ref{lm:Te_of_contexts} we finally conclude that $t\in\nf[\mathrm{r}]{\linc\hole{\vec w^1,\dots,\vec w^n}}\subseteq\NFT{C\hole{M_1,\dots,M_n}}$.\qedhere
\end{proof}

\begin{figure*}[t!]
 \centering

\[\small{
\begin{tikzcd}
& (\mu\alpha.\name{\beta}{\mu\gamma.\name{\eta}{s'}})[\vec{u}] \ar["(\alpha\neq\beta)"',dl] \ar[dr, "\big(\alpha\neq\beta" near start, "\dg{\gamma}{[\vec{u}]}=0\big)" near end] & & \\
\mu\alpha. \name{\beta}{\mu\gamma.\langle\name{\eta}{s'}\rangle^+_\alpha[\vec{u}]}
 \ar[->>, dr] & & 
(\mu\alpha.\name{\eta}{s'})[\vec{u}]\set{\beta/\gamma}
 \ar[->, dl] \\
& \mu\alpha.\langle\name{\eta}{s'}\rangle^+_\alpha[\vec{u}]\set{\beta/\gamma} & &
\end{tikzcd}
}\]
 \caption{Notable diagrams of Proposition~\ref{lamu-prop:lamu^+Confl}, point (1), subcase $t=(\mu\alpha.\name{\beta}{s})[\vec{u}]$, $\Sum T = \mu\alpha.\langle \name{\beta}{s} \rangle^+_\alpha[\vec{u}]$, $s=\mu\gamma.\name{\eta}{s'}$, $\alpha\neq\beta$.}
 \label{fig:6Diag.3}
\end{figure*}
\begin{figure*}[t!]
 \centering

\begin{flushleft}\[
\footnotesize{
\begin{tikzcd}[column sep=tiny]
& (\mu\alpha.\name{\alpha}{\mu\gamma.\name{\alpha}{s'}})[\vec{u}] \ar[dl] \ar[dr] & & 
\\
\mu\alpha.\coefflnamedapp{\name{\alpha}{\mu\gamma.\name{\alpha}{s'}}}{\alpha}{[\vec{u}]}
 \ar[equal, d, "\small{W:(\vec{u})\to\set{1,2}}" near start, "\small{D:(\vec{w}^{\,0})\to\set{1,2}}" near end] & & 
(\mu\alpha.\name{\alpha}{s'\set{\alpha/\gamma}})[\vec{u}]
 \ar[d] 
\\
\sum\limits_{W}
\sum\limits_{D}
\mu\alpha.\name{\alpha}{(\mu\gamma.\name{\alpha}{(\coefflnamedapp{s'}{\alpha}{[\vec{d}^{\,0}]})[\vec{d}^{\,1}]})[\vec{w}^{\,1}]}
 \ar[->>, d] & & 
\mu\alpha.\coefflnamedapp{\name{\alpha}{s'\set{\alpha/\gamma}}}{\alpha}{[\vec{u}]}
\ar[equal, d, "\small{W:(\vec{u})\to\set{1,2}}"] 
\\
\sum\limits_{W}
\sum\limits_{D}
\mu\alpha.\name{\alpha}{\mu\gamma.\name{\alpha}{\coefflnamedapp{(\coefflnamedapp{s'}{\alpha}{[\vec{d}^{\,0}]})[\vec{d}^{\,1}]}{\gamma}{[\vec{w}^{\,1}]}}}
 \ar[->>, dr] & & 
\sum\limits_{W}
\mu\alpha.\name{\alpha}{(\coefflnamedapp{s'\set{\alpha/\gamma}}{\alpha}{[\vec{w}^{\,0}]})[\vec{w}^{\,1}]}
 \ar[equal, dl] 
\\
&
\sum\limits_{W}
\sum\limits_{D}
\mu\alpha.\name{\alpha}{\coefflnamedapp{(\coefflnamedapp{s'}{\alpha}{[\vec{d}^{\,0}]})[\vec{d}^{\,1}]}{\gamma}{[\vec{w}^{\,1}]}}\set{\alpha/\gamma}
&
\end{tikzcd}
}
\]\end{flushleft}
 \caption{Notable diagrams of Proposition~\ref{lamu-prop:lamu^+Confl}, point (1), subcase $t=(\mu\alpha.\name{\beta}{s})[\vec{u}]$, $\Sum T = \mu\alpha.\langle \name{\beta}{s} \rangle^+_\alpha[\vec{u}]$, $s=\mu\gamma.\name{\eta}{s'}$,$\alpha=\beta,\gamma\neq\eta,\eta=\alpha$.}
 \label{fig:6Diag.4}
\end{figure*}
\begin{figure*}[t!]
 \centering

\begin{flushleft}\[
\footnotesize{
\begin{tikzcd}[column sep=tiny]
& (\mu\alpha.\name{\alpha}{\mu\gamma.\name{\eta}{s'}})[\vec{u}] \ar[dl] \ar[dr, "(\gamma\neq\eta)"] & & 
\\
\mu\alpha.\coefflnamedapp{\name{\alpha}{\mu\gamma.\name{\eta}{s'}}}{\alpha}{[\vec{u}]}
 \ar[equal, d, "\small{W:(\vec{u})\to\set{1,2}}" ] & & 
(\mu\alpha.\name{\eta}{s'\set{\alpha/\gamma}})[\vec{u}]
 \ar[d] 
\\
\sum\limits_{W}
\mu\alpha.\name{\alpha}{(\mu\gamma.\name{\eta}{\coefflnamedapp{s'}{\alpha}{[\vec{w}^{\,0}]}})[\vec{w}^{\,1}]}
 \ar[->>, d] & & 
\mu\alpha.\coefflnamedapp{\name{\eta}{s'\set{\alpha/\gamma}}}{\alpha}{[\vec{u}]}
\ar[equal, d, "(\eta\neq\alpha)"] 
\\
\sum\limits_{W}
\mu\alpha.\name{\alpha}{\mu\gamma.\name{\eta}{\coefflnamedapp{\coefflnamedapp{s'}{\alpha}{[\vec{w}^{\,0}]}}{\gamma}{[\vec{w}^{\,1}]}}}
 \ar[->>, dr] & & 
\mu\alpha.\name{\eta}{\coefflnamedapp{s'\set{\alpha/\gamma}}{\alpha}{[\vec{u}]}}
 \ar[equal, dl] 
\\
&
\sum\limits_{W}
\mu\alpha.\name{\eta}{\coefflnamedapp{\coefflnamedapp{s'}{\alpha}{[\vec{w}^{\,0}]}}{\gamma}{[\vec{w}^{\,1}]}\set{\alpha/\gamma}}
&
\end{tikzcd}
}
\]\end{flushleft}
 \caption{Notable diagrams of Proposition~\ref{lamu-prop:lamu^+Confl}, point (1), subcase $t=(\mu\alpha.\name{\beta}{s})[\vec{u}]$, $\Sum T = \mu\alpha.\langle \name{\beta}{s} \rangle^+_\alpha[\vec{u}]$, $s=\mu\gamma.\name{\eta}{s'}$,$\alpha=\beta,\gamma\neq\eta,\eta\neq\alpha$.}
 \label{fig:6Diag.5}
\end{figure*}
\begin{figure*}[t!]
 \centering

\begin{flushleft}\[
\footnotesize{
\begin{tikzcd}[column sep=tiny]
& (\mu\alpha.\name{\alpha}{\mu\gamma.\name{\gamma}{s'}})[\vec{u}] \ar[dl] \ar[dr] & & 
\\
\mu\alpha.\coefflnamedapp{\name{\alpha}{\mu\gamma.\name{\gamma}{s'}}}{\alpha}{[\vec{u}]}
 \ar[equal, d, "\small{W:(\vec{u})\to\set{1,2}}"] & & 
(\mu\alpha.\name{\alpha}{s'\set{\alpha/\gamma}})[\vec{u}]
 \ar[d] 
\\
\sum\limits_{W}
\mu\alpha.\name{\alpha}{(\mu\gamma.\name{\gamma}{\coefflnamedapp{s'}{\alpha}{[\vec{w}^{\,0}]}})[\vec{w}^{\,1}]}
 \ar[->>, d, "\small{D:(\vec{w}^{\,1})\to\set{1,2}}"] & & 
\mu\alpha.\coefflnamedapp{\name{\alpha}{s'\set{\alpha/\gamma}}}{\alpha}{[\vec{u}]}
\ar[equal, d, "\small{W:(\vec{u})\to\set{1,2}}"] 
\\
\sum\limits_{W}
\sum\limits_{D}
\mu\alpha.\name{\alpha}{\mu\gamma.\name{\gamma}{(\coefflnamedapp{\coefflnamedapp{s'}{\alpha}{[\vec{w}^{\,0}]}}{\gamma}{[\vec{d}^{\,0}]})[\vec{d}^{\,1}]}}
 \ar[->>, dr] & & 
\sum\limits_{W}
\mu\alpha.\name{\alpha}{(\coefflnamedapp{s'\set{\alpha/\gamma}}{\alpha}{[\vec{w}^{\,0}]})[\vec{w}^{\,1}]}
 \ar[equal, dl] 
\\
&
\sum\limits_{W}
\sum\limits_{D}
\mu\alpha.\name{\alpha}{(\coefflnamedapp{\coefflnamedapp{s'}{\alpha}{[\vec{w}^{\,0}]}}{\gamma}{[\vec{d}^{\,0}]})[\vec{d}^{\,1}]\set{\alpha/\gamma}}
&
\end{tikzcd}
}
\]\end{flushleft}
 \caption{Notable diagrams of Proposition~\ref{lamu-prop:lamu^+Confl}, point (1), subcase $t=(\mu\alpha.\name{\beta}{s})[\vec{u}]$, $\Sum T = \mu\alpha.\langle \name{\beta}{s} \rangle^+_\alpha[\vec{u}]$, $s=\mu\gamma.\name{\eta}{s'}$,$\alpha=\beta,\eta=\gamma$.}
 \label{fig:6Diag.6}
\end{figure*}
\begin{figure*}[t!]
 \centering
\[\small{
\begin{tikzcd}\centering
& \mu\gamma.\name{\alpha}{\mu\beta.\name{\eta}{\mu\gamma'.\name{\eta'}{s'}}} \ar[dl] \ar[dr] & & \\
\mu\gamma.\name{\delta_0}{\mu\gamma'.\name{\delta'_1}{s'\set{\alpha/\beta}}} \ar[->, dr] & & \mu\gamma.\name{\alpha}{\mu\beta.\name{\delta'_2}{s'\set{\eta/\gamma'}}} \ar[->, dl] \\
& \!\!\!\!\!\!\!\!\!\mu\gamma.\name{\delta_1}{s'\set{\alpha/\beta}\set{\delta_0/\gamma'}}
= 
\mu\gamma.\name{\delta_2}{s'\set{\eta/\gamma'}\set{\alpha/\beta}} & &
\end{tikzcd}
}\]
where $\delta_0:=\delta^\alpha_\eta(\beta)$, $\delta'_1:=\delta^\alpha_{\eta'}(\beta)$, $\delta_1:=\delta^{\delta_0}_{\delta'_1}(\gamma')$, $\delta'_2:=\delta^{\eta}_{\eta'}(\gamma')$ and $\delta_2:=\delta^{\alpha}_{\delta'_2}(\beta)$.

 \caption{Notable diagrams of Proposition~\ref{lamu-prop:lamu^+Confl}, point (1), subcase $t=\mu\gamma.\name{\alpha}{\mu\beta.\name{\eta}{s}}$, $\Sum T = \mu\gamma.\name{\eta}{s}\set{\alpha/\beta}$, $s=\mu\gamma'.\name{\eta'}{s'}$.}
 \label{fig:twoDiag}
\end{figure*}

\end{document}